\documentclass[a4paper,11pt]{article}
\usepackage[margin=1in]{geometry}
\usepackage{authblk}
\usepackage[utf8]{inputenc}
\usepackage[T1]{fontenc}
\usepackage{microtype}
\usepackage{amsfonts}
\usepackage{amssymb}
\usepackage{amsmath}
\usepackage{amsthm}
\usepackage{comment}
\usepackage{fullpage}
\usepackage{enumitem}
\usepackage[hidelinks,bookmarksopen=true]{hyperref}
\usepackage[capitalize]{cleveref}
\usepackage{mathtools}
\usepackage{tikz}
\usepackage{wrapfig}
\usepackage[font=small,labelfont=bf]{caption}

\setlist[enumerate]{nosep, topsep=1ex}
\setlist[itemize]{nosep, topsep=1ex}
\setlist[description]{nosep}
\allowdisplaybreaks

\newtheorem{theorem}{Theorem}[section]
\newtheorem{corollary}[theorem]{Corollary}
\newtheorem{lemma}[theorem]{Lemma}
\newtheorem{fact}[theorem]{Fact}
\newtheorem{observation}[theorem]{Observation}
\newtheorem{definition}[theorem]{Definition}
\newtheorem{proposition}[theorem]{Proposition}

\newcommand{\LCE}{\text{\rm LCE}}
\newcommand{\lcp}{\text{\rm lcp}}
\newcommand{\lpos}{\text{\rm lpos}}
\newcommand{\rpos}{\revstr{\text{\rm lpos}}}
\newcommand{\revstr}[1]{\overline{#1}}
\newcommand{\per}{\text{\rm per}}
\DeclareMathOperator{\Occ}{Occ}
\newcommand{\hook}{\mathsf{hook}}
\newcommand{\anch}{\mathsf{anch}}
\newcommand{\rend}[1]{\ensuremath{\alpha_{#1}}}
\newcommand{\bigO}{\mathcal{O}}
\renewcommand{\Pr}{\mathbb{P}}
\newcommand{\Oh}{\bigO}
\DeclareMathOperator{\polylog}{polylog}
\newcommand{\size}{\mathrm{size}}

\newcommand{\T}{T}
\renewcommand{\S}{\mathsf{S}}
\newcommand{\R}{\mathsf{R}}
\newcommand{\LCP}{\text{\rm LCP}}
\newcommand{\SA}{\text{\rm SA}}
\newcommand{\BWT}{\text{\rm BWT}}

\newcommand{\Su}{\mathcal{S}}

\newcommand{\sub}{\subseteq}
\newcommand{\dd}{\mathinner{.\,.}}
\newcommand{\eps}{\varepsilon}

\let\OLDthebibliography\thebibliography
  \renewcommand\thebibliography[1]{
  \OLDthebibliography{#1}
  \setlength{\parskip}{0pt}
  \setlength{\itemsep}{0pt plus 0.3ex}
}

\begin{document}

\renewcommand\Affilfont{\normalsize}

\title{Resolution of the Burrows-Wheeler Transform Conjecture}

\author[1]{Dominik Kempa\thanks{Supported by
  NSF grants no. 1652303 and 1934846, and an Alfred
  P. Sloan Fellowship grant.}}
\author[2,1]{Tomasz Kociumaka\thanks{Supported by ISF
  grants no. 1278/16 and 1926/19, by a BSF grant no.
  2018364, and by an ERC grant MPM under the EU's Horizon
  2020 Research and Innovation Programme (agreement no.
  683064).}}
\affil[1]{University of California, Berkeley, USA}
\affil[ ]{\href{mailto:kempa@berkeley.edu}
  {\nolinkurl{kempa@berkeley.edu}}, \href{mailto:kociumaka@berkeley.edu}
  {\nolinkurl{kociumaka@berkeley.edu}}}
\affil[ ]{\vspace{-1ex}}
\affil[2]{Bar-Ilan University, Ramat Gan, Israel}
\date{\vspace{-1.5cm}}
\maketitle

\begin{abstract}
  The Burrows--Wheeler Transform (BWT) is an invertible text
  transformation that permutes symbols of a text according to the
  lexicographical order of its suffixes. BWT is the main component of
  popular lossless compression programs (such as {\tt bzip2}) as well
  as recent powerful compressed indexes (such as $r$-index [Gagie et
    al., J.\ ACM, 2020]), central in modern bioinformatics. The
  compression ratio of BWT is quantified by the number $r$ of
  equal-letter runs. Despite the practical significance of BWT, no
  non-trivial bound on the value of $r$ is known.  This is in contrast
  to nearly all other known compression methods, whose sizes have been
  shown to be either always within a $\polylog n$ factor (where $n$ is
  the length of text) from $z$, the size of Lempel--Ziv (LZ77) parsing
  of the text, or significantly larger in the worst case (by a
  $n^{\eps}$ factor for $\eps > 0$).

  In this paper, we show that $r = \bigO(z \log^2n)$ holds for every
  text.  This result has numerous implications for text indexing and
  data compression; for example: (1) it proves that many results
  related to BWT automatically apply to methods based on LZ77, e.g.,
  it is possible to obtain functionality of the suffix tree in
  $\bigO(z \polylog n)$ space; (2) it shows that many text processing
  tasks can be solved in the optimal time assuming the text is
  compressible using LZ77 by a sufficiently large $\polylog n$ factor;
  (3) it implies the first non-trivial relation between the number of
  runs in the BWT of the text and its reverse.

  In addition, we provide an $\bigO(z \polylog n)$-time algorithm
  converting the LZ77 parsing into the run-length compressed BWT.  To
  achieve this, we develop a number of new data structures and
  techniques of independent interest. In particular, we introduce a
  notion of compressed string synchronizing sets (generalizing the
  recently introduced powerful technique of string synchronizing sets
  [STOC 2019]) and show how to efficiently construct them.  Next, we
  propose a new variant of wavelet trees for sequences of long
  strings, establish a non-trivial bound on their size, and describe
  efficient construction algorithms.  Finally, we describe new indexes
  that can be constructed directly from the LZ77-compressed text and
  efficiently support pattern matching queries on substrings of the
  text.
\end{abstract}

\setcounter{page}{0}
\thispagestyle{empty}
\newpage

\section{Introduction}\label{sec:intro}

Lossless data compression aims to exploit redundancy in the input data
to represent it in~a~small space. Despite the abundance of compression
methods, nearly every existing tool falls into one of the few general
frameworks, among which the three most popular are: Lempel--Ziv
compression (where the nominal and most commonly used is the LZ77
variant~\cite{LZ77}), statistical compression (this includes, for
example, context mixing~\cite{CM}, prediction by partial matching
(PPM)~\cite{ClearyW84}, and dynamic Markov coding~\cite{CormackH87}),
and Burrows--Wheeler transform (BWT)~\cite{BWT}. As seen in the Large
Text Compression Benchmark~\cite{ltcb}, these three frameworks
underlie most~existing~compressors.

One of the features that best differentiates these algorithms is
whether they better remove the redundancy caused by skewed symbols
frequencies or by repeated fragments. The idea in LZ77 (which
underlies, for example, {\tt 7-zip}~\cite{7zip} and {\tt
gzip}~\cite{gzip} compressors) is to partition the input text into
long substrings, each having an earlier occurrence in the text. Every
substring is then encoded as a pointer to the previous occurrence
using a pair of integers. This method natively handles long repeated
substrings and can achieve an exponential compression ratio given
sufficiently repetitive input.  Statistical compressors, on the other
hand, are based on representing (predicting) symbols in the input
based on their frequencies.  This is formally captured by the notion
of the \emph{$k$th order empirical entropy} $H_k(T)$~\cite{entropy}.
For any sufficiently long text $T$, symbol frequencies (taking context
into account) in any power of $T$ (the concatenation of several copies
of~$T$) do not change significantly~\cite[Lemma
2.6]{kreft2013compressing}. Therefore, $|T^t|H_k(T^t) \approx t\cdot
|T|H_k(T)$ for any $t>1$, i.e., entropy is not sensitive to long
repetitions, and hence statistical compressors are outperformed by
LZ77, when the goal is to capture long
repetitions~\cite{gagie2006large, Gagie2020, attractors,
  kreft2013compressing, SirenVMN08}.

The above analysis raises the question about the nature of
compressibility of the Burrows--Wheeler transform. The compression of
BWT-based compressors, such as {\tt bzip2}~\cite{bzip2}, is \mbox{quantified}
by the number $r$ of equal-letter blocks in the BWT. The clear picture
described above no longer applies to the measure $r$. On one hand,
Manzini~\cite{Manzini01} proved that $r$ can be upper-bounded in terms
of the $k$th order empirical entropy of the input string.  On the
other hand, already in 2008, Sir{\'{e}}n et al.~\cite{SirenVMN08}
observed that BWT achieves excellent compression (superior to
statistical methods) on highly repetitive collections and provided
probabilistic analysis exhibiting cases when $r$ is small. Yet, after
more than a decade, no upper bound on $r$ in terms of $z$ was
discovered.

This lack of understanding is particularly frustrating due to numerous
applications of BWT in the field of bioinformatics and compressed
computation. One of the most successful applications of BWT is in
\emph{compressed indexing}, which aims to store a compressed string,
simultaneously supporting various queries (such as random access,
pattern matching, or even suffix array queries) on the uncompressed
version. While classical (uncompressed) indexes, such as suffix
trees~\cite{Weiner73} and suffix arrays~\cite{mm1993}, have been
successful in many applications, they are not suitable for storing and
searching big highly repetitive databases.  Such datasets are
virtually impossible to search without preprocessing: Github
databases, for example, take more than 20 terabytes, and the recently
finished 100000 Human Genome Project~\cite{100k} produced 75 terabytes
of DNA~\cite{Navarro2020}.  These databases are, however, highly
compressible: Github averages 20 versions per
project~\cite{Navarro2020}, and two human genomes are $99.9\%$
similar~\cite{Przeworski2000}. This area has witnessed a remarkable
surge of interest in recent years~\cite{ArroyueloNS12,
  DCC2015,BilleEGV18,BLRSRW15, Christiansen2019, GagieGKNP12,
  RLCST,Kempa19, navarro201941, Prezza19, SinhaW19}. BWT-based
indexes, such as $r$-index~\cite{Gagie2020}, are among the most
powerful~\cite{RLCST}, and their space usage is up to $\bigO(\polylog
n)$ factors away from the value $r$.  For a comprehensive overview of
this field, we refer the reader
to~a~survey~by~Navarro~\cite{Navarro2020}.

In addition to text indexing, BWT has many applications in compressed
computation. For example, BWT is the main component of the popular
read aligners such as ${\tt Bowtie}$~\cite{bowtie}, ${\tt
BWA}$~\cite{bwa}, and ${\tt Soap2}$~\cite{soap2}. Modern textbooks
spend dozens of pages describing applications of BWT~\cite{bwtbook,
  MBCT2015, navarrobook, ohl2013}. The richness of these applications
has even spawned a dedicated seminar~\cite{dagstuhl}.  Given the
importance and practical significance of BWT, one of the biggest open
problems that emerged in the field of lossless data compression and
compressed computation asks:

\vspace{1.5ex}
\begin{center}
  \textit{What is the upper bound on the output size of the
    Burrows--Wheeler transform?}
\end{center}
\vspace{1.5ex}

With the exception of BWT, essentially every other known compression
method has been proven~\cite{GNPlatin18, attractors} to produce output
whose size is always within an $\Oh(\polylog n)$ factor from~$z$, the
output size of the LZ77 algorithm (e.g., grammar
compression~\cite{charikar,Rytter03}, collage
systems~\cite{KidaMSTSA03}, and macro schemes~\cite{storer1978macro}),
or larger by a polynomial factor ($n^{\eps}$ for some $\eps>0$) in the
worst case (e.g., LZ78~\cite{LZ78}, compressed word graphs
(CDAWGs)~\cite{blumer1987complete}).\footnote{The choice for LZ77 as a
representative in this class follows from the fact that most of the
other methods are NP-hard to optimize~\cite{charikar,
gallant1982string}, while LZ77 admits a simple linear-time
compression algorithm (see, e.g.,~\cite{kkp-jea}).}  BWT is known to
never compress much better than LZ77, i.e., $z = \bigO(r \log
n)$~\cite{GNPlatin18}. The opposite relation (that $r = \bigO(z
\polylog n)$) was generally conjectured to be false. For example,
after presenting how to support suffix array and suffix tree queries
in $\bigO(r \polylog n)$ space, Gagie et al.~\cite{RLCST} speculate
that ``\emph{(...) it seems unlikely that one can provide suffix array
or tree functionality within space related to $g$, $z$, or $\gamma$,
since these measures are not related to the structure of the suffix
array: this is likely to be a specific advantage of measure $r$}''.

\paragraph{Our Contribution}

We prove that $r=\bigO(z \log^2 n)$ holds for all strings, resolving
the BWT conjecture in the more surprising way than anticipated, and
solving an open problem by Prezza~\cite{prezzaopenproblem} and Gagie
et al.~\cite{RLCST,GNPlatin18}. This result alone has multiple
implications for indexing and \mbox{compression:}
\begin{enumerate}
\item It is possible to support suffix array and suffix tree
  functionality in $\bigO(z \polylog n)$ space~\cite{Gagie2020}.
\item It was shown in~\cite{Kempa19} that many string processing tasks
  (including BWT and LZ77 construction) can be solved in
  $\bigO(n/\log_{\sigma}n+r \polylog n)$ time (where $\sigma$ is the
  alphabet size), i.e., if the text is sufficiently compressible by
  BWT (formally, when $n/r=\Omega(\polylog n)$), these tasks can be
  solved in optimal time (which is unlikely to be possible for general
  texts~\cite{sss}).  Our result loosens this assumption to
  $n/z=\Omega(\polylog n)$.
\item Until now, methods based on the Burrows--Wheeler transform were
  thought to be neither statistical nor dictionary (LZ-like)
  compression algorithms~\cite{RLCST,SirenVMN08}. Our result
  challenges the notion that the BWT forms its own compression type: In
  view of our bound, BWT is much closer to LZ compressors than was
  previously thought.
\end{enumerate}
Our slightly stronger bound $r=\bigO(\delta \log^2 n)$, where
$\delta\leq z$ is a symmetric (insensitive to string reversal)
repetitiveness measure recently introduced in~\cite{delta}, further
shows that:
\begin{enumerate}
\setcounter{enumi}{3} 
\item The number $\bar{r}$ of BWT runs in the reverse of the text
  satisfies $\bar{r}=\bigO(r \log^2 n)$, which is the first
  non-trivial bound in terms of $r$.  This result is of practical
  importance due to many algorithms whose efficiency depends on
  $\bar{r}$~\cite{BannaiGI20,
      DjamalCGPR15,DjamalCGPR17,OhlebuschBA14, OhnoSTIS18,
      PolicritiP17,PolicritiP18}.
\end{enumerate}

After proving $r = \bigO(z \log^2 n)$ and $r = \bigO(\delta \log^2
n)$, by a tighter analysis, we obtain $r = \bigO(z \log z \max(1,
\log\frac{n}{z \log z}))$ and $r = \bigO(\delta \log \delta \max(1,
\frac{n}{\delta \log \delta}))$, and we show that the latter is
asymptotically tight for the full spectrum of values of $n$ and
$\delta$. As a side-result, we obtain a tight upper bound $\bigO(n
\log \delta)$ on the sum of irreducible LCP values.  This improves
upon the previously known bound $\bigO(n \log
r)$~\cite{KarkkainenKP16}.

We then describe an $\bigO(z \log^8 n)$-time algorithm converting the
LZ77 parsing into run-length compressed BWT (the $\polylog n$ factor
has not been optimized).  This offers up to exponential speedup over
the previously fastest space-efficient algorithms, which need
$\Omega(n \log z)$ time~\cite{OhnoSTIS18,PolicritiP17}.  To achieve
this, we develop new data structures and techniques of independent
interest. In particular, we introduce a notion of \emph{compressed
string synchronizing sets}, generalizing the powerful technique
introduced in~\cite{sss}. We also describe a new variant of wavelet
\mbox{trees~\cite{wt}, designed to} work for sequences of long
strings. Finally, we describe new indexes that can be
built directly from the LZ77-compressed text and support fast pattern
matching queries on~text~substrings.

\paragraph{Organization of the Paper}
\cref{sec:prelim} introduces the basic notation. We present our upper
and lower bounds in \cref{sec:upper-bound,sec:lower-bound},
respectively.  Finally, in \cref{sec:algorithm}, we develop our
algorithm converting LZ77 to run-length compressed BWT, with the
description of our LZ77-based text indexes deferred to \cref{sec:aux}.

\section{Preliminaries}\label{sec:prelim}

For any string $S$, we write $S[i\dd j]$, where $1 \leq i,j \leq |S|$,
to denote a substring of $S$. If $i>j$, we assume $S[i\dd j]$ to be
the empty string $\eps$.  By $\revstr{S}$ we denote the \emph{reverse}
of $S$.

An integer $p\in [1\dd |S|]$ is a \emph{period} of a string $S$ if
$S[i]=S[i+p]$ holds for every $i\in [1\dd |S|-p]$.  The classic
\emph{periodicity lemma}~\cite{fine1965uniqueness} states that if a
string $S$ has periods $p,q$ such that $p+q-\gcd(p,q)\le |S|$, then
$\gcd(p,q)$ is also a period of $S$.

The shortest period of $S$ is denoted as $\per(S)$. A string $S$ is
called \emph{periodic} if $\per(S) \leq \frac12|S|$. The following
fact is a folklore consequence of the periodicity lemma.

\begin{fact}[see~{\cite[Fact 1]{Amir2017}}]\label{fct:disper}
  Any two distinct periodic strings of the same length differ on at
  least two positions.
\end{fact}

\begin{wrapfigure}{r}{0.46\textwidth}
  \vspace{-.5cm}
  \begin{tikzpicture}[yscale=0.4]
  \foreach \x [count=\i] in {, a, aababa, aababababaababa, aba, abaababa, abaababababaababa, ababa, ababaababa, abababaababa, ababababaababa, ba, baababa, baababababaababa, baba, babaababa, babaababababaababa, bababaababa, babababaababa,bbabaababababaababa}{
    \draw (1.9, -\i) node[right] {$\texttt{\x\$}\vphantom{\textbf{\underline{7}}}$};
  }
  \draw(1.9,0) node[right] {\scriptsize $\T[\SA[i]\dd n]$};

  \foreach \x [count=\i] in {a, b, b, b, b, b, b, a, b, b, a, a, a, a, a, a, b, a, a,\$}{
    \draw (1.5, -\i) node {$\texttt{\x\vphantom{\$}}\vphantom{\textbf{\underline{7}}}$};
    \draw (-.85, -\i) node {\footnotesize $\i$};
  }
  \draw(1.5,0) node{\scriptsize $\BWT[i]$};
  \draw(-.85,0) node{\scriptsize $i\vphantom{\BWT[]}$};

  \foreach \x [count=\i] in {20,19,14,5,17,12,3,15,10,8,6,18,13,4,16,11,2,9,7,1}{
    \draw (.6, -\i) node {$\x\vphantom{\textbf{\underline{7}}}$};
  }
  \draw(.6,0) node{\scriptsize $\SA[i]$};

  \foreach \x [count=\i] in {\textbf{\underline{0}},\textbf{\underline{0}},1,6,1,3,8,\textbf{\underline{3}},\textbf{\underline{5}},5,\textbf{\underline{7}},0,2,7,2,4,\textbf{\underline{9}},\textbf{\underline{4}},6,\textbf{\underline{1}}}{
    \draw (-.25, -\i) node {$\x\vphantom{\textbf{\underline{7}}}$};
  }
  \draw(-.25,0) node{\scriptsize $\LCP[i]$};
  \end{tikzpicture}

  \caption{A list of lexicographically sorted suffixes of the string
    $\T= \texttt{bbabaababababaababa\$}$ along with the $\BWT$, $\SA$,
    and $\LCP$ tables. The irreducible $\LCP$ values are bold and
    underlined.}\label{fig:example}
  \vspace{-.5cm}
\end{wrapfigure}

Throughout the paper, we consider a string (text) $\T[1\dd n]$ of $n
\geq 1$ symbols from an ordered alphabet $\Sigma$ of size $\sigma$. We
assume $\T[n]=\$$, where $\$\in\Sigma$ is the lexicographically
smallest symbol in $\Sigma$, and that $\$$ does not occur anywhere
else in $\T$.

The \emph{suffix array}~\cite{mm1993} of $\T$ is an array $\SA[1\dd
n]$ containing a permutation of the integers $[1\dd n]$ such that
$\T[\SA[1]\dd n] \prec \T[\SA[2]\dd n] \prec \cdots \prec \T[\SA[n]\dd
n]$, where $\prec$ denotes the lexicographic order. The closely
related \emph{Burrows--Wheeler transform}~\cite{BWT} $\BWT[1\dd n]$ of
$\T$ is defined by $\BWT[i] = \T[\SA[i]-1]$ if $\SA[i] > 1$ and
$\BWT[i] = \T[n]$ otherwise. The BWT is invertible: given $\BWT[1 \dd
n]$, the text $\T$ can be restored in $\bigO(n)$ time.

For any string $S = c_1^{\ell_1} c_2^{\ell_2} \cdots c_h^{\ell_h}$,
where $c_i \in \Sigma$ and $\ell_i>0$ for $i \in [1 \dd h]$, and $c_i
\neq c_{i+1}$ for $i \in [1 \dd h)$, we define the \emph{run-length
encoding} of $S$ as a sequence ${\rm RL}(S) = ((c_1, \lambda_1),
\ldots, (c_h, \lambda_h))$, where $\lambda_i = \ell_1 + \ldots +
\ell_i$ for $i \in [1 \dd h]$. Throughout, we let $r=|{\rm RL}(\BWT)|$
denote the number of runs in the BWT of $\T$. E.g., for the text
$\T=\texttt{bbabaababababaababa\$}$ in \cref{fig:example}, we have
$\BWT = \texttt{a}^1\texttt{b}^6
\texttt{a}^1\texttt{b}^2\texttt{a}^6\texttt{b}^1\texttt{a}^2\texttt{\$}^1$,
and hence $r=8$.

By $\lcp(S_1, S_2)$ we denote the length of the longest common prefix
of strings $S_1$ and $S_2$. For $j_1, j_2 \in [1 \dd n]$, we let
$\LCE(j_1,j_2) = \lcp(\T[j_1 \dd n], \T[j_2 \dd n])$.  The \emph{LCP
array} (see~\cite{mm1993,klaap2001}), $\LCP[1\dd n]$, is defined as
$\LCP[i] = \LCE(\SA[i],\SA[i-1])$ for $i \in [2\dd n]$ and
$\LCP[1]=0$.  We say that the value $\LCP[i]$ is \emph{reducible} if
$\BWT[i] = \BWT[i-1]$ and \emph{irreducible} otherwise (including
$i=1$). Note that there are exactly $r$ irreducible LCP values.

\begin{theorem}[K\"arkk\"ainen et al.~\cite{KarkkainenKP16}]\label{lm:irreduc-sum}
  The sum of all irreducible LCP values is at most $n \log r$.
\end{theorem}

We say that a fragment $\T[i\dd i+\ell)$ is a \emph{previous factor}
if it has an earlier occurrence in $\T$, i.e., $\LCE(i,i')\ge \ell$
holds for some $i'\in [1\dd i)$.  An \emph{LZ77-like factorization} of
$\T$ is a factorization $\T = F_1\cdots F_f$ into non-empty
\emph{phrases} such that each phrase $F_j$ with $|F_j|>1$ is a
previous factor.  In the underlying \emph{LZ77-like representation},
every phrase $F_j=\T[i\dd i+\ell)$ that is a previous factor is
encoded as $(i',\ell)$, where $i'\in [1\dd i)$ satisfies
$\LCE(i,i')\ge \ell$ (and is chosen arbitrarily in case of multiple
possibilities); if $F_j=T[i]$ is not a previous factor, it is encoded
as $(T[i],0)$.

The LZ77 factorization~\cite{LZ77} (or the LZ77 parsing) of a string
$\T$ is then just an LZ77-like factorization constructed by greedily
parsing $\T$ from left to right into longest possible phrases.  More
precisely, the $j$th phrase $F_j$ is the longest previous factor
starting at position $1+|F_1\cdots F_{j-1}|$; if no previous factor
starts there, then $F_j$ consists of a single character.  We denote
the number of phrases in the LZ77 parsing by $z$.  For example, the
text $\texttt{bbabaababababaababa\$}$ of \cref{fig:example} has LZ77
factorization $\texttt{b}\cdot \texttt{b}\cdot \texttt{a}\cdot
\texttt{ba} \cdot \texttt{aba}\cdot \texttt{bababa} \cdot
\texttt{ababa} \cdot \texttt{\$}$ with $z=8$ phrases, and its LZ77
representation is $ (\texttt{b},0), (1,1), (\texttt{a},0), (2,2),
(3,3), (7,6), (10,5), (\texttt{\$},0).  $

The following relation between $z$ and $r$ is known.

\begin{theorem}[Gagie et al.~\cite{GNPlatin18}]
  Every string of length $n$ satisfies $z = \bigO(r \log n)$.
\end{theorem}

\section{Upper Bounds}\label{sec:upper-bound}

\subsection{Basic Upper Bound}\label{sec:basic-upper-bound}

To illustrate the main idea of our proof technique, we first prove the
upper bound in its simplest form $r=\bigO(z \log^2 n)$. The following
lemma stands at the heart of our proof.

\begin{lemma}\label{lm:irreduc-sum-2}
  For every $\ell\in [1\dd n]$, the number of irreducible LCP values
  in $[\ell\dd 2\ell)$ is $\bigO(z \log n)$.
\end{lemma}

  \noindent
  \textit{Proof.}
  Let $\T^{\infty}$ be an infinite string defined so that
  $\T^{\infty}[i]=\T[1+(i-1)\bmod n]$ for $i \in \mathbb{Z}$; in
  particular, $\T^{\infty}[1\dd n]=\T[1\dd n]$. Due to $T[n]=\$$, we
  have $\T^\infty[\SA[1]\dd ]\prec \cdots \prec \T^\infty [\SA[n]\dd
  ]$ and $\BWT[i]=\T^\infty[\SA[i]-1]$ for $i\in [1\dd n]$.
  
  Denote $\Su_m =\{S\in\Sigma^m : S\text{ is a substring of
  }\T^{\infty}\}$ for $m \geq 1$.  Observe that $|\Su_{m}| \leq mz$
  since every length-$m$ substring of $\T^{\infty}$ has an occurrence
  crossing or beginning at a phrase boundary of the LZ77 parsing of
  $\T$.  This includes substrings overlapping two copies of $\T$,
  which cross the boundary between the last and the first phrase.
  
  The idea of the proof is as follows.  With each irreducible value
  $\LCP[i]\in [\ell\dd 2\ell)$, we associate a cost of $\ell$ units,
  which are charged to individual characters of strings in
  $\Su_{3\ell}$.  We then show that each of the strings in
  $\Su_{3\ell}$ is charged at most $2\log n$ times. The number of
  irreducible LCP values in $[\ell \dd 2\ell)$ equals $\frac{1}{\ell}$
  times the total cost, which is at most
  \[
    |\Su_{3\ell}|\cdot 2\log n \leq 6\ell z\log n.
  \]

  \begin{wrapfigure}{r}{0.44\textwidth}
  \vspace{-1cm}    
  \begin{tikzpicture}[xscale=-.7,yscale=.37,every node/.style={inner sep=2pt},vert/.style={inner sep=1.5pt,circle,fill},rotate=-90]
    \draw (0,4.5) node[vert](eps) {};

    \draw (2,9) node[vert] (D) {};
    \draw (2,6.5) node[vert](a) {};
    \draw (2,2) node[vert](b) {};

    \draw[densely dotted] (eps) --node[left=15pt]{$\texttt{\vphantom{ab\$}\$}$} (D) {};
    \draw[densely dotted] (eps) --node[left=3pt]{$\texttt{\vphantom{ab\$}a}$} (a) {};
    \draw[very thick] (eps) --node[right=5pt]{$\texttt{\vphantom{ab\$}b}$} (b) {};

    \draw (4,9) node[vert](aD) {};
    \draw (4,8) node[vert](aa) {};
    \draw (4,6) node[vert](ba) {};
    \draw (4,4) node[vert](Db) {};
    \draw (4,2) node[vert](ab) {};
    \draw (4,0) node[vert](bb) {};

    \draw[very thick] (D) --node[left]{$\texttt{\vphantom{ab\$}a}$} (aD) {};
    \draw[densely dotted] (a) --node[left=3pt]{$\texttt{\vphantom{ab\$}a}$} (aa) {};
    \draw[very thick] (a) --node[right]{$\texttt{\vphantom{ab\$}b}$} (ba) {};
    \draw[densely dotted] (b) --node[left=4pt]{$\texttt{\vphantom{ab\$}\$}$} (Db) {};
    \draw[very thick] (b) --node[left]{$\texttt{\vphantom{ab\$}a}$} (ab) {};
    \draw[densely dotted] (b) --node[right=3pt]{$\texttt{\vphantom{ab\$}b}$} (bb) {};

    \draw (6,9) node[vert](baD) {};
    \draw (6,8) node[vert](baa) {};
    \draw (6,6.5) node[vert](aba) {};
    \draw (6,5) node[vert](bba) {};
    \draw (6,4) node[vert](aDb) {};
    \draw (6,3) node[vert](aab) {};
    \draw (6,1.5) node[vert](bab) {};
    \draw (6,0) node[vert](Dbb) {};

    \draw[very thick] (aD) --node[left]{$\texttt{\vphantom{ab\$}b}$} (baD) {};
    \draw[very thick] (aa) --node[left]{$\texttt{\vphantom{ab\$}b}$} (baa) {};
    \draw[very thick] (ba) --node[left]{$\texttt{\vphantom{ab\$}a}$} (aba) {};
    \draw[densely dotted] (ba) --node[right]{$\texttt{\vphantom{ab\$}b}$} (bba) {};
    \draw[very thick] (Db) --node[left]{$\texttt{\vphantom{ab\$}a}$} (aDb) {};
    \draw[densely dotted] (ab) --node[left]{$\texttt{\vphantom{ab\$}a}$} (aab) {};
    \draw[very thick] (ab) --node[right]{$\texttt{\vphantom{ab\$}b}$} (bab) {};
    \draw[very thick] (bb) --node[left]{$\texttt{\vphantom{ab\$}\$}$} (Dbb) {};

    \draw (8,9) node[vert](abaD) {};
    \draw (8,8) node[vert](abaa) {};
    \draw (8,7) node[vert](aaba) {};
    \draw (8,6) node[vert](baba) {};
    \draw (8,5) node[vert](Dbba) {};
    \draw (8,4) node[vert](baDb) {};
    \draw (8,3) node[vert](baab) {};
    \draw (8,2) node[vert](abab) {};
    \draw (8,1) node[vert](bbab) {};
    \draw (8,0) node[vert](aDbb) {};

    \draw[very thick] (baD) --node[left]{$\texttt{\vphantom{ab\$}a}$} (abaD) {};
    \draw[very thick] (baa) --node[left]{$\texttt{\vphantom{ab\$}a}$} (abaa) {};
    \draw[densely dotted] (aba) --node[left]{$\texttt{\vphantom{ab\$}a}$} (aaba) {};
    \draw[densely dotted] (aba) --node[right]{$\texttt{\vphantom{ab\$}b}$} (baba) {};
    \draw[very thick] (bba) --node[left]{$\texttt{\vphantom{ab\$}\$}$} (Dbba) {};
    \draw[very thick] (aDb) --node[left]{$\texttt{\vphantom{ab\$}b}$} (baDb) {};
    \draw[very thick] (aab) --node[left]{$\texttt{\vphantom{ab\$}b}$} (baab) {};
    \draw[densely dotted] (bab) --node[left]{$\texttt{\vphantom{ab\$}a}$} (abab) {};
    \draw[densely dotted] (bab) --node[right]{$\texttt{\vphantom{ab\$}b}$} (bbab) {};
    \draw[very thick] (Dbb) --node[left]{$\texttt{\vphantom{ab\$}a}$} (aDbb) {};
    \end{tikzpicture}
    \caption{The trie $\mathcal{T}$ of reversed length-$4$ substrings
      of $\T^\infty$ for $\T=\texttt{bbabaababababaababa\$}$ of
      \cref{fig:example}. Light edges are thin and
      dotted.}\label{fig:trie}
    \vspace{-.5cm}
  \end{wrapfigure}
  To devise the announced assignment of cost to characters of strings
  in $\Su_{3\ell}$, consider the trie $\mathcal{T}$ of all reversed
  strings in $\Su_{\ell}$ (see \cref{fig:trie} for an example).  By
  $v_X$ denote the node~of~$\mathcal{T}$ whose path from the root
  of~$\mathcal{T}$ is labelled by a string $X$.
  
  Let $\LCP[i] \in [\ell \dd 2\ell)$ be an irreducible LCP value; note
  that $i > 1$ due to $\LCP[i] \ge \ell > 0$.  Let $j_0 = \SA[i - 1]$
  and $j_1 = \SA[i]$ so that $\LCP[i] = \LCE(j_0, j_1)$.  Since
  $\LCP[i]$ is irreducible, we have $\T^{\infty}[j_0 - 1] = \BWT[i -
  1] \neq \BWT[i] = \T^{\infty}[j_1 - 1]$.  For $k \in [1 \dd
  \ell]$, the $k$th unit of the cost associated with $\LCP[i]$ is
  charged to the $k$th character ($\T^{\infty}[j_t - 1]$) of the
  string $\T^{\infty}[j_t - k \dd j_t - k + 3\ell) \in \Su_{3\ell}$,
  where $t \in \{0, 1\}$ is such that the subtree of $\mathcal{T}$
  rooted at $v_{\revstr{\T^{\infty}[j_t - 1 \dd j_t - k + \ell)}}$
  contains less leaves than the subtree rooted at
  $v_{\revstr{\T^{\infty}[j_{1-t} - 1 \dd j_{1-t} - k + \ell)}}$ (we
  choose $t = 0$ in case of ties).

  Note that at most $\log n$ characters of each $S \in \Su_{3\ell}$
  can be charged during the above procedure: whenever $S[k]$, with $k
  \in [1\dd \ell]$, is charged, the subtree of $\mathcal{T}$ rooted at
  $v_{\revstr{S[k+1\dd \ell]}}$ has at least twice as many leaves as
  the subtree rooted at $v_{\revstr{S[k\dd \ell]}}$, and this can
  happen for at most $\log|\Su_{\ell}| \leq \log n$ nodes
  $v_{\revstr{S[k\dd \ell]}}$ on the path from the root of
  $\mathcal{T}$ to the leaf $v_{\revstr{S[1\dd \ell]}}$.

  It remains to show that, for every $S\in\Su_{3\ell}$, a single
  position $S[k]$, with $k\in[1\dd \ell]$, can be charged at most
  twice. For this, observe that the characters charged for a single
  irreducible value $\LCP[i]$ are at different positions (of strings
  in $\Su_{3\ell}$). Hence, to analyze the total charge assigned to
  $S[k]$, we only need to bound the number of possible candidate
  positions~$i$.  Let $[b\dd e]$ be the set of indices $i'$ such that
  $\T^\infty[\SA[i']\dd]$ starts with $S[k+1\dd 3\ell]$.  In the above
  procedure, if a character $S[k]$ is charged a unit of cost
  corresponding to $\LCP[i]$, then $S[k+1\dd 3\ell]$ is a prefix of
  either $\T^\infty[\SA[i-1]\dd]=T^\infty[j_0\dd ]$ or
  $\T^\infty[\SA[i]\dd]=T^\infty[j_1\dd ]$. Hence, $\{i-1,i\}\cap[b\dd
  e]\neq\emptyset$. At the same time, $\LCE(\SA[i-1],\SA[i]) <
  2\ell$ and all strings $\T^\infty[\SA[i']\dd]$ with $i'\in [b\dd e]$
  share a common prefix $S[k+1\dd 3\ell]$ of length $3\ell-k \ge
  2\ell$.  Consequently, $i=b$ or $i=e+1$.\qed

\begin{theorem}\label{th:basic-upper-bound}
  Every string of length $n$ satisfies $r=\bigO(z \log^2 n)$.
\end{theorem}
\begin{proof}
  Recall that $r$ is the total number of irreducible LCP values.
  Thus, the claim follows by applying \cref{lm:irreduc-sum-2} for
  $\ell_i=2^i$, with $i\in [0 \dd \lfloor \log n \rfloor]$, and
  observing that the number of LCP values 0 is exactly $\sigma\leq z$.
\end{proof}

\subsection{Tighter Upper Bound}\label{sec:tighter-upper-bound}

To obtain a tighter bound, we refine the ideas from
\cref{sec:basic-upper-bound}, starting with a counterpart of
\cref{lm:irreduc-sum-2}.

\begin{lemma}\label{lm:irreduc-sum-3}
  For every $\ell\in[1\dd n]$, the number of irreducible LCP values in
  $[\ell \dd 2\ell)$ is $\bigO(z \log z)$.
\end{lemma}
\begin{proof}
  The proof follows closely that of \cref{lm:irreduc-sum-2}. However,
  with each irreducible value $\LCP[i]\in[\ell\dd 2\ell)$, we
  associate cost $\lceil \frac12\ell\rceil$ rather $\ell$.  We then
  show that each of the strings in $\Su_{3\ell}$ is charged at most
  $2\cdot(3+\log z)$ times (rather than $2\log n$ times). Then, the
  number of irreducible LCP values in the range $[\ell\dd 2\ell)$ does
  not exceed $\frac{2}{\ell}$ times the total cost, which is bounded by
  \[
    |\Su_{3\ell}|\cdot 2\cdot (3+\log z) \leq 6\ell z(3+\log z).
  \]

  Recall the trie $\mathcal{T}$ of all reversed strings in~$\Su_{\ell}$.
  For a node $v$ of $\mathcal{T}$, by $\size(v)$ we
  denote the number of leaves in the subtree of $\mathcal{T}$ rooted
  in $v$. An edge connecting $v\neq {\rm root}(\mathcal{T})$ to its
  parent in $\mathcal{T}$ is called \emph{light} if $v$ has a sibling
  $v'$ satisfying ${\rm size}(v') \ge \size(v)$ (see \cref{fig:trie}).
  In the proof of \cref{lm:irreduc-sum-2}, we observed that the
  characters $S[k]$ of $S\in\Su_{3\ell}$ that can be charged
  correspond to light edges on the path from the root of $\mathcal{T}$
  to the leaf $v_{\revstr{S[1 \dd \ell]}}$: whenever $S[k]$, with $k \in [1
  \dd \ell]$, is charged, the edge connecting $v_{\revstr{S[k \dd
  \ell]}}$ to its parent $v_{\revstr{S[k + 1 \dd \ell]}}$ is
  light.  We then noted that there are at most $\log |\Su_{\ell}| \leq
  \log n$ light edges on each root-to-leaf path
  in~$\mathcal{T}$. Here, we perform the same assignment of cost to
  the characters of strings in $\Su_{3\ell}$ as in
  \cref{lm:irreduc-sum-2}, but only for units $k\in[1 \dd \lceil
  \frac{1}{2} \ell \rceil]$. This implies that only characters
  $S[k]$ of $S\in \Su_{3\ell}$ with $k \le \lceil \frac{1}{2} \ell
  \rceil$ are charged. It remains to show that any root-to-leaf path
  in $\mathcal{T}$ contains at most $3 + \log z$ light edges between a
  node at depth at least $\lfloor \frac{1}{2}\ell \rfloor$ and
  its~child.

  Consider a light edge from a node $v$ to its parent $u$ at depth at
  least $\lfloor{\frac{1}{2}\ell}\rfloor$.  Let $v'$ be a sibling of
  $v$ satisfying $\size(v') \ge \size(v)$, and let $S_v, S_{v'}$ be
  the labels of the paths from the root to $v$ and $v'$, respectively.
  These labels differ on the last position only so, by
  \cref{fct:disper}, they cannot be both periodic.  Let
  $\widetilde{v}\in\{v,v'\}$ be such that $S_{\widetilde{v}}$ is not
  periodic, and let $\widetilde{m}=\size(\widetilde{v})$.

  Consider the set $\Su$ of length-$\ell$ strings corresponding to the
  leaves in the subtree of $\mathcal{T}$ rooted at $\widetilde{v}$
  (i.e., the labels of the root-to-leaf paths passing through
  $\widetilde{v}$).  Define $\revstr{\Su}:=\{\revstr{P} : P \in \Su\}$
  and note that $\revstr{\Su} \subseteq \Su_{\ell}$ because
  $\mathcal{T}$ is the trie of \emph{reversed} strings from
  $\Su_{\ell}$.  Let $e_1 < \cdots < e_{\widetilde{m}}$ denote the
  ending positions of the leftmost occurrences in $\T^\infty[1 \dd)$
  of strings in $\revstr{\Su}$.  By definition, we have an occurrence
  of $\revstr{S_{\widetilde{v}}}$ ending in $\T^\infty$ at every
  position $e_i$ with $i \in [1 \dd \widetilde{m}]$. Now,
  $\per(S_{\widetilde{v}}) > \frac12 |S_{\widetilde{v}}| \ge
  \frac{1}{4} \ell$ implies that $e_{i+1} - e_{i}> \frac{1}{4}\ell$
  for every $i \in [1 \dd \widetilde{m} - 1]$ (otherwise, the two
  close occurrences of $\revstr{S_{\widetilde{v}}}$ would yield
  $\per(\revstr{S_{\widetilde{v}}}) = \per(S_{\widetilde{v}}) \le
  \frac14 \ell$).  Consequently, at least $\frac14|\Su| = \frac14
  \widetilde{m}$ length-$\ell$ substrings of $T^\infty[1 \dd )$ have
  disjoint leftmost occurrences.  Since each leftmost occurrence
  crosses or begins at a phrase boundary of the LZ77 parsing of $\T$,
  we conclude that $z \ge \frac14 \widetilde{m}$, and therefore
  $\size(v) \le \size(\widetilde{v}) = \widetilde{m} \le 4z$.

  The reasoning above shows that once a root-to-leaf path encounters a
  light edge connecting a node $u$ at depth at least
  $\lfloor{\frac{1}{2}\ell}\rfloor$ to its child $v$, we have
  $\size(v) \leq 4z$. The number of the remaining light edges on the
  path is at most $\log(\size(v)) \leq 2 + \log z$ by the standard
  bound applied to the subtree of~$\mathcal{T}$ rooted at $v$.
\end{proof}

\begin{theorem}\label{th:tighter-upper-bound}
  Every string of length $n$ satisfies $r=\bigO(z \log z \max(1,\log
  \frac{n}{z \log z}))$.
\end{theorem}
\begin{proof}
  To obtain tighter bounds on the number of irreducible LCP values in
  $[\ell\dd 2\ell)$, we consider three cases:
  \begin{enumerate}
  \item $\ell \le \log z$. We repeat the proof of
    \cref{lm:irreduc-sum-2}, except that we observe that the number of
    light edges on each root-to-leaf path in $\mathcal{T}$ is bounded
    by $\ell$. Thus, the number of irreducible LCP values in $[\ell\dd
    2\ell)$ is $\bigO(z\ell)$.
  \item $\log z < \ell \le \frac{n}{z}$. We use the bound $\Oh(z\log
    z)$ of \cref{lm:irreduc-sum-3}.
  \item $\frac{n}{z} < \ell$. We repeat the proof of
    \cref{lm:irreduc-sum-3}, except that we observe that
    $|\Su_{3\ell}|\leq n$. Thus, the number of irreducible LCP values
    in $[\ell\dd2\ell)$ is $\bigO(\frac{n\log z}{\ell})$.
  \end{enumerate}
  The above upper bounds, applied for every $\ell=2^i$ with $i\in
  [0\dd\lfloor{\log n}\rfloor]$, yield
  \begin{align*}
    r &\leq \sigma + \sum_{i=0}^{\lfloor \log n \rfloor}
         \left|\left\{j\in [2 \dd n] : \BWT[j{-}1] \neq \BWT[j]
         \text{ and } \LCP[j] \in \left[2^i \dd 2^{i+1}\right)\right\}\right|\\
      &= \sigma + \sum_{i=0}^{\lfloor \log\log z \rfloor}
         \bigO\left(z2^{i} \right) + \sum_{i = \lfloor \log\log
         z \rfloor + 1}^{\left\lfloor\log\frac{n}{z}\right\rfloor}
         \bigO \left(z \log z \right) +
         \sum_{i = \left\lfloor \log \frac{n}{z}
         \right\rfloor + 1}^{\lfloor \log n \rfloor}
         \bigO \left( \frac{n \log z}{2^i} \right)\\
      &= \sigma + \bigO \left(z \log z \right) + \bigO \left(z \log z
          \max\left(1, \log \tfrac{n}{z \log z} \right) \right) +
         \bigO \left(z \log z \right)\\
      &= \bigO \left(z \log z \max\left(1, \log \tfrac{n}{z
         \log z} \right) \right). \qedhere
  \end{align*}
\end{proof}

\subsection{Upper Bound in Terms of \texorpdfstring{$\delta$}{Delta}}\label{sec:upper-bound-for-delta}

Let $\delta=\max_{m=1}^n\frac1m|\Su_m|$ denote the \emph{substring
complexity} of $\T$~\cite{delta}. Note that letting
$\delta=\sup_{m=1}^{\infty}\frac1m|\Su_m|$ is equivalent because
$|\Su_m|\leq n$ holds for $m\ge 1$, which implies $\frac1m|\Su_m|\le
1\leq |\Su_1|$ for $m\ge n$.  We start by noting that $\delta \le z$
since $|\Su_m| \leq mz$ holds for every $m\geq 1$, as observed in the
proof of \cref{lm:irreduc-sum-2}.  Furthermore, $|\Su_m| \leq m\delta$
holds by definition of $\delta$, so $\delta$ can replace $z$ in the
proof of \cref{lm:irreduc-sum-2}.

To adapt the proof \cref{lm:irreduc-sum-3}, we need to generalize the
observation that at most $z$ substrings from $\Su_{\ell}$ may have
disjoint leftmost occurrences in $\T^\infty[1\dd)$.  This observation
is easy since the LZ77 parsing naturally yields a set of $z$ positions
(phrase boundaries) in $\T$.  The substring complexity $\delta$ does
not provide such structure, but as the lemma below implies, we can
replace $z$ by $3\delta$ in the aforementioned observation.  The proof
of \cref{lm:leftmost-occ-delta} is a straightforward modification of
the argument used in~\cite[Lemma 6]{delta}.  For completeness, below
we write down the full reasoning, with technical details tailored to
our notation (e.g., $\Su_{\ell}$ defined in terms of $\T^{\infty}$
rather than $\T$).

\begin{lemma}[{based on~\cite[Lemma 6]{delta}}]\label{lm:leftmost-occ-delta}
  For any positive integer $\ell$, the total number of positions in
  $\T^\infty[1\dd)$ covered by the leftmost occurrences of strings
  from $\Su_{\ell}$ is at most $3\delta\ell$.
\end{lemma}
\begin{proof}
  Let $C$ denote the set of positions in $\T^\infty[1 \dd)$ covered by
  the leftmost occurrences of strings from $\Su_{\ell}$, and let $C' =
  C \setminus [1 \dd \ell)$.  For any $i \in C'$ denote $S_i =
  \T^{\infty}[i - \ell + 1 \allowbreak\dd i + \ell]$, and let
  $\Su=\{S_i : i \in C'\} \subseteq \Su_{2\ell}$.  We will show that
  $|\Su| = |C'|$.  Let $i \in C'$. First, observe that, due to $i\ge
  \ell$, the fragment $S_i$ is entirely contained in $\T^\infty[1
  \dd)$. Furthermore, by definition, $S_i$ contains the leftmost
  occurrence of some $S \in \Su_{\ell}$. Thus, this occurrence of
  $S_i$ in $\T^\infty[1 \dd)$ must also be the leftmost one in
  $\T^\infty[1 \dd)$. Consequently, the substrings $S_i$ for $i \in
  C'$ are distinct.

  We have thus shown that $|C'| = |\Su| \leq |\Su_{2\ell}|$.  Since
  $|\Su_{2\ell}| \leq 2\delta\ell$ holds by definition of $\delta$, we
  obtain $|C| < |C'|+\ell \leq |\Su_{2\ell}|+\ell \leq (2\delta+1)\ell
  \leq 3\delta\ell$.
\end{proof}

\begin{lemma}\label{lm:irreduc-sum-3-delta}
  For every $\ell\in [1\dd n]$, the number of irreducible LCP values
  in $[\ell\dd 2\ell)$ is $\bigO(\delta \log \delta)$.
\end{lemma}
\begin{proof}
  Compared to the proof of \cref{lm:irreduc-sum-3}, we use the bound
  $|\Su_{3\ell}|\le 3\ell\delta$ instead of $|\Su_{3\ell}|\le 3\ell
  z$.  The only other modification needed is that, for every light
  edge connecting a node $u$ of $\mathcal{T}$ at depth at least
  $\lfloor{\frac12 \ell}\rfloor$ to its child $v$, we need to prove
  $\size (v) =\bigO(\delta)$.

  Let $\widetilde{m} \ge \size(v)$ be defined as in the proof of
  \cref{lm:irreduc-sum-3}.  Recall that we have identified at least
  $\frac{\widetilde{m}}{4}$ strings in $\Su_{\ell}$ whose leftmost
  occurrences in $\T^\infty[1\dd)$ are disjoint.  By
  \cref{lm:leftmost-occ-delta}, there are at most $3\delta$ such
  substrings.  Thus, $\size(v)\le \widetilde{m}\le 12\delta$.
\end{proof}\vspace{1ex}

By replacing the thresholds $\log z$ and $\frac{n}{z}$ with $\log
\delta$ and $\frac{n}{\delta}$, respectively, in the proof of
\cref{th:tighter-upper-bound}, we immediately obtain a bound in terms
of $\delta$.

\begin{theorem}\label{th:tight-upper-bound}
  Every string of length $n$ satisfies $r=\bigO(\delta \log \delta
  \max(1,\log \frac{n}{\delta \log \delta}))$.
\end{theorem}

Note that the trivial upper bound $r = \Oh(n)$ is tighter if $\delta
\log \delta > n$.  In~\cref{sec:lower-bound}, we show that a
combination of these two upper bounds is asymptotically optimal. For
this, we construct tight examples in which the values $\delta$ cover
the whole spectrum between $\Oh(1)$ and $\Omega(n)$.

\subsection{Further Upper Bounds}
By combining \cref{th:tight-upper-bound} with known properties of
the substring complexity $\delta$, we obtain the first bound relating
the number of BWT runs in the string and its reverse. No such bounds
(even polynomial in $r\log n$) were known before.

\begin{corollary}\label{cor:rr}
  If $r$ and $\bar{r}$ denote the number of runs in the BWT of a
  length-$n$ text and its reverse, respectively, then $\bar{r} =
  \bigO(r \log r \max(1, \frac{n}{r \log r}))$.
\end{corollary}
\begin{proof}
  Since the value of $\delta$ is the same for the text and its
  reverse, \cref{th:tight-upper-bound} yields $\bar{r}=\bigO(\delta
  \log \delta \max(1,\log \frac{n}{\delta \log \delta}))$.
  Combining~\cite[Theorem~3.9]{attractors} and~\cite[Lemma~$2$]{delta}
  gives $\delta\leq r$.  Consequently, we obtain $\bar{r}= \bigO(r
  \log r \max(1, \frac{n}{r \log r}))$.
\end{proof}\vspace{1ex}

Our technique also lets us strengthen the bound of
\cref{lm:irreduc-sum} on the sum of irreducible LCP values.

\begin{theorem}\label{th:irreduc-sum-upper-bound}
  For every string of length $n$, the sum of all irreducible LCP
  values is $\bigO(n \log \delta)$.
\end{theorem}
\begin{proof}
  As for the irreducible LCP values not exceeding
  $\frac{n}{\delta}$, \cref{lm:irreduc-sum-3-delta}
  immediately yields
  \begin{align*}
    \sum_{\substack{\LCP[j]\le n/\delta\\\BWT[j{-}1]\neq\BWT[j]}} \LCP[j]
    &< \sum_{i=0}^{\lfloor\log \frac{n}{\delta}\rfloor}
        2^{i+1}\cdot \left|\left\{j : \BWT[j-1]\ne \BWT[j]\text{ and }
        \LCP[j]\in [2^i\dd 2^{i+1})\right\}\right|
     \\&= \sum_{i=0}^{\lfloor\log \frac{n}{\delta}\rfloor}
         2^{i+1}\cdot \Oh\left(\delta \log \delta\right)
      = \bigO\left(n \log \delta\right).
  \end{align*}
  For irreducible LCP values larger than $\frac{n}{\delta}$, a simple
  approach would be to separately consider $\Theta(\log\delta)$ ranges
  of LCP values, $[2^{i} \dd 2^{i+1})$ for $i \in [\lfloor \log
  \frac{n}{\delta} \rfloor \dd \lfloor \log n \rfloor]$, and bound the
  number of irreducible LCP values in each range by $\Oh(\frac{n \log
  \delta}{2^i})$, as in the proof of \cref{th:tight-upper-bound}.
  Unfortunately, this only gives an overall bound of $\bigO(n \log^2
  \delta)$ on the sum of irreducible $\LCP$ values.

  Instead, we employ a similar scoring as in the proof of
  \cref{lm:irreduc-sum-3-delta}, except that we handle all the
  irreducible values $\LCP[i]>\frac{n}{\delta}$ together.  With each
  such value, we associate $\LCP[i]-\lfloor \frac{n}{2\delta}\rfloor
  \ge \frac12\LCP[i]$ units of cost, and we charge them to individual
  characters of strings in $\Su_{n}$.  We then show that each of the
  strings in $\Su_n$ is charged at most $6+2\log \delta$ times.
  Consequently, the sum of irreducible LCP values larger than
  $\frac{n}{\delta}$ does not exceed twice the total cost, which is
  bounded by
  \[
    2|\Su_{n}|\cdot (6+2\log \delta) = 4n(3+\log \delta).
  \]

  The cost assignment is based on the trie $\mathcal{T}$ of all
  reversed strings in $\Su_n$.  Let $\LCP[i]>\frac{n}{\delta}$ be an
  irreducible $\LCP$ value, and let $j_0 = \SA[i-1]$ and $j_1 =
  \SA[i]$ so that $\LCP[i]=\LCE(j_0,j_1)$.  Since $\LCP[i]$ is
  irreducible, we have $T^\infty[j_0-1]=\BWT[i-1]\ne
  \BWT[i]=T^\infty[j_1-1]$.  For $k\in (n-\LCP[i]\dd n-\lfloor
  \frac{n}{2\delta}\rfloor]$, the $k$th unit of the cost associated
  with $\LCP[i]$ is charged to the $k$th character ($T^\infty[j_t-1]$)
  of the string $T^\infty[j_t-k\dd j_t-k+n)\in \Su_n$, where
  $t\in\{0,1\}$ is such that $\size(v_{\revstr{\T^\infty[j_t-1\dd
  j_t-k+n)}})\le \size(v_{\revstr{\T^\infty[j_{1-t}-1\dd
  j_{1-t}-k+n)}})$.  Note that $\T^\infty[j_0\dd j_0-k+n) =
  T^\infty[j_1\dd j_1-k+n)$ holds due to $\LCE(j_0,j_1)>n-k$,
  so the edge from $v_{\revstr{\T^\infty[j_t-1\dd j_t-k+n)}}$ to its
  parent $v_{\revstr{\T^\infty[j_t\dd j_t-k+n)}}$ is light.

  For every $S\in \Su_n$, a single position $S[k]$, with $k\in [1\dd
  n-\lfloor \frac{n}{2\delta}\rfloor]$, can be charged at most
  twice.  This is because there is a unique position $j\in [1\dd n]$
  such that $S = T^\infty[j-k\dd j-k+n)$, and $S[k]$ can be charged
  for $\LCP[i]$ only if $j=\SA[i-1]$ or $j=\SA[i]$.

  It remains to prove that at most $3+\log \delta$ characters of each
  $S\in \Su_n$ can be charged.  For this, we note that every charged
  character $S[k]$, with $k\in [1\dd n-\lfloor
  \frac{n}{2\delta}\rfloor]$, corresponds to a light edge on the
  path from the root of $\mathcal{T}$ to the leaf $v_{\revstr{S}}$
  connecting a node $u$ at depth at least
  $\lfloor\frac{n}{2\delta}\rfloor$ to its child~$v$.  Let us fix the
  highest such pair $(u,v)$.  It is not difficult to see that the
  argument from the proof \cref{lm:irreduc-sum-3} yields at least
  $\frac{\size(v)}{4\delta}$ strings in $\Su_n$ with disjoint leftmost
  occurrences in $\T^\infty[1\dd )$.  However, all such occurrences
  overlap, so $\size(v)\le 4\delta$ and, consequently, the number of
  the remaining light edges on the path from the root of $\mathcal{T}$
  to $v_{\revstr{S}}$ is at most $\log(\size(v))\le 2+\log \delta$.
  Including the edge from $u$ to $v$, we obtain a bound of $3+\log
  \delta$ chargeable characters in total.
\end{proof}

Due to $\delta \leq r$, the presented upper bound is always
(asymptotically) at least as strong as the bound of
\cref{lm:irreduc-sum}; furthermore, it can be strictly stronger since
$\log \delta=o(\log r)$ is possible when $\delta = \log^{o(1)} n$.
In~\cref{sec:lower-bound}, we construct strings proving tightness of
the new bound for the values $\delta$ ranging from $\Oh(1)$ to
$\Omega(n)$.

\section{Lower Bounds}\label{sec:lower-bound}

\newcommand{\bin}{\mathrm{bin}}

In this section, we present examples showing asymptotic tightness of
the upper bounds in \cref{sec:upper-bound-for-delta}.

\subsection{Lower Bound for the Number of BWT Runs}\label{sec:lower-bound-r}

We give two constructions, corresponding to the bound of
\cref{th:tight-upper-bound} and the trivial bound $r=\Oh(n)$,
respectively.

For $\ell\geq 1$, let $\bin_{\ell}(x) \in
\{\mathtt{0},\mathtt{1}\}^{\ell}$ be the binary representation of
$x\in[0\dd 2^\ell)$, and let $\bin^{-1}_\ell$ be the inverse mapping.

\begin{lemma}\label{lm:small-delta-lower-bound}
  For all integers $\ell\ge 2$ and $K\ge 1$, the length $n$, the
  substring complexity $\delta$, and the number of runs $r$ in the BWT
  of a string $\T_{\ell,K}\in
  \{\$,\mathtt{0},\mathtt{1},\mathtt{2}\}^+$, defined with
  \[
    \T_{\ell,K}=\left(\bigodot_{k=0}^{K-1}\bigodot_{i=0}^{2^\ell-1}
    \left(\mathtt{2}^{2^k\ell}\cdot
      {\bin}_{\ell}(i)\right)\right)\cdot\$,
  \]
  satisfy $n = \Theta(2^{K+\ell} \ell)$, $\delta = \Theta(2^\ell)$,
  and $r = \Omega(2^\ell \ell K)$.
\end{lemma}
\begin{proof}
  Let $B_{\ell,k} := \bigodot_{i=0}^{2^\ell-1}
  \left(\mathtt{2}^{2^k\ell} \cdot \bin_{\ell}(i)\right)$ so that
  $\T_{\ell,K} $ is the concatenation of $B_{\ell,k}$ for $k \in [0
  \dd K)$, followed by a $\$$.  Note that $|B_{\ell,k}| = 2^\ell
  \cdot (2^k\ell + \ell) = \Theta(2^{k+\ell}\ell)$, so the length of
  $\T_{\ell,K}$ satisfies
  \[
    n = 1+\sum_{k=0}^{K-1}|B_{\ell,k}|
      = 1+\sum_{k=0}^{K-1}\Theta\left(2^{k+\ell}\ell\right)
      = \Theta(2^{K+\ell}\ell).
  \]

  To show that $\frac{1}{m} |\Su_m| = \bigO(2^\ell)$ holds for every
  $m\in [1\dd n]$, we consider three cases:
  \begin{itemize}
  \item $m\leq \ell$. Observe that any occurrence of $S \in \Su_m \cap
    \{\mathtt{0},\mathtt{1},\mathtt{2}\}^{*}$ in $\T_{\ell,K}$
    overlaps at most one maximal block of $\mathtt{2}$s. It is easy to
    see that there are $\bigO(2^m)$ such strings $S$. Adding $m$
    substrings containing the symbol $\$$ thus yields $|\Su_m| =
    m+\bigO(2^m) = \bigO(2^\ell)$.
  \item $\ell< m\leq 2^{K-1}\ell$. Observe that any length-$m$
    substring of $B_{\ell,k+1}$ with $m \leq 2^k\ell$ is also a
    substring of $B_{\ell,k}$.  More generally, if a length-$m$
    substring of $\T_{\ell,K}$ does not contain a $\$$, then its
    leftmost occurrence starts within $B_{\ell,k}$ for some $k \in
    [0\dd {\lceil \log \frac{m}{\ell}\rceil}]$. Hence,
    \[
      |\Su_m| \leq m + \sum_{ k = 0}^{\lceil \log \frac{m}{\ell}
        \rceil} | B_{\ell,k}| = m +\sum_{k =0}^{\lceil \log \frac{m}{\ell}
        \rceil} \Oh\left(2^{k+\ell}\ell \right) =
      \bigO(\tfrac{m}{\ell}\cdot 2^\ell \cdot \ell)=\Oh(m\cdot
      2^\ell).
    \]
  \item $2^{K-1}\ell< m$. Then, $\frac{1}{m} |\Su_m| \le \frac{n}{m} =
    \bigO\big(\frac{2^{K+\ell}\ell}{2^K\ell}\big) = \bigO(2^\ell)$.
  \end{itemize}

  To show $\delta = \Omega(2^\ell)$, observe that, for any $i \in [0
  \dd 2^\ell - 1)$ and $t \in [0 \dd \ell)$, the string $2^t \cdot
  \bin_\ell(i) \cdot 2^{\ell - t} \in \{\mathtt{0}, \mathtt{1},
  \mathtt{2}\}^{2\ell}$ is a substring of $\T_{\ell, K}$. Since all
  these substrings are different, we have $|\Su_{2\ell}| \geq(2^\ell -
  1) \cdot \ell$ and $\delta \ge \frac{1}{2\ell}|\Su_{2\ell}| \ge
  2^{\ell-1}-\frac{1}{2} = \Omega(2^\ell)$.

  As for the lower bound on $r$, we start by observing that, for any
  strings $S,P\in\{\mathtt{0}, \mathtt{1}\}^{+}$ such that
  $|S|+|P|=\ell$ and any integer $k\in [0\dd K)$, the string
  $S\mathtt{2}^{2^k\ell}P$ is a substring of $B_{\ell,k}^{\infty}$.
  Let us define $s = \bin^{-1}_{|S|}(S)$ and $p =
  \bin^{-1}_{|P|}(P)$, and consider three cases:
  \begin{itemize}
  \item $S \neq \mathtt{1}^{|S|}$. Let $x=p2^{|S|} + s < 2^\ell - 1$
    and observe that $S\mathtt{2}^{2^k\ell}P$ is a substring of
    $\bin_{\ell}(x) \cdot \mathtt{2}^{2^k\ell} \cdot \bin_{\ell}(x +
    1) = \bin_{|P|}(p)\cdot \bin_{|S|}(s)\cdot \mathtt{2}^{2^k\ell}
    \cdot \bin_{|P|}(p)\cdot\bin_{|S|}(s+1)$.
  \item $S = \mathtt{1}^{|S|}$ and $P \neq \mathtt{0}^{|P|}$. Let $x =
    (p-1)2^{|S|}+s< 2^\ell-1$ and observe that
    $S\mathtt{2}^{2^k\ell}P$ is a substring of $\bin_{\ell}(x) \cdot
    \mathtt{2}^{2^k\ell} \cdot \bin_{\ell}(x + 1) =
    \bin_{|P|}(p-1)\cdot \bin_{|S|}(s)\cdot \mathtt{2}^{2^k\ell} \cdot
    \bin_{|P|}(p)\cdot\bin_{|S|}(0)$.
  \item $S=\mathtt{1}^{|S|}$ and $P=\mathtt{0}^{|P|}$. Then,
    $S\mathtt{2}^{2^k\ell}P$ is a substring of $\bin_{\ell}(2^\ell-1)
    \cdot \mathtt{2}^{2^k\ell} \cdot \bin_{\ell}(0)$.
  \end{itemize}

  Note that in all but the last case, $S\mathtt{2}^{2^k\ell}P$ is a
  substring of $B_{\ell,k}$. However, since every $B_{\ell,k}$ ends
  with $\bin_{\ell}(2^\ell-1)$, $S\mathtt{2}^{2^k\ell}P$ also occurs
  in $\T_{\ell,K}$ unless $S=\mathtt{1}^{|S|}$, $P=\mathtt{0}^{|P|}$,
  and $k=0$.  The number of remaining triples $(S,P,k)$ is equal to
  $\ell-1$ times the number of maximal blocks of $\mathtt{2}$s in
  $\T_{\ell,K}$ excluding the prefix $\mathtt{2}^\ell$ of
  $\T_{\ell,K}$. As distinct triples~yield distinct substrings, no
  substring $S\mathtt{2}^{2^k\ell}P$ occurs in $\T_{\ell,K}$ twice.

  We will focus on triples satisfying $S \neq \mathtt{1}^{|S|}$ and $P
  \neq \mathtt{1}^{|P|}$. By the discussion above, for each such
  triple, $S\mathtt{2}^{2^k\ell}P$ occurs exactly once in
  $\T_{\ell,K}$.  We will now show that every such triple $(S,P,k)$
  corresponds to a different run in the BWT of $\T_{\ell,K}$. Thus,
  computing the number of considered triples yields
  \[
    r \geq K\cdot \sum_{t=1}^{\ell-1}(2^t-1)(2^{\ell-t}-1) =
    K\left(2^\ell(\ell - 3) + \ell + 3 \right)=
    \Omega \left(2^\ell \ell K \right).
  \]

  Let $i$ be the rank of the suffix of $\T_{\ell,K}$ having
  $S\mathtt{2}^{2^k\ell}P$ as a prefix.  We will show that $i + 1$ is
  the rank of the suffix prefixed with $S\mathtt{2}^{2^k\ell}P'$,
  where $P' = \bin_{|P|}(p + 1)$.  Suppose that these suffixes are not
  adjacent in $\SA$.  Any suffix that is in between them must start
  with $S\mathtt{2}^{2^k\ell}c$ for some $c \in
  \{\mathtt{0},\mathtt{1}\}$. Furthermore, since every maximal block
  of $\mathtt{2}$s in $\T_{\ell,K}$ is followed by some $\widehat{P}
  \in \{\mathtt{0},\mathtt{1}\}^{|P|}$, and there is exactly one
  occurrence of $S \mathtt{2}^{2^k\ell} \widehat{P}$ in $\T_{\ell,K}$
  for every possible $\widehat{P}$, there would exist $\widehat{P}$
  such that $P\prec \widehat{P} \prec P'$.  This is clearly impossible
  since it implies $p < \bin^{-1}_{|P|}(\widehat{P}) < p + 1$.
  
  Observe now that, by $S \neq \mathtt{1}^{|S|}$, the substrings
  $S\mathtt{2}^{2^k\ell}P$ and $S\mathtt{2}^{2^k\ell}P'$ are preceded
  in $\T_{\ell,K}$ by $P=\bin_{|P|}(p)$ and $P'=\bin_{|P|}(p+1)$,
  respectively.  Therefore, $\BWT[i + 1] = (p + 1) \bmod 2 \neq p
  \bmod 2 = \BWT[i]$.
\end{proof}

\begin{lemma}\label{lm:large-delta-lower-bound}
  For all integers $\ell\ge 2$ and $\Delta\in \Omega(2^\ell)\cap
  \Oh(2^\ell \ell)$, the length $n$, the substring
  complexity~$\delta$, and the number of runs $r$ in the BWT of a
  string $\T'_{\ell,\Delta}\in
  \{\$_1,\ldots,\$_{\Delta},\mathtt{0},\mathtt{1},\mathtt{2}\}^+$,
  defined with
  \[
    \T'_{\ell,\Delta} =
    \left(\bigodot_{i=0}^{2^\ell-1}\left(\mathtt{2}^{\ell} \cdot
    \bin_\ell(i)\right)\right) \cdot \$_1\$_2\cdots\$_{\Delta},
  \]
  satisfy $n = \Theta(2^\ell\ell)$, $\delta = \Theta(\Delta)$, and $r
  = \Omega(2^\ell\ell)$.
\end{lemma}
\begin{proof}
  The length satisfies $n=2^{\ell}\cdot 2\ell + \Delta =
  \Theta(2^\ell\ell)$.

  To show that $\frac{1}{m} |\Su_m| = \bigO(\Delta)$ holds for every
  $m\in [1\dd n]$, we consider two cases:
  \begin{itemize}
  \item $m\leq \ell \log_3 2$.  Trivially, $|\Su_m \cap
    \{\mathtt{0},\mathtt{1},\mathtt{2}\}^{*}|\le 3^m$.  Adding
    $m+\Delta-1$ substrings containing the symbols $\$_i$ yields
    \[
      \tfrac{1}{m}|\Su_m| = \bigO\left(\tfrac{3^m+\Delta}{m}\right) =
      \bigO\big(\tfrac{2^\ell+\Delta}{m}\big) =
      \Oh\left(\tfrac{\Delta}{m}\right)=\Oh(\Delta).
    \]
  \item $m > \ell \log_3 2$. Then, $\frac{1}{m} |\Su_m| \le
    \frac{n}{m}= \bigO\big(\frac{2^\ell\ell}{m} \big) =
    \Oh(2^\ell)=\bigO(\Delta)$.
  \end{itemize}
  To show $\delta = \Omega(\Delta)$, we observe that
  $|\Su_1|=\Delta+3$.

  To bound $r$ from below, as in the proof of
  \cref{lm:small-delta-lower-bound}, we observe that, for any strings
  $S,P\in\{\mathtt{0},\mathtt{1}\}^{+}$ such that $|S|+|P|=\ell$, the
  string $S\mathtt{2}^{\ell}P$ occurs exactly once in
  $\T'_{\ell,\Delta}$ unless $S=\mathtt{1}^{|S|}$ and
  $P=\mathtt{0}^{|P|}$.  Moreover, every string $S\mathtt{2}^{\ell}P$
  with $S\ne \mathtt{1}^{|S|}$ and $P\ne \mathtt{1}^{|P|}$ corresponds
  to a different run in the BWT.  Hence, counting such strings yields
  $r \ge 2^\ell\cdot (\ell-3)+\ell + 3 = \Omega(2^\ell\ell)$.
\end{proof}

Combining \cref{lm:small-delta-lower-bound,lm:large-delta-lower-bound},
we obtain the following lower bound.

\begin{theorem}\label{thm:lower}
  For every $N\ge 1$ and $\Delta\in [1\dd N]$, there exists a string
  $\T$ whose length $n$, substring complexity $\delta$, and number of
  runs $r$ in the BWT satisfy $n = \Theta(N)$, $\delta =
  \Theta(\Delta)$, and $r = \Theta(\min(n,\delta \log \delta \max(1,
  \log \frac{n}{\delta \log \delta})))$.
\end{theorem}
\begin{proof}
  If $\Delta \log \Delta < \frac12N$, we set $\T = \T_{\ell,K}$ with
  $\ell = \max(2,\lceil{\log \Delta}\rceil)$ and $K = \big\lceil{\log
  \frac{N}{\Delta \log \Delta}}\big\rceil$.  By
  \cref{lm:small-delta-lower-bound}, we thus have
  \begin{align*}
    n &= \Theta(2^{K+\ell}\ell)=\Theta\left(\tfrac{N}{\Delta\log
      \Delta} \Delta \log \Delta\right)=\Theta(N),\\ \delta &=
    \Theta(2^\ell) = \Theta(\Delta),\\ r &= \Omega(2^\ell \ell K)=
    \Omega\left(\Delta \log \Delta \log\tfrac{N}{\Delta \log \Delta}
    \right) = \Omega\left(\delta \log \delta \max(1,\log
    \tfrac{n}{\delta \log \delta})\right).
  \end{align*}

  If $\Delta \log \Delta \ge \frac12N$, we set $\T =
  \T'_{\ell,\Delta}$ with $\ell = \max(2,\lceil{\log \frac{N}{\log
  N}}\rceil)$. Note that $\Delta = \Omega(\tfrac{N}{\log
  N})=\Omega(2^\ell)$ and $\Delta = \Oh(N)=\Oh(2^\ell  \ell)$,
  so the assumptions of \cref{lm:large-delta-lower-bound} are
  satisfied.  We thus have
  \begin{align*}
    n &= \Theta(2^\ell \ell)=\Theta\left(\tfrac{N}{\log N}\log
    \tfrac{N}{\log N}\right)=\Theta(N),\\ \delta &=
    \Theta(\Delta),\\ r &= \Omega(2^\ell \ell)=
    \Omega\left(\tfrac{N}{\log N}\log \tfrac{N}{\log N}\right) =
    \Omega(N)=\Omega(n).
  \end{align*}

  In both cases, either the upper bound of \cref{th:tight-upper-bound}
  or the bound $r=\Oh(n)$ is tight.
\end{proof}

\subsection{Lower Bound for the Sum of Irreducible LCP Values}\label{sec:lower-bound-sum}

The same strings show that our upper bound $\bigO(n \log \delta)$ on
the sum of irreducible LCP values is also tight.

\begin{lemma}\label{lm:small-irreduc-sum-lower-bound}
  For all integers $\ell\ge 2$ and $K\ge 1$, the sum of irreducible
  LCP values of $\T_{\ell,K}$ is $r_{\Sigma} = \Omega(2^{K+\ell}
  \ell^2)$.
\end{lemma}
\begin{proof}
  In the proof of \cref{lm:small-delta-lower-bound}, we showed that,
  for every $k\in [0\dd K)$ and $S, P \in \{\mathtt{0},
  \mathtt{1}\}^{+} \setminus \{ \mathtt{1} \}^{+}$ such that $|S| +
  |P| = \ell$, the symbols preceding suffixes starting with
  $S\mathtt{2}^{2^k\ell}P \in \{\mathtt{0}, \mathtt{1},\mathtt{2}
  \}^{+}$ and $S\mathtt{2}^{2^k\ell}P' \in \{\mathtt{0},
  \mathtt{1},\mathtt{2}\}^{+}$, where $P' =
  \bin_{|P|}({\bin^{-1}_{|P|}(P) + 1})$, are distinct, and that these
  suffixes are adjacent lexicographically.  This implies that the
  corresponding irreducible LCP value is at least $2^k\ell$.  With $k
  = K-1$, noting that there are $\Omega(2^\ell \ell)$ choices for $S$
  and $P$, we obtain $ r_{\Sigma} = 2^{K-1}\ell \cdot \Omega(2^\ell
  \ell) = \Omega(2^{K+\ell}\ell^2)$.
\end{proof}

Analogously to \cref{lm:large-delta-lower-bound,thm:lower}, we also
get the following results:
\begin{lemma}\label{lm:large-irreduc-sum-lower-bound}
  For all integers $\ell\ge 2$ and $\Delta\in \Omega(2^\ell)\cap
  \Oh(2^\ell \ell)$, the sum of irreducible LCP values of
  $\T'_{\ell,\Delta}$ is $r_{\Sigma} = \Omega(2^{\ell}\ell^2)$.
\end{lemma}

\begin{theorem}\label{th:irreduc-sum-lower-bound}
  For every $N\ge 1$ and $\Delta\in [1\dd N]$, there exists a string
  $\T$ whose length $n$, substring complexity $\delta$, and sum
  $r_{\Sigma}$ of irreducible LCP values satisfy $n = \Theta(N)$,
  $\delta = \Theta(\Delta)$, and $r_{\Sigma}=\Theta(n \log \delta)$.
\end{theorem}

\section{Converting LZ77 to Run Length BWT}\label{sec:algorithm}                                                         

In this section, we describe an algorithm that, given the LZ77 parsing
of a text $\T \in \Sigma^{n}$, computes its run-length compressed BWT
in $\bigO(z \polylog n)$ time. We start with an overview that explains
the key concepts. Next, we present two new data structures utilized in
our algorithm: the compressed string synchronizing set
(\cref{sec:csss}) and the compressed wavelet tree
(\cref{sec:cwt}). The conversion algorithm is then developed in
\cref{sec:algorithm-details}.

For any substring $Y$ of $\T^{\infty}$, we define $\lpos(Y) = \min\{i
\in [1 \dd n] : \T^{\infty}[i \dd i \,{+}\, |Y|) = Y\}$. If $Y$ is a
substring of $\revstr{\T}^{\infty}$, we define $\rpos(Y)$ by replacing
$\T^{\infty}$ in the definition with $\revstr{\T}^{\infty}$. We say
that a substring $Y$ of $\T^\infty$ is \emph{left-maximal} if there
exist distinct symbols $a, b \in \Sigma$ such that the strings $aY$
and $bY$ are also substrings of $\T^\infty$.  The following
definition, assuming $\Sigma \cap \mathbb{N} = \emptyset$, plays a key
role in our construction.

\begin{definition}[{\fontfamily{lmss}\selectfont BWT modulo $\ell$}{\fontfamily{lmr}\selectfont}]\label{def:bwtm}
  Let $\T \in \Sigma^{n}$, $\ell\ge 1$ be an integer, and $Y_i = \T^\infty[\SA[i]
  \dd \SA[i] + \ell)$ for $i \in [1 \dd n]$. We define the string
  $\BWT_{\ell} \in (\Sigma \cup \mathbb{N})^n$, called the \emph{BWT
  modulo $\ell$} (of $\T$), as follows. For $i \in [1 \dd n]$,
  \[
    \BWT_{\ell}[i] =
    \begin{dcases*}
      \lpos(Y_i)
        & if  $Y_i$ is left-maximal, \\
      \BWT[i]
        & otherwise.
    \end{dcases*}
  \]
\end{definition}

The algorithm runs in $k = \lceil\log n \rceil$ rounds. For $q \in [0
\dd k)$, the input to the $q$th round is ${\rm RL}(\BWT_{\ell})$,
where $\ell = 2^q$, and the output is ${\rm RL}(\BWT_{2\ell})$.  At
the end of the algorithm, we have ${\rm RL}(\BWT_{2^k}) = {\rm
RL}(\BWT)$ because $X \in \mathcal{S}_{2^k}$ is never left-maximal
for $2^k \ge n$.

Informally, in round $q$, we are given a (run-length compressed)
subsequence of $\BWT$ that can be determined based on sorting the
suffixes only up to their prefixes of length $2^q$.  $\BWT_{\ell}[b
\dd e] \in \Sigma^{+}$ implies $\BWT_{\ell+1}[b \dd e] \in
\Sigma^{+}$ (because a prefix of a left-maximal substring is
left-maximal). Hence, these subsequences need not be modified until
the end of the algorithm (except possibly merging their runs with
adjacent runs). For the remaining positions, $\BWT_{\ell}$ identifies
the (leftmost occurrences of) substrings to be inspected in the $q$th
round with the aim of replacing their corresponding runs in
$\BWT_{\ell}$ with previously unknown $\BWT$ symbols (as defined in
$\BWT_{2\ell}$).

We call a block $\BWT[b \dd e]$ \emph{uniform} if all symbols in
$\BWT[b \dd e]$ are equal, and \emph{non-uniform} otherwise.  The
following lemma ensures feasibility of the above construction.

\begin{lemma}\label{lm:bwtm-size}
  For any integer $\ell \geq 1$, it holds $|{\rm RL}(\BWT_{\ell})| < 2r$.
\end{lemma}
\begin{proof}
  Denote ${\rm RL}(\BWT_{\ell}) = ((c_1, \lambda_1), \ldots, (c_h,
  \lambda_h))$, letting $\lambda_0 = 0$.  By definition of
  $\BWT_{\ell}$, if $c_i \in \mathbb{N}$, then the block
  $\BWT(\lambda_{i-1} \dd \lambda_i]$ is non-uniform. Thus, there are
  at most $r - 1$ runs of symbols from $\mathbb{N}$ in $\BWT_{\ell}$.

  On the other hand, $c_i \in \Sigma$ and $c_j \in \Sigma$, with
  $i<j$, cannot both belong to the same run in $\BWT$. If this was
  true, then either $c_{i+1} \in \Sigma$ (which implies $c_{i+1} =
  c_i$, contradicting the definition of ${\rm RL}(\BWT_{\ell})$), or
  $c_{i+1} \in \mathbb{N}$, which is impossible since then
  $\BWT(\lambda_i \dd \lambda_{i + 1}]$ is non-uniform. Thus, there
  are at most $r$ runs of symbols from $\Sigma$ in $\BWT_{\ell}$.
\end{proof}

\subsection{Compressed String Synchronizing Sets}\label{sec:csss}

Our algorithm builds on the notion of \emph{string synchronizing
  sets}, recently introduced in~\cite{sss}.  Synchronizing sets are
one of the most powerful techniques for sampling suffixes. As
demonstrated in~\cite{phdtomek}, in the uncompressed setting, they are
the key in obtaining time-optimal solutions to many problems, and
their further applications are still being
discovered~\cite{circfactor, sss}.  In \mbox{this section}, we introduce a
notion of \emph{compressed string synchronizing sets}.  Our
construction is the first implementation of synchronizing sets in the
compressed setting and thus of independent interest.

We start with the definition of basic synchronizing sets.

\begin{definition}[{\fontfamily{lmss}\selectfont $\tau$-synchronizing
      set~\cite{sss}}{\fontfamily{lmr}\selectfont}]\label{def:sss}
  Let $T$ be a string of length $n$, and let $\tau \in [1\dd
    \lfloor\frac{n}{2}\rfloor]$. A set $\S \subseteq [1 \dd n - 2\tau
    + 1]$ is called a \emph{$\tau$-synchronizing set} of $T$ if it
  satisfies the following \emph{consistency} and \emph{density}
  conditions:
  \begin{enumerate}
  \item If $\T[i \dd i + 2\tau) = \T[j\dd j + 2\tau)$, then $i \in \S$
      if and only if $j \in \S$ (for $i, j \in [1 \dd n - 2\tau
      + 1]$),
  \item $\S\cap[i \dd i + \tau)=\emptyset$ if and only if $\per(\T[i
      \dd i + 3\tau - 2]) \leq \frac{1}{3} \tau$ (for $i \in [1 \dd n
      - 3\tau + 2]$).
  \end{enumerate}
\end{definition}

In most applications, we want to minimize $|\S|$. Observe that the
Thue--Morse sequence $\T_{\rm TM}$~\cite{thue1} does not contain any
\emph{cube} (substring of the form $W^3$). Thus, by density condition,
any synchronizing set $\S$ of the length-$n$ prefix of $\T_{\rm TM}$
satisfies $|\S| = \Omega \left( \frac{n}{\tau} \right )$ unless $n <
3\tau$.  Therefore, we cannot hope to achieve an upper
bound improving in the worst case upon the following one.

\begin{theorem}[\cite{sss}]\label{th:sss-existence-and-construction}
  For any string $T$ of length $n$ and parameter $\tau \in [1\dd
    \lfloor\frac{n}{2}\rfloor]$, there exists a $\tau$-synchronizing
  set $\S$ of size $|\S| = \bigO \left( \frac{n}{\tau}
  \right)$. Moreover, such $\S$ can be (deterministically) constructed
  in $\bigO(n)$ time.
\end{theorem}

Storing $\S$ for compressible strings presents the following
challenge: As shown in~\cite{MantaciRRRS19}, a length-$n$ prefix of
$\T_{\rm TM}$ satisfies $z = \bigO(\log n)$ and yet, as discussed
above, every $\tau$-synchronizing set of $\T_{\rm TM}$ satisfies $|\S|
= \Omega \left( \frac{n}{\tau} \right)$. Thus, although $|\S|$
\emph{can} be smaller than $\frac{n}{\tau}$, the assumption $z \ll n$
does not imply $|\S| \ll n$, preventing us from keeping plain $\S$
when $\tau = o(\tfrac{n}{z})$.

We thus exploit a different property of compressible strings: their
substrings $Y$ satisfy $\lpos(Y)\in \bigcup_{j=1}^z(e_j-|Y|\dd e_j]$,
where $e_j$ is the last position of the $j$th phrase in the LZ77
parsing of $\T$.  By consistency of $\S$, it suffices to store
$\bigcup_{j = 1}^{z}\S \cap (e_j {-} 2\tau \dd e_j]$. To check if
$i\in \S$, we then locate $i'=\lpos(T[i\dd i+2\tau))$ and check if $i'
\in \bigcup_{j = 1}^{z}\S \cap (e_j {-} 2\tau \dd e_j]$.  This
motivates the following (more general) definition.

\begin{definition}[{\fontfamily{lmss}\selectfont Compressed
      $\tau$-synchronizing set}{\fontfamily{lmr}\selectfont}]\label{def:csss}
  Let $\S$ be a $\tau$-synchronizing set of \mbox{string $\T[1\dd n]$} for
  some $\tau \in [1\dd \lfloor\frac{n}{2}\rfloor]$, and, for every
  $j\in[1 \dd z]$, let $e_j$ denote the last position of the $j$th
  phrase in the LZ77 parsing of $\T$. For $k\in\mathbb{N}_{\ge 2}$, we
  define the \emph{compressed representation} of $\S$ as
  \[
    \text{\rm comp}_k(\S) := \bigcup_{j = 1}^{z} \S \cap
    {\Big(} e_j {-} k\tau \dd e_j {+} k\tau {\Big]}.
  \]
\end{definition}

Next, we prove that every text $\T$ has a synchronizing set $\S$ with
a small compressed representation, and we show how to efficiently
compute such $\S$ from the LZ77 parsing of~$\T$.

\subsubsection{The Nonperiodic Case}\label{sec:csss-nonperiodic}

We initially assume that $\per(\T[i \dd i {+} \tau)) > \frac{1}{3}
\tau$ holds for all $i\in[1 \dd n - \tau + 1]$.

\begin{theorem}\label{th:csss-existence-nonperiodic}
  Let $\T$ be a string of length $n$ and let $\tau \in [1\dd
  \lfloor\frac{n}{2}\rfloor]$. Assume that $\per(\T[i \dd i + \tau))
  > \frac{1}{3}\tau$ holds for all $i \in [1 \dd n - \tau +
  1]$. Then, for every $k \in \mathbb{N}_{\ge 2}$, there exists a
  $\tau$-synchronizing set $\S$ of $\T$ with $\text{\rm comp}_k(\S)
  \leq 12kz$.
\end{theorem}

\newcommand{\id}{\mathrm{id}}

\begin{proof}
  Let $h:\Su_{\tau} \to [0,1]$ be a function mapping strings to real
  values in $[0,1]$ independently and uniformly at random.  Note that
  $h$ is collision-free almost surely (with probability 1).  Let us
  define $\id:[1\dd n-\tau+1]\to [0,1]$ with $\id(i) = h(T[i\dd
  i+\tau))$.  Observe that (almost surely) $\id$ is an
  \emph{identifier function}, that is, $\id(i)=\id(j)$ holds if and
  only if $\T[i \dd i {+} \tau) = \T[j \dd j {+}
  \tau)$. In~\cite[Lemma 8.2]{sss}, it is proved that then
  \[
    \S := \left \{ i \in [1 \dd n {-} 2\tau {+} 1] : \min \left\{
      \id(j) : j \in [i \dd i {+} \tau] \right\} \in \left\{
      \id(i), \id(i {+} \tau) \right\} \right\}
  \] 
  is a $\tau$-synchronizing set of $\T$.  Moreover, $\mathbb{E}
  \left[ |\S| \right] = \bigO \left(\frac{n}{\tau} \right)$.  To see
  this, observe that, for $j, j' \in [i \dd i + \tau]$, $\id(j) =
  \id(j')$ implies $|j' - j| > \frac{1}{3} \tau$ (otherwise, assuming
  $j < j'$, we have $\per(\T[j \dd j' + \tau)) \leq \frac{1}{3}
  \tau$, which contradicts $\per(\T[j \dd {j + \tau})) >
  \frac{1}{3} \tau$). Thus, $\T[i \dd i + 2\tau - 1]$ contains $d
  \geq \frac{\tau}{3}$ distinct length-$\tau$ substrings.  Since
  each has equal chance of having the smallest id, we have $\Pr
  \left[ i \in \S \right] \leq \frac{2}{d} \leq \frac{6}{\tau}$, and
  consequently, by linearity of expectation, $\mathbb{E} \left[ |\S|
  \right] \leq \frac{6n}{\tau}$. More generally, $\mathbb{E}
  \left[ |\S \cap (i \dd i + \ell]| \right] \leq \frac{6\ell}{\tau}$,
  and therefore
  \[
    \mathbb{E} {\big[} \left| {\rm comp}_k(\S) \right| {\big]} =
    \mathbb{E} \left[ \left| \bigcup_{j=1}^{z}\S \cap \left( e_j {-}
      k\tau \dd e_j {+} k\tau \right] \right| \right]
      \leq \sum_{j = 1}^{z} \mathbb{E} {\big[} \left| \S \cap \left(
      e_j {-} k\tau \dd e_j {+} k\tau \right] \right|
      {\big]} \leq 12kz.
  \]
  In particular, $\left| {\rm comp}_k(\S) \right|\le 12kz$ holds for
  some $h$.
\end{proof}

The above proof does not lead to an efficient algorithm for
constructing $\S$ as it relies on the random assignment of unique
names to \emph{all} substrings in $\Su_{\tau}$ and, since $|\Su_{\tau}| =
\Theta(z\tau)$ holds in the worst case, we cannot hope to achieve
$\bigO(z \polylog n)$ time this way. Next, we prove that
assigning unique names to all elements of $\Su_{\tau}$ is in
fact not necessary.

\begin{lemma}\label{lm:sss-whp}
  Let $\T$ be a string of length $n$ and let $\tau \in [1\dd
  \lfloor\frac{n}{2}\rfloor]$.  Assume that $\per(\T[i \dd i + \tau))
  > \frac{1}{3} \tau$ holds for all $i \in [1 \dd n - \tau + 1]$.
  For $i \in [1 \dd n - \tau + 1]$, let $\id(i) := h(\T[i \dd i +
  \tau))$, where $h: \Su_{\tau} \rightarrow [0,1]$ assigns
  independent and uniformly random values.

  If $\kappa = \max(1,\tfrac{\tau}{3c \ln n})$ for $c> 1$, then, with
  probability at least $1-n^{1-c}$, all positions $i\in [1\dd
    n-2\tau+1]$ satisfy $\min\{\id(j) : j \in [i \dd i + \tau]\}\le
  \frac{1}{\kappa}$.
\end{lemma}
\begin{proof}
  Recall from the proof of \cref{th:csss-existence-nonperiodic} that
  $\T[i \dd i + 2\tau-1]$ contains $d \geq \frac{\tau}{3}$ distinct
  length-$\tau$ substrings.  Since the values of $h$ are independent
  and uniformly distributed, we have $ \Pr\left[ \min\{\id(j) : j \in
    [i \dd i + \tau]\} > \tfrac{1}{\kappa} \right] =
  \left(1-\tfrac{1}{\kappa}\right)^d \le \exp(-\tfrac{d}{\kappa}) \le
  \exp(-\tfrac{\tau}{3\kappa}).  $ The probability above is trivially
  $0$ if $\kappa = 1$, so we can bound it by $n^{-c}$ if
  $\kappa=\max(1,\frac{\tau}{3c \ln n})$.  Taking the union bound
  across all positions $i$, we derive the final claim.
\end{proof}

If $\kappa$ is set as in the above lemma for a sufficiently large
constant $c$, then, with high probability, each window contains at
least one substring with a ``small'' identifier $\leq \frac1\kappa$.
The ``large'' identifiers of other substrings are never used in the
construction of the synchronizing set $\S$ and hence need not be
specified.  Consequently, to carry out the randomized construction of
$\S$ using \cref{th:csss-existence-nonperiodic}, rather than choosing
a random function $h: \Su_{\tau} \rightarrow [0,1]$, it suffices to
select a random subset $\Su_{\rm sample} \subseteq \Su_{\tau}$ with
rate $\frac{1}{\kappa}$ (each string in $\Su_{\tau}$ is included in
$\Su_{\rm sample}$ independently with probability $\frac{1}{\kappa}$)
and then construct a uniformly random function $h_{\rm sample}:
\Su_{\rm sample}\to [0,\frac{1}{\kappa}]$ (mapping strings in
$\Su_{\rm sample}$ to real values in $[0,\frac{1}{\kappa}]$
independently and uniformly at random).

Clearly, the element-wise sampling of $\Su_{\tau}$ is equivalent to
sampling the set $P_{\rm left}$ containing the starting positions of
the leftmost occurrences of strings in $\Su_{\tau}$.  Sampling $P_{\rm
left}$ directly is still hard, though. The key observation is that
instead of $P_{\rm left}$ (which is difficult to compute), we can
sample (at the same rate) elements of its superset $P_{\rm close} :=
\bigcup_{j = 1}^{z}(e_j\,{-}\,\tau \dd e_j]$, which is readily
available, and yet still sufficiently small.  Let $P_{\rm sample}'
\subseteq P_{\rm close}$ be a resulting sample. We then define the
desired sample with $P_{\rm sample} := P_{\rm sample}' \cap P_{\rm
left}$. Crucially, however, we have
\[
  \mathbb{E} \left[ |P_{\rm sample}'| \right] =
  \tfrac{1}{\kappa}|P_{\rm close}| \leq \tfrac{3c \ln n}{\tau} \cdot
  z\tau = \bigO(z \log n).
\]

To finish the construction, it suffices to pick a random function
$h_{\rm sample} : \Su_{\rm sample} \rightarrow [ 0,\frac{1}{\kappa}]$
to obtain $\id(i) := h_{\rm sample}(\T[i \dd i + \tau))$ (letting
$\id(i) = 1$ if $\T[i \dd i + \tau) \not\in \Su_{\rm sample}$).
Then, by \cref{lm:sss-whp} and the discussion above, using $h_{\rm
sample}$ is with high probability equivalent to using a uniformly
random function $h: \Su_{\tau} \rightarrow [0,1]$ during the
construction behind \cref{th:csss-existence-nonperiodic}.  Moreover,
we can also detect failures (that $\min \left\{ \id(j) : j \in [i \dd
i + \tau] \right\} = 1$ for some $i$), so the algorithm is Las-Vegas
randomized.

\begin{theorem}\label{th:csss-construction-nonperiodic}
  Let $\T$ be a string of length $n$ and let $\tau \in [1\dd
  \lfloor\frac{n}{2}\rfloor]$.  Assume that $\per(\T[i \dd i + \tau))
  > \frac{1}{3}\tau$ holds for all $i \in [1 \dd n - \tau +
  1]$. There exists a Las-Vegas randomized algorithm that, for any
  constant $k \in \mathbb{N}_{\ge 2}$, given the LZ77 parsing of $\T$,
  constructs in $\bigO(z \log^5 n)$ time a compressed representation
  $\text{\rm comp}_k(\S)$ of a $\tau$-synchronizing set $\S$ of $\T$
  satisfying $\text{\rm comp}_{k}(\S) \leq 24kz$.
\end{theorem}
\begin{proof}
  We first compute ${\rm comp}'(\S) := \bigcup_{j=1}^{z} \S \cap (e_j
  - 3\tau + 2 \dd e_j + \tau)$.  The algorithm consists of three
  steps.

  1. We start by computing the sample $P_{\rm sample} \subseteq P_{\rm
  left}$ for $\kappa = \tfrac{\tau}{3c \ln n}$. As discussed above,
  for this we first compute $P_{\rm sample}' \subseteq P_{\rm close}$
  using the same rate $\kappa$. Using the algorithm
  from~\cite{sampling,sampling2}, this takes $\bigO \left(
  \frac{1}{\kappa} |P_{\rm close}|+\log n\right) = \bigO(z \log n)$
  time with high probability.\footnote{The algorithm technically
  samples $[0\dd |P_{\rm close}|)$ rather than $P_{\rm close}$, but
  the desired subset $P_{\rm sample}'\sub P_{\rm close}$ is easy
  to obtain by exploiting the fact that $P_{\rm close}$ consists
  of at most $z$ contiguous integer ranges.}  We then discard
  every $i \in P_{\rm sample}' \setminus P_{\rm left}$. By
  \cref{th:index-leftmost-occ}, this takes $\bigO(\log^4 n)$ time
  per position. Overall, by taking \cref{th:index-leftmost-occ} into
  account, computing $P_{\rm sample} = P_{\rm sample}' \cap P_{\rm
  left}$, we spend $\bigO(z \log^5 n)$ time.

  2. Let $\Su_{\rm sample} = \{ \T[i \dd i + \tau) : i \in P_{\rm
  sample} \}$. By consistency of $\S$, whether $i \in \S$ or not,
  depends only on $\T[i \dd i + 2\tau)$. Thus, to determine ${\rm
  comp}'(\S)$ using the construction in
  \cref{th:csss-existence-nonperiodic}, it suffices to find all
  occurrences of strings from $\Su_{\rm sample}$ inside
  length-$(6\tau{-}4)$ substrings of $\T$ centered at the boundaries
  of LZ77 phrases.  Let $P_{\rm occ} = \left\{ i \in [1 \dd n] : \T[i
  \dd i + \tau) \in \Su_{\rm sample} \right\}$.  Formally, we
  compute
  \[
    P_{\rm nearocc} := P_{\rm occ} \cap \left( \bigcup_{j = 1}^{z}
        [e_j - 3\tau + 3 \dd e_j + 2\tau - 1] \right).
  \]

  As discussed above, for every $i \in P_{\rm left}$, we have
  $\Pr\left[ i \in P_{\rm sample} \right] = \frac1\kappa$, or
  equivalently,\linebreak $\Pr\left[ X \in \Su_{\rm sample} \right] =
  \frac1\kappa$ for $X \in \Su_{\tau}$.  Thus, $\mathbb{E} \left[
  P_{\rm occ} \cap [i \dd i + \ell) \right] = \frac{\ell}{\kappa}$
  for $i, \ell \in [1 \dd n]$ and hence
  \[
    \mathbb{E} {\Big [} | P_{\rm nearocc} | {\Big ]} \leq 
    \tfrac{z(5\tau-3)}{\kappa} = \bigO(z \log n).
  \]

  Let $\ell = 3\tau-2$, $e_0 = 0$, and $e_{z+1} = n + 1$. We call the
  $j$th, $j \in [1 \dd z]$, phrase $\T[e_{j-1} + 1 \dd e_j]$ in the
  LZ77 parsing of $\T$ \emph{short} if $e_{j}-e_{j-1} \leq 2\ell -
  \tau$, and \emph{long} otherwise. We define the sentinel phrases 0
  and $z+1$ to be long.  We create a text $\T_{\rm near} \in \left(
  \Sigma \cup \{ \# \} \right)^{*}$, where $\# \not\in \Sigma$, as
  follows.  Consider listing all long phrases left to right. Let
  $\T[e_{j - 1} + 1 \dd e_{j}]$, $j \in [1 \dd z + 1]$, be the current
  long phrase and let $e_{j'}$ be the last position of the preceding
  long phrase. Append $\T[\max(e_{j'} - \ell + 1, 1) \dd
  \min(e_{j-1}+\ell, n)]\#$ to $\T_{\rm near}$.  It is easy to check
  that there is a bijection between $P_{\rm nearocc}$ and occurrences
  of strings from $\Su_{\rm sample}$ in $\T_{\rm near}$, and every
  string from $\Su_{\rm sample}$ has at least one occurrence in
  $\T_{\rm near}$.

  To compute $P_{\rm nearocc}$, we observe that excluding all $\#$
  symbols, $\T_{\rm near}$ consists of not more than $2z$ substrings,
  each with an earlier occurrence in $\T_{\rm near}$. In the
  construction of $\T_{\rm near}$, these substrings are:
  $\T[\max(e_{j'}-\ell+1,1) \dd e_{j'}], \T[e_{j'} + 1 \dd \min(e_{j'}
  + \ell, e_{j' + 1})], \T[\min(e_{j'} + \ell, e_{j' + 1}) + 1 \dd
  e_{j' + 1}], \ldots, \T[e_{j - 1} + 1 \dd \min(e_{j - 1} + \ell,
  n)]$. Thus, after accounting for the $\#$ symbols, $\T_{\rm near}$
  has an LZ77-like parsing with at most $3z$ phrases. The phrase
  boundaries of this parsing can be obtained immediately from the LZ77
  parsing of $\T$. To guarantee that their sources are in $\T_{\rm
  near}$ we need to find their leftmost occurrences in $\T$. Using
  \cref{th:index-leftmost-occ} (applied to the LZ77 parsing of $\T$),
  this takes $\bigO(z \log^4 n)$ time.  Then, to compute $P_{\rm
  nearocc}$, we use the reporting index from
  \cref{thm:index-reporting} (constructed from the parsing of $\T_{\rm
  near}$). Starting positions of example occurrences of strings from
  $\Su_{\rm sample}$ can be easily mapped to $\T_{\rm near}$ via
  $P_{\rm sample}$. Therefore, computing $P_{\rm nearocc}$ takes
  $\bigO(z \log^4 n + |P_{\rm sample}| \log^3 n + |P_{\rm nearocc}|
  \log n) = \bigO(z \log^4 n)$ time in total.

  3. We now compute ${\rm comp}'(\S)$ from $P_{\rm nearocc}$. Assume
  $|P_{\rm nearocc}| = \bigO(z \log n)$.  We start by constructing
  $h_{\rm sample} : \Su_{\rm sample}\to [0,\frac1\kappa]$ that
  independently assigns uniformly random values. With high
  probability, $\Oh(\log n)$-bit precision is sufficient to guarantee
  that there are no ties.  We implicitly assign $h_{\rm sample}(X) =
  1$ for $X \not\in \Su_{\rm sample}$.  We keep a hash table mapping
  $i \in P_{\rm sample}$ to $h_{\rm sample}(\T[i \dd i +
  \tau))$. Moreover, with each $i \in P_{\rm nearocc}$ we store $i'
  \in P_{\rm sample}$ such that $\T[i \dd i + \tau) = \T[i' \dd i' +
  \tau)$. All $i'$ can be computed during the construction of $P_{\rm
  nearocc}$. Then, given $i \in P_{\rm nearocc}$, we obtain $h_{\rm
  sample}(\T[i \dd i + \tau))$ in $\bigO(1)$ time.
  
  Let $P_{\rm nearocc} = \{ p_1, \ldots, p_k \}$ denote its elements
  in ascending order.  Let $j \in [1 \dd z]$, and let range $[p_b,
  \ldots, p_e]$ be such that $P_{\rm nearocc} \cap [e_j {-} 3\tau
  {+} 3 \dd e_j {+} 2\tau {-} 1] = \{p_b, \ldots, p_e\}$.  Denote
  $\mathcal{I} = [e_j{-}3\tau{+}3 \dd e_j{+}\tau{-}1]$ and consider
  the computation of $\S \cap \mathcal{I}$ according to
  \cref{th:csss-existence-nonperiodic}. By \cref{lm:sss-whp}, with high
  probability, for every $i \in \mathcal{I}$, we have $ [i \dd i+\tau]
  \cap \{ p_b, \ldots, p_e \} \neq \emptyset $.  Assume this is the
  case. Our goal is to find all $i \in \mathcal{I}$, for which
  \[
    m_i := \min \{ h_{\rm sample}(\T[t \dd t + \tau)) : t \in [i \dd
        i + \tau] \cap \{ p_b, \ldots, p_e \} \}
  \]
  satisfies $m_i \in \{h_{\rm sample}(\T[i \dd i + \tau)), h_{\rm
  sample}(\T[i + \tau \dd i + 2\tau))\}$. This can only happen if
  $\{i, i+\tau\} \cap \{p_b, \ldots, p_e\} \neq \emptyset$ and hence
  it suffices to inspect $\leq 2(e-b+1)$ values of $i$. Using balanced
  BST to maintain $[i \dd i+\tau] \cap \{ p_b, \ldots, p_e \}$, the
  search can be implemented in $\bigO((e-b) \log n)$ time. It is not
  difficult to modify this approach so that in total for all $j \in [1
  \dd z]$, it takes $\bigO(|P_{\rm nearocc}| \log n) = \bigO(z
  \log^2 n)$ time. Note that during this algorithm we can detect
  whether for some $i\in \mathcal{I}$, $[i \dd i+\tau] \cap \{p_b,
  \ldots, p_e \} = \emptyset$. If this happens, we restart the
  algorithm.

  Let us return to the construction of ${\rm comp}_{k}(\S)$. Observe
  that $\S$ is entirely determined by ${\rm comp}'(\S)$. Furthermore,
  it allows computing $\S \cap [i \dd i + \tau)$ in $\bigO
  \left(\log^4 n + |\S \cap [i \dd i + \tau) | \right)$
  time. Namely, first locate the leftmost occurrence $\T[i_{\rm left}
  \dd i_{\rm left} + 3\tau - 1)$ of $\T[i \dd i + 3\tau - 1)$ using
  \cref{th:index-leftmost-occ}, and then copy the positions using
  the observation
  \[
    i+\Delta \in \left( \S \cap [i \dd i + \tau) \right)
    \text{ if and only if }
    i_{\rm left} + \Delta \in \left( \S \cap [i_{\rm left}
        \dd i_{\rm left} + \tau) \right).
  \]
  Thus, ${\rm comp}_k(\S)$ is easily computed in $\bigO(kz\log^4 n +
  |{\rm comp}_{k}(\S)|)$ time. During the construction, we keep track
  of the number of positions in ${\rm comp}_k(\S)$ and restart the
  algorithm if their number gets too large. This concludes the
  construction of ${\rm comp}_k(\S)$.
\end{proof}

\subsubsection{The General Case}\label{sec:csss-general}

Periodic fragments are handled similarly as in~\cite{sss}. This yields
the following two results, which constitute the main outcome of this
section.

\begin{theorem}\label{th:csss-existence}
  Let $\T$ be a string of length $n$ and let $\tau \in [1\dd
    \lfloor\frac{n}{2}\rfloor]$.  For any $k \in \mathbb{N}_{\ge 2}$,
  there exists a $\tau$-synchronizing set $\S$ of $\T$ satisfying
  $\text{\rm comp}_k(\S) \leq 36kz$.
\end{theorem}
\begin{proof}
  The key difference, compared to
  \cref{th:csss-existence-nonperiodic}, is that we can no longer
  assume that $\per(\T[i \dd i {+} \tau)) > \frac{1}{3}\tau$ holds for
  all $i$. Let us denote the set of positions that violate this
  assumption ${\sf Q} := \left \{ i \in [1 \dd n - \tau + 1] :
  \per(\T[i \dd i + \tau)) \leq \frac{1}{3}\tau \right \}$. We can
  then modify the construction as follows.  Letting again $\id : [1
  \dd n - \tau + 1] \rightarrow [0,1]$ be any identifier function,
  we now define:
  \[
    \S := \left \{ i \in [1 \dd n {-} 2\tau {+} 1] : \min \left\{ {\rm
      id}(j) : j \in [i \dd i {+} \tau] \setminus {\sf Q} \right\}
      \in \left\{ \id(i), \id(i {+} \tau) \right\} \right \}
  \]
  It is proved in~\cite{sss} that: (1) Such construction yields a
  correct synchronizing set of $\T$ (Lemma 8.2). (2) If ${\sf B} :=
  \{ i \in [1 \dd n - \tau + 1] \setminus {\sf Q}: \per(\T[i \dd
  i + \tau - 1)) \leq \frac{1}{3}\tau \text{ or } \per(\T[i + 1 \dd i
  + \tau)) \leq \frac{1}{3}\tau \}$ and $\mathcal{B} := \left\{
  \T[i \dd i + \tau) : i \in {\sf B} \right\}$, then $\id(i)
  := h(\T[i \dd i + \tau))$, where $h : \Su_{\tau}
  \rightarrow [0 \dd |\Su_{\tau}|)$ is uniformly random function
  such that $h(X) < h(Y)$ holds for all $X \in \mathcal{B}, Y \not \in
  \mathcal{B}$, satisfies $\mathbb{E}[|\S|] = \bigO \left (
  \frac{n}{\tau} \right )$.

  To see (2), let ${\sf B}_{\rm near} := \bigcup_{i \in {\sf B}} [i -
  \tau \dd i]$.  Then $\mathbb{E}[|\S|] = \mathbb{E}[|\S \cap {\sf
  B}_{\rm near}|] + \mathbb{E}[|\S \cap ([1 \dd n] \setminus {\sf
  B}_{\rm near})|]$.
  \begin{itemize}
  \item By the property of $h$, if $i \in \S \cap {\sf B}_{\rm near}$
    then $\{i, i + \tau \} \cap {\sf B} \neq \emptyset$, since then
    $\id$ achieves the minimum value on the position in ${\sf
    B}$. Moreover, since by periodicity lemma, for any $i$, we have
    $|[i \dd i + \left \lceil \frac{\tau}{3} \right \rceil) \cap {\sf
    B} | \leq 2$, then altogether we obtain $|\S \cap {\sf B}_{\rm
    near}| \leq 2|{\sf B}| \leq \frac{12n}{\tau}$.
  \item It is easy to see that if $[i \dd i + \tau] \cap {\sf Q} \neq
    \emptyset$ and $[i \dd i + \tau] \not\subseteq {\sf Q}$ then $[i
    \dd i + \tau] \cap {\sf B} \neq \emptyset$. Thus, by
    contraposition, if $i \in [1 \dd n] \setminus {\sf B}_{\rm near}$
    then, as in the analysis in
    \cref{th:csss-existence-nonperiodic}, $\Pr\left[ i \in \S \right]
    \leq \frac{6}{\tau}$.  Consequently, $\mathbb{E}[|\S \cap ([1 \dd
    n] \setminus {\sf B}_{\rm near})|] \leq \frac{6n}{\tau}$.
  \end{itemize}

  More generally, $|\S \cap (i \dd i + \ell] \cap {\sf B}_{\rm near}|
  \leq |(i \dd i + \ell] \cap {\sf B}| + |(i + \tau \dd i + \ell +
  \tau] \cap {\sf B}| \leq 4 \left\lceil \ell / \lceil \frac{\tau}{3}
  \rceil \right\rceil \leq 4 \left\lceil \frac{3\ell}{\tau}
  \right\rceil$ and $\mathbb{E} \left[ |\S \cap ((i \dd i + \ell]
  \setminus {\sf B}_{\rm near})| \right] \leq \frac{6\ell}{\tau}$.
  Thus, denoting $\mathcal{I}_j := (e_j {-} k \tau \dd e_j {+}
  k\tau]$, we have
  \begin{align*}
    \mathbb{E} {\big[} \left| {\rm comp}_k(\S) \right| {\big]}
      &= \mathbb{E} \left[ \left| \bigcup_{j=1}^{z}\S \cap
         \mathcal{I}_j \right| \right]
      \leq \sum_{j = 1}^{z} \mathbb{E} {\big[} \left| \S \cap
         \mathcal{I}_j \right| {\big]} \\
      & \leq \sum_{j = 1}^{z} \left| \S \cap
         \mathcal{I}_j \cap {\sf B}_{\rm near} \right| +
      \sum_{j = 1}^{z} \mathbb{E} {\big[} \left| \S \cap
         (\mathcal{I}_j \setminus {\sf B}_{\rm near}) \right|
         {\big]} \leq 36kz.
  \end{align*}
  In particular, $\left| {\rm comp}_k(\S) \right|\le 36kz$ holds for
  some $h$.
\end{proof}

\begin{theorem}\label{th:csss-construction}
  Let $\T$ be a string of length $n$ and let $\tau \in [1\dd
  \lfloor\frac{n}{2}\rfloor]$.  There exists a Las-Vegas randomized
  algorithm that, for any constant $k {\in} \mathbb{N}_{\ge 2}$, given
  the LZ77 parsing of $\T$, constructs in $\bigO(z \log^5 n)$ time a
  compressed representation $\text{\rm comp}_k(\S)$ of a
  $\tau$-synchronizing set $\S$ of $\T$ satisfying $|\text{\rm
  comp}_k(\S)| \leq 72kz$.
\end{theorem}
\begin{proof}
  The algorithm is a modified construction from
  \cref{th:csss-construction-nonperiodic}.  The problematic part of
  adapting the randomized construction is enforcing that our sampling
  of substrings is equivalent to using a uniformly random function $h
  : \Su_{\tau} \rightarrow [0,1]$ among those satisfying $h(X) < h(Y)$
  for all $X \in \mathcal{B}, Y \not \in \mathcal{B}$.

  The key observation is that the bound $|\S \cap (i \dd i + \ell]
  \cap {\sf B}_{\rm near}| \leq 4\left\lceil\frac{3\ell}{\tau}
  \right\rceil$ holds (for all $i$) in the \emph{worst case}, and not
  only in expectation. We can thus explicitly set $\S \cap {\sf
  B}_{\rm near} := \bigcup_{i \in {\sf B}} \{ i - \tau, i\} $, and
  then what remains is to determine $\S$ for positions $i \not \in
  {\sf B}_{\rm near}$. For all such $i$, as discussed above, we either
  have $[i \dd i + \tau] \cap {\sf Q} = \emptyset$ (in which case the
  randomized construction remains unchanged), or $[i \dd i + \tau]
  \subseteq {\sf Q}$ (in which case there is nothing to do). To
  implement this modification, we need to be able to efficiently
  represent and compute sets
  \begin{align*}
    {\sf Q}_{\rm nearocc} := {\sf Q} \cap \left( \cup_{j = 1}^{z}
        \mathcal{I}_j \right) \text{\ \ \ \ and\ \ \ \ }
    {\sf B}_{\rm nearocc} := {\sf B} \cap \left( \cup_{j = 1}^{z}
        \mathcal{I}_j \right),
  \end{align*}
  where $\mathcal{I}_j = [e_j - 3\tau + 2 \dd e_j + 2\tau - 1]$.

  Towards efficient representation, observe that if $i + 1 \in {\sf
  Q}$ and $i \not \in {\sf Q}$ then $i \in {\sf B}$. Analogously, if
  $i - 1 \in {\sf Q}$ and $i \not \in {\sf Q}$ then $i \in {\sf B}$.
  Thus, ${\sf B}$ forms a boundary between ${\sf Q}$ and $[1 \dd n]
  \setminus {\sf Q}$. We may also have $\{ i, i + 1 \} \subseteq {\sf
  B}$, i.e., the enclosed interval of positions is
  empty. Nevertheless, augmenting every $i \in {\sf B}_{\rm nearocc}$
  with a single bit yields a representation of $\bigcup\{{\sf Q} \cap
  \mathcal{I}_j : j \in [1 \dd z]\text{ and }{\sf B} \cap
  \mathcal{I}_j \neq \emptyset\} \subseteq {\sf Q}_{\rm nearocc}$. To
  represent remaining elements of ${\sf Q}_{\rm nearocc}$, it suffices
  to store a sequence of $z$ bits, with the $j$th bit indicating
  whether ${\sf Q}\cap \mathcal{I}_j$ is nonempty. It remains to see
  that since (as noticed in the proof of \cref{th:csss-existence}) for
  any $i$ we have $|[i \dd i + \left \lceil \frac{\tau}{3} \right
  \rceil) \cap {\sf B}| \leq 2$, the set ${\sf B}_{\rm nearocc}$
  satisfies $|{\sf B}_{\rm nearocc}| = \bigO(z)$. Thus, we obtain an
  $\bigO(z)$-space representation of ${\sf B}_{\rm nearocc}$ and ${\sf
  Q}_{\rm nearocc}$.  This also implies that we can efficiently
  store (as $\bigO(z)$ runs of consecutive positions) and
  random-access the set $P_{\rm close} \setminus ({\sf B}_{\rm
  nearocc} \cup {\sf Q}_{\rm nearocc})$, during the sampling phase.

  To finish the construction, it remains to show how to find ${\sf
  B}_{\rm nearocc}$. By~\cite[Lemma 8.8]{sss}, for any $i$,
  computation of $[i \dd i + b) \cap {\sf B}$, where $b \leq \left
  \lceil \frac{\tau}{3} \right \rceil$, can be reduced to two LCE
  queries and the computation of the shortest period of some
  substring.  More precisely, letting $X = \T[i + b \dd i + \tau -
  1)$, we first determine $p := \per(X)$. If $p > \frac{1}{3}\tau$
  we conclude $[i \dd i + b) \cap {\sf B} = \emptyset$. Otherwise, we
  compute the longest superstring $Y = \T[y \dd y']$ of $X$ that has
  period $p$ and is a substring of $\T[i \dd i + b + \tau - 1)$. If
  $|Y| \geq \tau - 1$ we have $[i \dd i + b) \cap {\sf B} = \{ y - 1,
  y' - \tau + 2\}$ and otherwise we again have $[i \dd i + b) \cap
  {\sf B} = \emptyset$. The set $[i \dd i + b) \cap {\sf Q}$ is
  deduced similarly.

  If $\tau > 3$, we set $b = \lfloor \frac{\tau-1}{3} \rfloor$ and use
  the above method to determine ${\sf B}_{\rm nearocc}$ and ${\sf
  Q}_{\rm nearocc}$.  To find $p$, we observe that we only need its
  exact value if $p \leq \lfloor \frac{1}{3} \tau \rfloor$.
  Otherwise, the knowledge that $p > \lfloor \frac{1}{3}\tau \rfloor$
  holds is sufficient. Since for our choice of $b$ it holds $\lfloor
  \frac{1}{3}\tau \rfloor \leq \frac{|X|}{2}$, we can reduce the
  computation of $p$ to the so-called \emph{2-period query} that given
  a string $X$ asks to return $\per(X)$ or to report that $X$ is not
  periodic, that is $\per(X) > \frac12 |X|$. Using \cref{thm:period},
  we can answer 2-period queries for substrings of $\T$ in
  $\bigO(\log^3 n)$ time after a $\bigO(z \log^2 n)$-time
  preprocessing of the LZ77-compressed $\T$.  Computing $y, y'$ on the
  other hand, is done using LCE queries in $\bigO(\log n)$ time after
  the $\bigO(z \log n)$-time preprocessing
  (\cref{th:index-lce}). Since $\bigcup_{j = 1}^{z}\mathcal{I}_j$ can
  be decomposed into $\bigO(z)$ length-$b$ intervals, overall we
  execute $\bigO(z)$ queries, and hence the computation of ${\sf
  B}_{\rm nearocc}$ and ${\sf Q}_{\rm nearocc}$ (including
  preprocessing) takes $\bigO(z \log^3 n)$ time.

  To handle $\tau \leq 3$, we observe that if $\tau \leq 2$, then
  ${\sf Q} = {\sf B} = \emptyset$. For $\tau = 3$, we compute ${\sf
  Q}_{\rm nearocc}$ and ${\sf B}_{\rm nearocc}$ by checking
  consecutive $i \in \bigcup_{j=1}^{z}\mathcal{I}_j$ using the
  definition of ${\sf Q}$ and ${\sf B}$.
\end{proof}

\subsection{Compressed Wavelet Trees}\label{sec:cwt}

Along with string synchronizing sets, wavelet trees~\cite{wt},
originally invented for text indexing, play a central role in our
algorithm. Unlike virtually all prior applications of wavelet trees,
ours uses a sequence of very long strings (up to $\Theta(n)$ symbols).
This approach is feasible since all strings are substrings of the
text, which is stored in the LZ77-compressed form. In this section, we
describe this novel variant of wavelet trees, dubbed here
\emph{compressed wavelet trees}. In particular, we prove the upper
bound on their size, describe an efficient construction from the
LZ77-compressed text, and show how to augment them to support some
fundamental queries.

Let $\Sigma$ be an alphabet of size $\sigma \geq 1$. Consider a string
$W[1 \dd m]$ over the alphabet $\Sigma^{\ell}$ so that $W$ is a
sequence of $m \geq 0$ strings of length $\ell \geq 0$ over the
alphabet $\Sigma$.  The wavelet tree of $W$ is defined as follows. Let
$\mathcal{T}$ be a perfect $\sigma$-ary rooted tree of height $\ell$
with edges labelled by symbols of $\Sigma$ such that, for every $Y \in
\Sigma^{\ell}$, there exists a root-to-leaf path in $\mathcal{T}$
whose edges are labelled $Y[1], \ldots, Y[\ell]$.  We define the label
of a node as the concatenation of the edge labels on the path from the
root. For $X \in \Sigma^d$, where $d \in [0 \dd \ell]$, by $v_X$ we
denote the node of $\mathcal{T}$ labelled $X$. We let $V(\mathcal{T})
= \bigcup_{d=0}^{\ell}\{v_X : X \in \Sigma^d\}$ denote the node set of
$\mathcal{T}$.

With each node $v_X \in V(\mathcal{T})$ we associate an increasing
sequence $I_X[1\dd h]$ of \emph{primary indices} such that \[\{I_X[i]
: i\in [1\dd h]\} = \{j\in [1\dd m] : W[j][1\dd |X|] = X\}.\] Based on
$I_X$, we define $B_X \in \Sigma^*$ such that, for $i\in [1\dd h]$,
\[
  B_X[i] = W[I_X[i]][|X| + 1],
\]
if $|X|<\ell$ and $B_X= \eps$ if $|X|=\ell$. In other words, $B_X$ is
a string containing the symbol at position $|X|+1$ for each string of
$W$ that is prefixed by $X$. Importantly, the symbols in $B_X$ occur
in the same order as these strings occur in $W$.

As typically done in the applications of wavelet trees, we only
explicitly store the strings $B_X$. The values of primary indices
$I_X$ are retrieved using additional data structures, based on the
following observation.

\begin{lemma}[\cite{wt}]\label{lm:wt}
  Let $X \in \Sigma^d$, where $d \in [0 \dd \ell)$. 
  For every  $c \in \Sigma$ and $j \in [1 \dd |I_{Xc}|]$, we have
  $I_{Xc}[j]=I_X[i]$, where $B_X[i]$ is the $j$th occurrence
  of $c$ in $B_X$.
\end{lemma}

We define the \emph{compressed wavelet tree} $\mathcal{T}_c$ of $W$ as
the wavelet tree of $W$ in which all strings $B_X$ have been
run-length compressed and, with the exception of $\{ v_{\eps} \} \cup
\{v_{W[i]}\}_{i=1}^{m}$, all nodes $v_X$ satisfying $|{\rm RL}(B_X)|
\leq 1$, have been removed (the unary paths are collapsed into single
edges). The shape and edge labels of the resulting tree are identical
to the compact trie of strings $W[1], \ldots, W[m]$.

We store edge labels of $\mathcal{T}_c$ as pointers to substrings in
$W$. We assume that values of $\ell$ and $m$ fit into a single machine
word so that each edge of $\mathcal{T}_c$ and each element of ${\rm
RL}(B_X)$ can be encoded in $\bigO(1)$ space. Since $|{\rm RL}(B_Y)|
\geq 1$ holds for every internal node $v_Y \in V(\mathcal{T}_c)$, and
unless $|V(\mathcal{T}_c)| = 1$, each leaf $v_Z$ in $\mathcal{T}_c$
can be injectively mapped to an element of ${\rm RL}(B_{Z'})$ for the
parent $v_{Z'}$ of $v_Z$, the space to store $\mathcal{T}_c$ is
dominated by the run-length compressed strings $B_X$, i.e.,
$\mathcal{T}_c$ needs $\bigO(1 + \sum_{v_X \in V(\mathcal{T}_c)}|{\rm
RL}(B_X)|)$ space.

\begin{theorem}\label{th:cwt-size}
  Let $W$ be a non-empty sequence of equal-length strings and let
  $\mathcal{T}_c$ be its compressed wavelet tree. Then, $\sum_{v_X \in
    V(\mathcal{T}_c)}|{\rm RL}(B_X)| = \bigO(1 + |{\rm RL}(W)| \log
  |{\rm RL}(W)|)$.
\end{theorem}
\begin{proof}
  Let $m = |W|$, $k = |{\rm RL}(W)| \leq m$, and $k' = |\{ W[i] : i
  \in [1 \dd m] \}| \leq k$. Due to $|V(\mathcal{T}_c)| \leq 2k' =
  \bigO(k)$, we can focus on nodes $v_X \in V(\mathcal{T}_c)$ such
  that $|{\rm RL}(B_X)| \geq 2$.

  The proof resembles that of \cref{lm:irreduc-sum-2}.  With each $X
  \in \Sigma^{*}$ such that $|{\rm RL}(B_X)| \geq 2$, we associate
  $|{\rm RL}(B_X)|-1$ units of cost and charge them to individual
  elements of~$W$.  We then show that each run in ${\rm RL}(W)$ is in
  total charged at most $2\log k'$ units of cost. Consequently,
  \[
    \sum_{\substack{v_X \in
      V(\mathcal{T}_c)\\|{\rm RL}(B_X)|
      \geq 2}}|{\rm RL}(B_X)| \leq 4k \log k' =
    \bigO(k \log k).
  \]
  
  Consider $X \in \Sigma^d$ with $|{\rm RL}(B_X)| \geq 2$; note that
  ${d < \ell}$.  Let ${\rm RL}(B_X) = ((c_1, \lambda_1), \ldots, (c_h,
  \lambda_h))$.  Observe that if we let $p_0 = I_X[\lambda_i]$ and
  $p_1 = I_X[\lambda_i + 1]$ for some $i \in [1 \dd h)$, then
  $W[p_0][d + 1] = c_i \neq c_{i+1} = W[p_1][d + 1]$.  Moreover,
  $B_X[\lambda_i]\ne B_X[\lambda_{i}+1]$ implies $W[p_0+1]\ne W[p_0]$
  and $W[p_1-1]\ne W[p_1]$.  The $i$th unit of cost is charged to
  $W[p_{t}]$, where $t \in \{0, 1\}$ is chosen depending on the sizes
  of subtrees of $\mathcal{T}_c$ rooted at the children of $v_X$, so
  that the subtree containing $v_{W[p_t]}$ has at most as many leaves
  as the subtree containing $v_{W[p_{1-t}]}$.
  
  Now, consider a run $W[b\dd b'] = Y^{\delta}$ in ${\rm RL}(W)$.  For
  a single depth $d$, the run could be charged at most twice, with at
  most one unit assigned to $W[b]$ due to $p_1=b$ and at most one unit
  assigned to $W[b']$ due to $p_0=b'$, both for $X=Y[1\dd d]$.
  Moreover, note that the subtree size on the path from $v_Y$ to the
  root $v_{\eps}$ of $\mathcal{T}_c$ doubles for every depth $d$ for
  which the run was charged.  Thus, the total charge of the run is at
  most $2\log k'$ units.
\end{proof}\vspace{1ex}

Let $W[1 \dd m]$ be a sequence of substrings of $\T^{\infty}$ of the
same length $\ell$. Observe that if we have access to $\T$, then the
sequence $W$ can be compactly encoded in $\bigO(1 + |{\rm RL}(W)|)$
space. Namely, it suffices to store the length $\ell$ and the sequence
${\rm RL}((\lpos(W[i]))_{i \in [1 \dd m]})$.
The following theorem shows that given such compact encoding of $W$
and the LZ77 parsing of $\T$, the compressed wavelet tree of $W$ can
be constructed efficiently.

\begin{theorem}\label{th:cwt-construction}
  Given the LZ77 parsing of text $\T[1 \dd n]$, and a sequence $W[1
    \dd m]$ of $m \leq n$ substrings of $\T^{\infty}$ of length $\ell
  \leq n$, represented as ${\rm RL}((\lpos(W[i]))_{i \in [1 \dd m]})$,
  the compressed wavelet tree of $W$ can be constructed in $\bigO((z +
  |{\rm RL}(W)|) \log^2 n)$ time.
\end{theorem}
\begin{proof}
  We start by constructing the compact trie $\mathcal{T}_c$ of $W$.
  Let $k=|{\rm RL}(W)| \leq m$, $\mathcal{W} = \{W[i]: i \in [1 \dd
  m]\}$ and $k' = |\mathcal{W}| \leq k$. We first construct
  $\mathsf{P} = \{\lpos(W[i]) : i \in [1 \dd m]\}$ to represent
  $\mathcal{W}$. Given ${\rm RL}((\lpos(W[i]))_{i \in [1 \dd m]})$,
  $\mathsf{P}$ is easy to compute in $\bigO(k \log k)$ time.  We now
  observe that an LCE query on $\T$ suffices to determine the
  lexicographical order between any substrings of $\T^{\infty}$. We
  construct the data structure for LCE queries on $\T$ using
  \cref{th:index-lce} in $\bigO(z \log^2 n)$ time.  We then
  lexicographically sort all length-$\ell$ substrings of $\T^{\infty}$
  starting at positions in $\mathsf{P}$ in $\bigO(k' \log k' \log n) =
  \bigO(k \log^2 n)$ time (e.g., using mergesort).  Finally, we
  construct $\mathcal{T}_c$ by inserting elements of $\mathcal{W}$ in
  the order given by $\mathsf{P}$. We maintain the stack containing
  the internal nodes on the rightmost path, with the deepest node on
  top.  Adding each string first removes some elements from the stack,
  and then adds at most one new element. The total number of steps is
  thus $\bigO(k)$.

  In the second step of the algorithm, we compute ${\rm RL}(B_X)$ for
  all $v_X \in V(\mathcal{T}_c)$. Let $\Sigma = \{c_0, \ldots,
  c_{\sigma-1}\}$ be the alphabet of $\T$.  If $k \leq 1$ or $\ell =
  0$, then $|V(\mathcal{T}_c)| \leq 2$ and the step takes $\bigO(1)$
  time.  Let us thus assume $k\geq 2$ and $\ell \geq 1$. This implies
  $\sigma \geq 2$.  Denote ${\rm RL}(W) = ((R_1, \lambda_1), \ldots,
  (R_k, \lambda_k))$, letting $\lambda_0 = 0$. Let also $\delta_i =
  \lambda_i - \lambda_{i-1}$ for $i \in [1 \dd k]$.  For any $i \in [1
  \dd |I_X|]$, let $J_X[i]$ be the index $j \in [1 \dd k]$
  satisfying $I_X[i] \in [\lambda_{j-1}+1 \dd \lambda_j]$. It holds
  \[
    W[I_X[i]] = R_{J_X[i]}.
  \]
  Recall, that for $|X|<\ell$, $B_X[i] = W[I_X[i]][|X| + 1] =
  R_{J_X[i]}[|X| + 1]$. Thus, when computing $B_X$, we can use the
  sequence $J_X$ instead of $I_X$.  Note that letting $J_X[1 \dd
  |J_X|] = p_1^{t_1}p_2^{t_2} \ldots p_h^{t_h}$, where $p_i \neq
  p_{i+1}$ for $i \in [1 \dd h)$, it holds $t_j = \delta_{p_j}$ for $j
  \in [1 \dd h]$.  Moreover, $p_1 < \ldots < p_h$. We can thus store
  each $J_X$ simply as a set $J'_X = \{p_1, \ldots, p_h\}$.

  We process all $v_X \in V(\mathcal{T}_c)$ in the order of
  nonincreasing $|X|$. During the algorithm, we maintain sets $J'_X$
  for a subset of nodes $v_X \in V(\mathcal{T}_c)$. More precisely, at
  any point, we keep a data structure representing the set $J'_X$ in
  sorted order for (only) the highest processed node on each
  leaf-to-root path.

  Let us first assume $\sigma = 2$.  To start, observe that the trie
  construction is easily augmented to compute, for any $i \in [1 \dd
  k]$, the pointer to $v_{R_i}$. Thus, to initialize $J'_X$ for all
  leafs $v_X$, iterate through ${\rm RL}(W)$ and add $i$ to
  $J'_{R_i}$.  Consider now $v_X \in V(\mathcal{T}_c)$ that is not a
  leaf. If $v_X$ has one child $v_{XY}$, then it holds $X = \eps$ and
  $B_{\eps} = c^m$ where $c$ is the first symbol in $Y$.  Thus, ${\rm
  RL}(B_X)$ is easy to compute.  Assume thus that $v_{XY_0}$ and
  $v_{XY_1}$ are children of $v_X$, and for $i \in \{0, 1\}$, $Y_i$
  starts with $c_i \in \Sigma$. Clearly, $J'_X = J'_{XY_0} \cup
  J'_{XY_1}$. Moreover, the sorted set $J'_{X}$ is obtained by merging
  $|{\rm RL}(B_X)|$ contiguous subsequences from either $J'_{XY_0}$ or
  $J'_{XY_1}$ (both regarded as sorted sequences). Suppose that during
  the $i$th step of the merging, we append to $J'_X$ the elements $\{
  p_b, \ldots, p_e \} \subseteq J'_{XY_h}$, where $p_b < \ldots < p_e$
  and $h \in \{0, 1\}$. Note, that this implies that the $i$th run in
  $B_X$ is $c_h^{\delta}$, where
  \[
    \delta = \sum_{p \in J'_{XY_h} \cap [p_b \dd p_e]}
      \delta_p.
  \]

  To implement the merging efficiently, each set $J'_X$ is represented
  using an AVL tree augmented (using standard techniques, i.e., each
  node stores the sum of elements in its subtree) to support
  computing, for any $p', p''$ the value of $\sum_{p \in J'_{X} \cap
  [p' \dd p'']} \delta_p$ in $\bigO(\log k)$ time. Since AVL trees
  also support split/join in $\bigO(\log k)$ time, each subsequence
  $\{p_b, \ldots, p_e\}$ can be appended to $J'_X$ in $\bigO(\log k)$
  time. During each append, we compute the next element of ${\rm
  RL}(B_X)$. Consequently, the AVL tree for $J'_X$, together with
  ${\rm RL}(B_X)$, can be computed in $\bigO(|{\rm RL}(B_X)| \log k)$
  time. By \cref{th:cwt-size}, this yield the claim.

  To generalize the merging to any $\sigma \geq 2$, during the
  computation of $J'_X$, we maintain a priority queue containing, for
  $h \in [0 \dd \sigma)$, the smallest unextracted element of
  $J'_{XY_h}$. After extracting (and removing) the minimum element
  $p'$ corresponding to $h \in [0 \dd \sigma)$ from the queue, the new
  minimum element $p''$ specifies the range $[p' \dd p'']$ of elements
  to be extracted from $J'_{XY_h}$. The queue contains at most
  $\sigma$ elements. We perform $\bigO(1)$ queue operations for each
  run. Thus, the complexity of the construction is not affected.
\end{proof}\vspace{1ex}

The above Theorem holds also if $W$ is the sequence of substrings of
$\revstr{T}^{\infty}$.

\begin{theorem}\label{th:cwt-construction-rev}
  Given the LZ77 parsing of text $\T[1 \dd n]$, and a sequence $W[1
  \dd m]$ of $m \leq n$ substrings of $\revstr{\T}^{\infty}$ of
  length $\ell \leq n$, represented as ${\rm RL}((\rpos(W[i]))_{i \in
  [1 \dd m]})$, the compressed wavelet tree of $W$ can be
  constructed in $\bigO((z + |{\rm RL}(W)|) \log^2 n)$ time.
\end{theorem}
\begin{proof}
  It suffices to observe that \cref{th:index-lce} supports also LCE
  queries on $\revstr{\T}$. Thus, the construction is identical as in
  the proof of \cref{th:cwt-construction}.
\end{proof}\vspace{1ex}

\paragraph{Computing Primary Indices}\label{sec:pi}

The key operation that we want to support on $\mathcal{T}_c$ is, given
a pointer to $v_X \in V(\mathcal{T}_c)$ and an integer $q \in [1 \dd
|I_{X}|]$, compute the value $I_X[q]$.

Let us first consider a simpler problem. Given a pointer to $v_{X} \in
V(\mathcal{T}_c)$ different from the root, and $q \in [1 \dd
|I_{X}|]$, compute $q'$ such that $I_X[q] = I_{X_{\rm p}}[q']$,
where $v_{X_{\rm p}}$ is the parent of $v_X$.

\begin{proposition}\label{pr:pi-ds-1}
  Let $\mathcal{T}_c$ be the compressed wavelet tree of $W[1 \dd m]$.
  There exists a data structure of size $\bigO(1 + |{\rm RL}(B_X)|)$
  that, given a pointer to $v_X \in V(\mathcal{T}_c)$ and an integer
  $q \in [1 \dd |I_X|]$, in $\bigO(\log m)$ time returns $q'$ such
  that $I_X[q] = I_{X_{\rm p}}[q']$, where $v_{X_{\rm p}}$ is the
  parent of $v_X$.
\end{proposition}
\begin{proof}
  Observe that by \cref{lm:wt}, $q'$ is the position of the $q$th
  occurrence of $X[|X_{\rm p}|+1]$ in $B_{X'}$.

  Denote ${\rm RL}(B_Y) = ((c_1, \lambda_1), \ldots,
  (c_{k},\lambda_{k}))$, where $v_Y \in V(\mathcal{T}_c)$. For any $i
  \in [1 \dd k]$, define $n_i = |\{i \in [1 \dd \lambda_i] : B_Y[i] =
  c_i\}|$. For $c \in \Sigma$, let $L_{Y, c}$ be the sequence
  containing all pairs in $\{(\lambda_i, n_i) : i \in [1 \dd k]\text{
  and }c_i = c \}$, sorted by $\lambda_i$.
  
  Observe, that for any $q \geq 1$, given $L_{Y,c}$, the position of
  the $q$th occurrence of $c$ in $B_Y$, if it exists, can be computed
  in $\bigO(\log m)$ time using binary search. Thus, the data
  structure consists simply of sequences $L_{X,c}$ for all pairs
  $(X,c)$, such that $c$ occurs in $B_X$.  The total length of all
  sequences is $\bigO(|{\rm RL}(B_X)|)$.  All sequences are
  concatenated and stored in one array, and we store a balanced BST
  containing the size and a pointer to the beginning of $L_{Y,c}$, for
  every $c$ occurring in $B_Y$.
\end{proof}\vspace{1ex}

Let us now consider a more general problem. Let $X_{\rm up}, X_{\rm
down} \in \Sigma^*$ be such that $X_{\rm up}$ is a prefix of $X_{\rm
down}$, and $v_{X_{\rm up}}, v_{X_{\rm down}} \in V(\mathcal{T}_c)$.
Let $\mathcal{X} \subseteq \Sigma^*$ be the set of labels of nodes of
$\mathcal{T}_c$ on the path connecting (and including) $v_{X_{\rm
up}}$ and $v_{X_{\rm down}}$. Given a pointer to any node $v_{X}$
of $\mathcal{T}_c$ such that $X \in \mathcal{X}$, and an integer $q
\in [1 \dd |I_X|]$, we aim to compute $q'$ such that $I_X[q] =
I_{X_{\rm up}}[q']$.

Let $s = |I_{X_{\rm up}}|$ and
\[
  \ell_i = \lcp \left( W[I_{X_{\rm up}}[i]], X_{\rm down}\right),
\]
where $i \in [1 \dd s]$. Consider $\mathcal{P} = \{(i,\ell_i)\}_{i \in
[1 \dd s]}$ as a set of points on a plane. Let us fix some $X \in
\mathcal{X}$, and let $m_i = |\{j \leq i : \ell_j \geq |X|\}|$ for $i
\in [1 \dd s]$. The value of $m_i$ is the number of points of
$\mathcal{P}$ inside the upper left ``quadrant'' defined by $i$ and
$|X|$. Clearly, $m_1 \leq m_2 \leq \ldots \leq m_s$.

\begin{lemma}\label{lm:qq}
  Let $q' = \min\{i \in [1 \dd s] : m_i \geq q\}$.  Then, $I_X[q] =
  I_{X_{\rm up}}[q']$.
\end{lemma}
\begin{proof}
  By $X \in \mathcal{X}$, the string $X_{\rm up}$ is a prefix of
  $X$. By \cref{lm:wt}, $I_X$ is a subsequence of $I_{X_{\rm up}}$,
  and $q'$ is the position of the $q$th index $i \in [1 \dd s]$
  satisfying $\lcp(W[I_{X_{\rm up}}[i]], X) \geq |X|$. Since $X$ is a
  prefix of $X_{\rm down}$, we obtain that for $i \in [1 \dd s]$,
  $\lcp(W[I_{X_{\rm up}}[i]], X) \geq |X|$ if and only if $\ell_i \geq
  |X|$. Therefore, $q'$ is the $q$th index $i \in [1 \dd s]$ for which
  it holds $\ell_i \geq |X|$, or equivalently, the smallest $i \in [1
  \dd s]$ satisfying $m_i \geq q$.
\end{proof}\vspace{1ex}

It thus suffices to store a data structure answering orthogonal range
counting queries on the set of points $\mathcal{P}$. Then, by the
above lemma, the value $q'$ can be found using binary search. Using
the data structure of~\cite{chazelle} for range queries, we obtain the
solution to our problem using $\bigO(s)$ space and answering queries
in $\bigO(\log^2 s)$ time. To reduce the space usage, we observe that
the sequence $(\ell_i)_{i \in [1 \dd s]}$ is compressible.

\begin{lemma}\label{lm:ell_i}
  \[
    \left| {\rm RL} \left( \left( \ell_i \right)_{i \in [1 \dd s]}
      \right) \right| \leq
      1 + \sum_{X \in \mathcal{X} \setminus \{X_{\rm down}\}}
      \left| {\rm RL}(B_X) \right|.
  \]
\end{lemma}
\begin{proof}
  The proof is by induction on $|\mathcal{X}|$. If $|\mathcal{X}| =
  1$, then $X_{\rm up} = X_{\rm down}$ and hence $\ell_i = |X_{\rm
  down}|$ holds for all $i \in [1 \dd s]$. Consequently $|{\rm
  RL}((\ell_i)_{i \in [1 \dd s]})| = 1$.

  Let us thus assume $|\mathcal{X}| \geq 2$. Let $k$ denote the number
  of $i \in [1 \dd s)$ satisfying $\ell_i \neq \ell_{i+1}$.  Let $i
  \in [1 \dd s)$ be any such position. Let $X'_{\rm up}$ be the
  shortest string in $\mathcal{X}' = \mathcal{X} \setminus \{X_{\rm
  up}\}$, and denote $s' = |I_{X'_{\rm up}}|$, and $\ell'_j =
  \lcp(W[I_{X'_{\rm up}}[j]], X_{\rm down})$, for $j \in [1 \dd s']$.
  Consider two cases:
  \begin{enumerate}
  \item $B_{X_{\rm up}}[i] = B_{X_{\rm up}}[i + 1]$. Denote $c =
    X_{\rm down}[|X_{\rm up}| + 1]$. We observe that it must hold
    $B_{X_{\rm up}}[i] = c$, since otherwise $\ell_i =
    \lcp(W[I_{X_{\rm up}}[i]], X_{\rm down}) = |X_{\rm up}|$ and
    analogously $\ell_{i+1} = |X_{\rm up}|$, contradicting $\ell_i
    \neq \ell_{i+1}$.  Let $i' = |\{j \leq i : B_{X_{\rm up}}[j] =
    c\}|$. By \cref{lm:wt}, we have $I_{X'_{\rm up}}[i'] = I_{X_{\rm
    up}}[i]$ and $I_{X'_{\rm up}}[i' + 1] = I_{X_{\rm up}}[i +
    1]$. Thus, we have $\ell'_{i'} \neq \ell'_{i'+1}$. It is easy to
    see that this mapping of $i$ to $i'$ is injective. Thus, this case
    can happen at most $|{\rm RL}((\ell'_j)_{j \in [1 \dd s']})| - 1$
    times.  By the inductive assumption, it holds $|{\rm
    RL}((\ell'_j)_{j \in [1 \dd s']})| - 1 \leq \sum_{X \in
    \mathcal{X}' \setminus \{X_{\rm down}\}}|{\rm RL}(B_X)|$.
  \item $B_{X_{\rm up}}[i] \neq B_{X_{\rm up}}[i + 1]$. This case can
    happen at most $|{\rm RL}(B_{X_{\rm up}})| - 1$ times.
  \end{enumerate}
  Thus, $k \leq |{\rm RL}(B_{X_{\rm up}})| - 1 + \sum_{X \in
  \mathcal{X}' \setminus \{X_{\rm down}\}}|{\rm RL}(B_X)| < \sum_{X
  \in \mathcal{X} \setminus \{X_{\rm down}\}}|{\rm RL}(B_X)|$.  It
  remains to note that $|{\rm RL}((\ell_i)_{i \in [1 \dd s]})| = 1 +
  k$.
\end{proof}\vspace{1ex}

Denote ${\rm RL}((\ell_i)_{i \in [1 \dd s]}) = ((t_1, \lambda_1),
\ldots, (t_{s'}, \lambda_{s'}))$, letting $\lambda_0 = 0$, and
$\delta_i = \lambda_i - \lambda_{i - 1}$ for $i \in [1 \dd s']$. If we
now let $m'_i = \sum_{j \leq i,\ t_j \geq |X|}\delta_j$ for $i \in [1
\dd s']$, then $m'_1 \leq \ldots \leq m'_{s'}$ and \cref{lm:qq}
generalizes as follows.

\begin{lemma}\label{lm:qq2}
  Let $q'' = \min\{i \in [1 \dd s'] : m'_i \geq q\}$ and $q' =
  \lambda_{q''-1} - m'_{q''-1} + q$. Then, $I_X[q] = I_{X_{\rm
  up}}[q']$.
\end{lemma}
\begin{proof}
  As shown in the proof of \cref{lm:qq}, it suffices to prove that
  $q'$ is the $q$th index $i \in [1 \dd s]$ satisfying $\ell_i \geq
  |X|$. It is easy to see, by definition of $(m'_j)_{j \in [1 \dd
  s']}$, that such $i$ satisfies $i \in (\lambda_{q'' - 1} \dd
  \lambda_{q''}]$, where $q''$ is the smallest index $i \in [1 \dd
  s']$ for which it holds $m'_i \geq q$. More precisely, it is the
  $(q - m'_{q'' - 1})$th position inside this interval. This yields
  $q' = \lambda_{q'' - 1} - m'_{q'' - 1} + q$, as claimed.
\end{proof}\vspace{1ex}

Thus, rather than $s$ points in the range counting structure, we can
instead store a set of only $s' \leq 1 + \sum_{X \in \mathcal{X}
\setminus \{X_{\rm down} \}} |{\rm RL}(B_X)|$ \emph{weighted} points
$\mathcal{P}' = \{ (i, t_i) \}_{i \in [1 \dd s']}$, where the weight
of the $i$th point is $\delta_i$. The range counting query is now
replaced with the query that returns the total weight of points in a
given range. Using~\cite{chazelle}, the data structure takes $\bigO(s'
\log s')$ space and answers range sum queries in $\bigO(\log^2 s')$
time. Accounting for the binary search, we obtain the solution to our
problem running in $\bigO(\log^3 s')$ time per query. We have thus
proved the following result.

\begin{proposition}\label{pr:pi-ds-2}
  Let $\mathcal{T}_c$ be the compressed wavelet tree of $W[1 \dd m]$.
  Let $v_{X_{\rm up}}, v_{X_{\rm down}} \in V(\mathcal{T}_c)$ be such
  that $X_{\rm up}$ is a prefix of $X_{\rm down}$. Let $\mathcal{X}$
  be the set of labels of nodes on the path between $v_{X_{\rm up}}$
  and $v_{X_{\rm down}}$. There exists a data structure of size
  $\bigO(1 + \sum_{X \in \mathcal{X}} |{\rm RL}(B_X)| \log m)$ that,
  given a pointer to any $v_{X} \in V(\mathcal{T}_c)$ with $X \in
  \mathcal{X}$, and $q \in [1 \dd |I_X|]$, returns in $\bigO(\log^3
  m)$ time $q'$ such that $I_X[q] = I_{X_{\rm up}}[q']$.
\end{proposition}

We now show how to combine the above data structures to solve the
general problem of computing primary index queries.

\begin{theorem}\label{th:pi-ds}
  Let $\mathcal{T}_c$ be the compressed wavelet tree of $W[1 \dd m]$.
  There exists a data structure of size $\bigO(1 + |{\rm RL}(W)|
  \log^2 m)$ that, given a pointer to $v_X \in V(\mathcal{T}_c)$ and
  an integer $q \in [1 \dd |I_{X}|]$, in $\bigO(\log^4 m)$ time
  returns $I_X[q]$.
\end{theorem}
\begin{proof}
  We apply the \emph{heavy path
  decomposition}~\cite{SleatorT83,HarelT84} to $\mathcal{T}_c$. Let
  $u$ be a node of $\mathcal{T}_c$ other than its root. An edge
  connecting $u$ to its parent is called \emph{light} if $u$ has a
  sibling $u'$ satisfying ${\rm size}(u) \leq {\rm size}(u')$;
  otherwise, the edge is called \emph{heavy} (see also
  \cref{lm:irreduc-sum-3}). Then, every node has at most one child
  connected via a heavy edge, and there is at most $\log m$ light
  edges on any leaf-to-root path.
  
  Let $v_{X_{\rm up}}, v_{X_{\rm down}} \in V(\mathcal{T}_c)$ be such
  that $X_{\rm up}$ is a proper prefix of $X_{\rm down}$. A path
  between $v_{X_{\rm up}}$ and $v_{X_{\rm down}}$ is called
  \emph{heavy} if all edges on the path are heavy, and each of
  $v_{X_{\rm up}}, v_{X_{\rm down}}$ is incident with exactly one
  heavy edge.

  The data structure consists of two components. First, for every
  heavy path with endpoints $v_{X_{\rm up}}$ and $v_{X_{\rm down}}$,
  we store the data structure from \cref{pr:pi-ds-2}. Since
  every node in $V(\mathcal{T}_c)$ belongs to at most one heavy path,
  the total size of these data structures, by
  \cref{th:cwt-size}, is $\bigO(1 + |{\rm RL}(W)| \log^2 m)$.
  The second component is the data structure from
  \cref{pr:pi-ds-1} for every node of $\mathcal{T}_c$.

  Given a pointer to any $v_X \in V(\mathcal{T}_c)$, and an integer $q
  \in [1 \dd |I_X|]$, the algorithm first checks if $X = \eps$.  If
  yes, then it holds $I_X[q] = q$ and hence it returns $q$ as the
  answer. Otherwise, we distinguish between two cases.  If the edge to
  the parent $v_{X_p}$ of $v_X$ is light then, by \cref{pr:pi-ds-1},
  we first compute $q' \in [1 \dd |I_{X_p}|]$ such that $I_{X_p}[q'] =
  I_X[q]$, and then recursively compute $I_{X_p}[q']$. Otherwise
  (i.e., if the edge to $v_{X_p}$ is heavy), by \cref{pr:pi-ds-2}, we
  first compute $q' \in [1 \dd |I_{X_{\rm up}}|]$ such that $I_{X_{\rm
  up}}[q'] = I_X[q]$ (where $v_{X_{\rm up}}$ is the highest node
  on the heavy path containing $v_X$; each node stores the pointer to
  such node), and then recursively compute $I_{X_{\rm up}}[q']$.

  Every two steps of the recursion, the number of light edges on path
  to the root of $\mathcal{T}_c$ decreases by at least one.  Thus, the
  query takes $\bigO(\log^4 m)$ time.
\end{proof}\vspace{1ex}

Finally, we prove that the data structure described above can be
constructed efficiently, given the compact representation of $W$, and
the LZ77 parsing of text $\T[1 \dd n]$.

\begin{theorem}\label{th:pi-construction}
  Given the LZ77 parsing of a string $\T[1 \dd n]$, and a sequence
  $W[1 \dd m]$ of $m \leq n$ substrings of $\T^{\infty}$ of length
  $\ell \leq n$, represented as ${\rm RL}((\lpos(W[i]))_{i \in [1 \dd
  m]})$, the compressed wavelet tree of $W$, supporting primary
  index queries in $\bigO(\log^4 n)$ time, can be constructed in
  $\bigO((z + |{\rm RL}(W)|) \log^2 n)$ time.
\end{theorem}
\begin{proof}
  The algorithm extends the construction of the compressed wavelet
  tree presented in \cref{th:cwt-construction}.

  As in the basic algorithm, after constructing the compact trie
  $\mathcal{T}_c$ of $W[1 \dd m]$, we process all nodes $v_X \in
  V(\mathcal{T}_c)$ bottom-up. During the traversal, we compute ${\rm
  size}(v_X)$, i.e., the number of leaves in the subtree rooted in
  $v_X$, for all $v_X \in V(\mathcal{T}_c)$.  This lets us identify
  the heavy edges. The algorithm maintains the invariant, that after
  the construction of ${\rm RL}(B_X)$ is complete, it stores a
  representation of the sequence ${\rm RL}((\ell_i)_{i \in [1 \dd
  |I_X|]})$, defined by $\ell_i = \lcp(W[I_X[i]], X_{\rm down})$,
  where $X_{\rm down}$ is the longest string having $X$ as a prefix,
  and for which it holds that $v_{X_{\rm down}} \in V(\mathcal{T}_c)$,
  and all edges on the path from $v_X$ to $v_{X_{\rm down}}$ are
  heavy. Then, if the path connecting $v_X$ and $v_{X_{\rm down}}$ is
  heavy, i.e., $X_{\rm down} \neq X$, and $v_X$ is the root of
  $\mathcal{T}_c$ or the edge to the parent of $v_X$ is light (see
  also the proof of \cref{th:pi-ds}), the algorithm uses this
  representation to initialize the data structure from
  \cref{pr:pi-ds-2}. By \cref{th:cwt-size} and~\cite[Table
  I]{chazelle}, over all heavy paths, this takes $\bigO(1 + |{\rm
  RL}(W)| \log^2 n)$ time.
  
  Consider any $v_X \in V(\mathcal{T}_c)$ and let $s = |I_X|$. The
  first challenge is representing the sequence ${\rm RL}((\ell_i)_{i
  \in [1 \dd s]})$ to support efficient queries and updates. Let us
  first focus on representing simply $(\ell_i)_{i \in [1 \dd s]}$. We
  start by observing that the set of pairs $\{(I_X[i], \ell_i)\}_{i
  \in [1 \dd s]}$ is a valid representation of $(\ell_i)_{i \in [1
  \dd s]}$, since for any $i,i' \in [1 \dd s]$ it holds $i < i'$
  if and only if $I_X[i] < I_X[i']$. To see the advantage of this
  representation, compared to $\{(i,\ell_i)\}_{i \in [1 \dd s]}$, let
  $v_{X'} \in V(\mathcal{T}_c)$ be child of $v_X$ connected by a heavy
  edge, and let $(\ell'_j)_{j \in [1 \dd p]}$ be the sequence defined
  by $\ell'_j = \lcp(W[I_{X'}[j]], X_{\rm down})$ for $j \in [1 \dd
  p]$, where $p = |I_{X'}|$. By \cref{lm:wt}, $(\ell'_j)_{j \in [1
  \dd p]}$ is a subsequence of $(\ell_i)_{i \in [1 \dd
  s]}$. Moreover, the ``position'' $I_{X'}[j]$ of $\ell'_j$ in the
  representation of $(\ell'_j)_{j \in [1 \dd p]}$ is automatically the
  correct position of $\ell'_j$ in the representation of $(\ell_i)_{i
  \in [1 \dd s]}$. More precisely, it holds
  \[
    \{(I_{X'}[j],\ell'_j)\}_{j \in [1 \dd p]} \subseteq
      \{(I_X[i],\ell_i)\}_{i \in [1 \dd s]}.
  \]

  The remaining elements of $\{(I_X[i], \ell_i)\}_{i \in [1 \dd s]}$
  correspond to children of $v_X$ other than $v_{X'}$. Specifically,
  letting $c = X'[|X|+1]$, we have
  \[
    \{(I_X[i], \ell_i)\}_{i \in [1 \dd s]} \setminus \{(I_{X'}[j],
    \ell'_j)\}_{j \in [1 \dd p]} = \{(I_X[i], |X|) : i \in [1 \dd
      s]\text{ and }B_X[i] \neq c\}.
  \]

  We now note that the sequence $J_X$, used as a replacement of $I_X$
  in the proof of \cref{th:cwt-construction}, satisfies the sufficient
  conditions to replace $I_X$ also in this case.  Namely, it holds
  $J_X[1] \leq \ldots \leq J_X[s]$, and for $i \in [1 \dd s)$, $\ell_i
  \neq \ell_{i + 1}$ implies $J_X[i] < J_X[i + 1]$. Thus, $\{(J_X[i],
  \ell_i)\}_{i \in [1 \dd s]}$ is also a valid representation of
  $(\ell_i)_{i \in [1 \dd s]}$, and moreover, since $I_X[i] =
  I_{X'}[j]$ implies $J_X[i] = J_{X'}[j]$, $\{(J_X[i], \ell_i)\}_{i
  \in [1 \dd s]}$ is also a disjoint union of
  $\{(J_{X'}[j],\ell'_j)\}_{j \in [1 \dd p]}$ and $\{(J_X[i], |X|) : i
  \in [1 \dd s]\text{ and }B_X[i] \neq c\}$.  The advantage of using
  the sequence $J_X$ is that the needed values are easy to compute
  during the construction in \cref{th:cwt-construction}. Therefore,
  letting ${\rm RL}((\ell_i)_{i \in [1 \dd s]}) = ((t_1, \lambda_1),
  \ldots, (t_{s'}, \lambda_{s'}))$, where $\lambda_0 = 0$, and
  $\delta_j = \lambda_j - \lambda_{j - 1}$ for $j \in [1 \dd s']$, we
  represent the sequence ${\rm RL}((\ell_i)_{i \in [1 \dd s]})$ as a
  set of pairs
  \[
    Q_X = \left\{ (J_X[\lambda_{j - 1} + 1], (t_j, \delta_j))
    \right\}_{j \in [1 \dd s']}.
  \]
  Note, that $J_X[\lambda_{j - 1} + 1] < J_X[\lambda_j + 1]$ for $j
  \in [1 \dd s')$.

  During the algorithm, each set $Q_X$ is stored in an AVL tree, with
  the first element of each pair as the key. In addition, for each set
  of pairs $Q$, we augment the tree that stores it (using the same
  technique as in the proof of \cref{th:cwt-construction}) to compute,
  for any $k'$, the value of
  \[
    \sum_{\substack{(k, (t, \delta))
        \in Q\\k < k'}} \delta.
  \]

  Assume that $v_X$ is the node currently processed by the
  algorithm. Denote ${\rm RL}(B_X) = ((c_1, \kappa_1), \ldots,
  (c_{s''}, \kappa_{s''}))$, letting $\kappa_0 = 0$ and $c_0 = c_{s''
  + 1} = c$ (recall that $c = X'[|X| + 1]$). Observe that during the
  computation of ${\rm RL}(B_X)$, at no extra cost, the algorithm in
  the proof of \cref{th:cwt-construction} can also return all values
  $J_X[\kappa_{i - 1} + 1]$, where $i \in [1 \dd s'']$.  If $X_{\rm
  down} = X$, i.e., $v_X$ is a leaf or all its children are
  connected using light edges, then it holds $\ell_i = |X|$ for every
  $i \in [1 \dd s]$. Thus, the set representing ${\rm RL}((\ell_i)_{i
  \in [1 \dd s]})$ consists of a single pair $Q_X = \{(J_X[\kappa_0
  + 1], (|X|, s))\}$. Let us thus assume that $v_{X'}$ is the child
  of $v_X$ connected using a heavy edge. Assume also that we are given
  the representation $Q_{X'}$ of ${\rm RL}((\ell'_j)_{j \in [1 \dd
  p]})$.  By the above characterization, the set $Q_X$ contains
  the pair $(J_X[\kappa_{j_b - 1} + 1], (|X|, \kappa_{j_e} -
  \kappa_{j_b - 1}))$, for all subsets $\{j_b, j_e\} \subseteq [1 \dd
  s'']$ that satisfy $j_b \leq j_e$, $c \not\in \{c_{j_b}, c_{j_b +
  1}, \ldots, c_{j_e}\}$, and $c_{j_b - 1} = c_{j_e + 1} =
  c$. Simply adding these pairs into $Q_{X'}$, however, does not
  produce the correct $Q_X$.

  Let us initialize $Q := Q_{X'}$.  Consider all subsets $\{j_b, j_e\}
  \subseteq [1 \dd s'']$ satisfying the above conditions in the order
  of increasing $j_b$.  The complication, following from the fact that
  $\ell_i = \ell_{i + 1}$ does not imply $J_X[i] = J_X[i + 1]$, is to
  ensure that after inserting the pair corresponding to the subset
  $\{j_b, j_e\}$, into the current set $Q$, the value
  \[
    \delta_{\rm prev} = \sum_{\substack{(k,(t,\delta)) \in
        Q\\ k<J_{X}[\kappa_{j_b - 1} + 1]}}\delta
  \]
  satisfies $\delta_{\rm prev} = \kappa_{j_b - 1}$. To achieve this,
  before inserting the pair, we first compute (using the augmented AVL
  tree) the current value $\delta_{\rm prev}$ and the pointer to the
  node storing the pair $(k,(t,\delta)) \in Q$ with the largest $k$
  smaller than $J_{X}[\kappa_{j_b - 1} + 1]$. If $\delta_{\rm prev} >
  \kappa_{j_b - 1}$ (note that this implies $j_e < s''$), we replace
  the pair $(k,(t,\delta))$ in $Q$ with two pairs $(k, (t, \delta'))$
  and $(J_{X}[\kappa_{j_e} + 1], (t, \delta''))$, such that $\delta' +
  \delta'' = \delta$ and $\delta'' = \delta_{\rm prev} - \kappa_{j_b -
  1}$. Only then, we insert the pair $(J_X[\kappa_{j_b - 1} + 1],
  (|X|, \kappa_{j_e} - \kappa_{j_b - 1}))$ into $Q$.

  Each operation on $Q$ takes $\bigO(\log m)$ time. Thus, computing
  the set $Q_X$ from $Q_{X'}$ takes $\bigO(|{\rm RL}(B_X)| \log m)$
  time.  By \cref{th:cwt-size}, over all $v_X \in V(\mathcal{T}_c)$,
  this amounts to $\bigO(1 + |{\rm RL}(W)| \log^2 n)$ time.

  To complete the construction, we observe that given ${\rm RL}(B_Y)$,
  all sequences $L_{Y, c}$ (with $c$ occurring in $B_Y$) in the proof
  of \cref{pr:pi-ds-1} are easy to compute in $\bigO(1 + |{\rm
  RL}(B_Y)| \log |{\rm RL}(B_Y)|)$ time. Thus, by
  \cref{th:cwt-size}, initializing the structure from
  \cref{pr:pi-ds-1} takes $\bigO(1 + |{\rm RL}(W)| \log^2 n)$ time.
\end{proof}\vspace{1ex}

Similarly, as for \cref{th:cwt-construction-rev}, the above result
applies also when $W$ is a sequence of substrings of
$\revstr{\T}^{\infty}$.

\begin{theorem}\label{th:pi-construction-rev}
  Given the LZ77 parsing of a string $\T[1 \dd n]$, and a sequence
  $W[1 \dd m]$ of $m \leq n$ substrings of $\revstr{\T}^{\infty}$ of
  length $\ell \leq n$, represented as ${\rm RL}((\rpos(W[i]))_{i \in
  [1 \dd m]})$, the compressed wavelet tree of $W$, supporting
  primary index queries in $\bigO(\log^4 n)$ time, can be constructed
  in $\bigO((z + |{\rm RL}(W)|) \log^2 n)$ time.
\end{theorem}

\subsection{The Algorithm}\label{sec:algorithm-details}

\newcommand{\Wboth}{\widetilde{W}}
\newcommand{\Wpref}{W}
\newcommand{\Wsuf}{W'}

We are now ready to show how to construct the sequences ${\rm
RL}(\BWT_{\ell})$, where $\ell = 2^q$ for $q \in [0 \dd \lceil \log
n\rceil]$.

For small $\ell$, constructing ${\rm RL}(\BWT_{\ell})$ reduces to
sorting and computing frequencies of length-$\Theta(\ell)$ substrings
of $\T^{\infty}$.

\begin{proposition}\label{pr:bwt-base}
  Let $\ell = \bigO(1)$. Given the LZ77 parsing of $\T[1 \dd n]$, the
  sequence ${\rm RL}(\BWT_{\ell})$ can be constructed in $\bigO(z
  \log^4 n)$ time.
\end{proposition}
\begin{proof}
  First, we compute $\Su_{\ell}$ with each string $X\in \Su_{\ell}$
  represented by $p_X := \lpos(X)$.  For this, we exploit the fact
  that $p_X \in ({e_j-\ell}\dd e_j]$ for some $j\in [1\dd z]$.  Thus,
  it suffices to list the $\Oh(\ell z)=\Oh(z)$ candidate positions
  $p$, group them according to $\T^\infty[p\dd p+\ell)$ (by sorting),
  and keep the leftmost position in each group.  Next, for each $X\in
  \Su_{\ell}$, we count the number $c_X$ of positions $i\in [1\dd n]$
  such that $\T^\infty[i\dd i+\ell)=X$.  For this, we use
  \cref{th:index-count} to compute the number of occurrences of $X$ in
  $\T^\infty[1\dd n+\ell)$.  Similarly, we count the number $c'_X$ of
  positions $i\in [1\dd n]$ such that $\T^\infty[i-1\dd
  i + \ell)=T^\infty[p_X-1\dd p_X+\ell)$; observe that $X$ is
  left-maximal if and only if $c'_X < c_X$.  Finally, we construct
  $\BWT_\ell$ from left to right processing the strings $X\in
  \Su_{\ell}$ in lexicographic order and appending $p_X^{c_X}$ if
  $c'_X < c_X$, and $(T^\infty[p_X-1])^{c_X}$ otherwise.  The overall
  running time $\Oh(z\log^4 n)$ is dominated by the use of
  \cref{th:index-count}.
\end{proof}\vspace{1ex}

Let $q \geq 4$. We show how to compute ${\rm RL}(\BWT_{2\ell})$, given
the LZ77 parsing of $\T$ and ${\rm RL}(\BWT_{\ell})$.  The main idea
of the algorithm is as follows.

Let $\S$ be a $\tau$-synchronizing set of $\T$, where $\tau = \lfloor
\tfrac{\ell}{3} \rfloor$.  As noted earlier, $\BWT_{\ell}[j] \in
\Sigma$ implies $\BWT_{2\ell}[j] \in \Sigma$.  Let $\BWT_{\ell}[y \dd
y'] \in \mathbb{N}^{+}$ be a run in $\BWT_{\ell}$.  By definition of
$\BWT_{\ell}$, the suffixes of $\T^\infty$ starting at positions
$i\in\SA[y \dd y']$ share a common prefix of length $\ell\ge
3\tau$. Thus, assuming that $\S \cap [i \dd i \,{+}\, \tau) \neq
\emptyset$ holds for all $i\, {\in}\, \SA[y \dd y']$ (the periodic
case is handled separately), by the consistency of $\S$, all text
positions $i\,{\in}\, \SA[y \dd y']$ share a common offset $\Delta$
with $i\, {+}\, \Delta = \min(\S \cap [i \dd i \,{+}\, \tau) )$.  This
lets us deduce the order of length-$2\ell$ prefixes $\T[i \dd i \, {+}
\, 2\ell)$ based on the order of strings $\T[i \, {+} \, \Delta \dd
i \, {+} \, 2\ell)$ starting at synchronizing positions. For this,
from the sorted list of fragments $\T[s\dd s \, {+} \, 2\ell\, {-} \,
\Delta)$ across $s \in \S$, we extract, using a wavelet tree, those
preceded by $\T[i\dd i \, {+} \, \Delta)$ (a prefix common to
$\T^\infty[i \dd)$ for $i\in \SA[y \dd y']$). Importantly, the
synchronizing positions $s$ sharing $\T[s \, {-} \, \ell \dd s \, {+}
\, 2\ell)$ can be processed together; hence, by
\cref{th:csss-existence}, it suffices to use $\bigO(z)$ distinct
substrings.

We formalize these ideas as follows. Let
\[
  \R = \left\{ i \in [1 \dd n \,{-}\, 3\tau \,{+}\, 2] : \per \left(
  \T[i \dd i \,{+}\, 3\tau \,{-}\, 2] \right) \leq \tfrac{1}{3} \tau
  \right\}.
\]
The description of the algorithm is divided into the nonperiodic case
(when $\R = \emptyset$) and the general case.

\vspace{1ex}
\subsubsection{The Nonperiodic Case}\label{sec:bwt-nonperiodic}

Let $(s_i')_{i \in [1 \dd |\S|]}$ be the sequence containing all
positions in $\S$ such that $i<j$ holds~if
\begin{itemize}
\item $\T^{\infty}[s_i' \dd s_{i}' + 7\tau) \prec \T^{\infty}[s_j' \dd
    s_j' + 7\tau)$, or
\item $\T^{\infty}[s_i' \dd s_i' + 7\tau) = \T^{\infty}[s_j' \dd s_j'
    + 7\tau)$ and\\ $\revstr{\T^{\infty}[s_i' - \tau \dd s_i')} \prec
      \revstr{\T^{\infty}[s_j' - \tau \dd s_j')}$.
\end{itemize}
Based on $(s_i')_{i \in [1 \dd |\S|]}$, we define three length-$|\S|$
sequences. For $i \in [1 \dd |\S|]$, we set
\begin{align*}
  \Wboth[i] &= \T^{\infty}[s'_i - \tau \dd s'_i + 7\tau),\\
  \Wpref[i] &= \revstr{\T^{\infty}[s_i' - \tau \dd s_i')},\\
  \Wsuf[i] &= \T^{\infty}[s_i' \dd s_i' + 7\tau).
\end{align*}

Recall that we can compactly represent the sequence $\Wboth$ in
$\bigO(1 + |{\rm RL}(\Wboth)|)$ space using ${\rm
RL}((\lpos(\Wboth[j]))_{j \in [1 \dd |\S|]})$.  The sequences
$\Wpref$ and $\Wsuf$ can be represented analogously, except that we
use ${\rm RL}((\rpos(\Wpref[j]))_{j \in [1 \dd |\S|]})$ for $\Wpref$.

\begin{lemma}\label{lm:w-size}
  The sequences $\widetilde{W}$, $W$, and $W'$ defined above satisfy
  $|{\rm RL}(\Wpref)|$, $|{\rm RL}(\Wsuf)| \leq |{\rm RL}(\Wboth)|
  \,{\leq}\, |{\rm comp}_7(\S)|$.
\end{lemma}
\begin{proof}
  For the first inequality, note that $\Wboth[i] = \Wboth[i + 1]$
  implies $\Wpref[i] = \Wpref[i + 1]$ and $\Wsuf[i] = \Wsuf[i + 1]$.
  
  Let ${\rm RL}(\Wboth) = ((R_1, \lambda_1), \ldots, (R_h,
  \lambda_h))$.  Observe that $i \neq j$ implies $R_i \neq R_j$. For
  $i \in [1 \dd h]$, let $\T^{\infty}[p - \tau \dd p + 7\tau)$ be the
  occurrence of $R_i$ in $\T^{\infty}$ that minimizes $p \in [1 \dd
  n]$. Then, there exists $j \in [1 \dd z]$ such that $e_j - 8\tau <
  p - \tau \leq e_j$.  By the consistency of $\S$, we conclude that $p
  \in \S \cap (e_j-7\tau \dd e_j + \tau] \subseteq {\rm comp}_7(\S)$.
  The claim follows, since this map is injective.
\end{proof}\vspace{1ex}

Importantly, the compact representations of $\Wboth$, $\Wpref$, and
$\Wsuf$ can be computed efficiently.

\begin{lemma}\label{lm:w-construction}
  Given ${\rm comp}_7(\S)$ and the LZ77 parsing of $\T$, the compact
  representations of $\Wboth$, $\Wpref$, and $\Wsuf$ can be
  constructed in $\bigO(z \log^4 n + |{\rm comp}_7(\S)| \log^3 n)$
  time.
\end{lemma}
\begin{proof}
  Denote ${\rm RL}(\Wboth) = ((R_1, \lambda_1), \ldots, (R_k,
  \lambda_k))$, letting $\lambda_0 = 0$, and $\delta_i = \lambda_i -
  \lambda_{i - 1}$ for $i \in [1 \dd k]$. As observed in the proof of
  \cref{lm:w-size}, it holds $\{R_i\}_{i \in [1 \dd k]} =
  \{\T^{\infty}[i - \tau \dd i + 7\tau) : i \in {\rm
  comp}_7(\S)\}$. Using LCE queries on $\T$ we can sort any set of
  substrings of $\T^{\infty}$. Thus, using
  \cref{th:index-lce,th:index-leftmost-occ}, we first compute the
  sequence $(\lpos(R_i))_{i \in [1 \dd k]}$. We then observe that by
  the consistency of $\S$, the value $\delta_i$ is the number of
  occurrences of $R_i$ in $\T^{\infty}$ starting at a position $j$
  satisfying $j + \tau \in [1 \dd n]$. Thus, we obtain $(\lambda_i)_{i
  \in [1 \dd k]}$ using \cref{th:index-count}. In total, this takes
  $\bigO(z \log^4 n + |{\rm comp}_7(\S)| \log^3 n)$ time.

  To obtain the compact representation of $\Wpref$ (resp. $\Wsuf$), it
  now suffices to merge the adjacent equal runs obtained by discarding
  the length-$7\tau$ suffix (resp.\ length-$\tau$ prefix) of strings
  in $\Wboth$.  Importantly, the counts for the runs in ${\rm
  RL}(\Wpref)$ and ${\rm RL}(\Wsuf)$ are computed from the counts of
  ${\rm RL}(\Wboth)$.  The merging is performed using
  \cref{th:index-lce}. We then compute the leftmost occurrences using
  \cref{th:index-leftmost-occ,th:index-rightmost-occ}.  In total, we
  spend $\bigO(z \log^4 n + |{\rm comp}_7(\S)| \log^3 n)$ time.
\end{proof}

Next, we recall the notion of \emph{distinguishing prefixes},
originally introduced in~\cite{sss}, that allows mapping each suffix
$\T^\infty[i\dd )$ to the corresponding node of the wavelet tree
of~$W$.
\begin{definition}[{\fontfamily{lmss}\selectfont Distinguishing prefix}{\fontfamily{lmr}\selectfont}]\label{def:dp}
  For any position $i \in [1 \dd \max (\S \cup \{0\})]$, let $i_{\rm
  succ} = \min \{ j \in \S : j \geq i \}$. The \emph{distinguishing
  prefix} of $\T[i \dd n]$ is $D_i = \T[i \dd i_{\rm succ} +
  2\tau)$.
\end{definition}

Let $\mathcal{D} = \{D_j : j \in [1 \dd \max (\S \cup \{0\})]\}$. Note
that if $Y$ starts with $D \in \mathcal{D}$, then, for every
occurrence $\T^{\infty}[i \dd i {+} |Y|) = Y$ with $i \in [1 \dd n]$,
the distinguishing prefix $D_i$ is defined and satisfies $D_i =
D$.\footnote{Here, we utilize the assumption that $\T[n] = \$$.  For
this reason, if $Y$ contains $\$$, then $\T^{\infty}[i \dd i{+}|Y|) =
Y$ for at most one index $i\in [1 \dd n]$.}  Thus, for any such $Y$,
we define $D_Y = D$.  We denote $D'_Y = D_Y[1 \dd |D_Y| {-} 2\tau]$.

We now present the key lemma used in our algorithm.  Assume that we
have constructed a wavelet tree of $\Wpref$.

\begin{lemma}\label{lm:main}
  Let $Y$ be a string starting with an element of $\mathcal{D}$.
  Denote $Y = XX'$, where $X=D'_Y$, and assume that $|X|<\tau$ and
  $|X'| \leq 7\tau$. Let $[y \dd y']$ be the range of all indices $i$
  such that $\T^{\infty}[\SA[i] \dd]$ starts with $Y$ for $i \in [y
  \dd y']$.

  Let $\Wsuf[f \dd f']$ be the range containing all elements of
  $\Wsuf$ prefixed with the string $X'$, and let $[b \dd b'] = \{i \in
  [1 \dd |I_{\revstr{X}}|] : I_{\revstr{X}}[i] \in [f \dd f']\}$. Then
  \begin{enumerate}
  \item $B_{\revstr{X}}[b \dd b']$ and $\BWT[y \dd y']$ are equal as
    multisets.
  \item $|{\rm RL}(B_{\revstr{X}}[b \dd b'])| \leq 3|{\rm RL}(\BWT[y
    \dd y'])|$.
  \end{enumerate}
\end{lemma}
\begin{proof}
  1. Due to $\T[n] = \$$, by the consistency of $\S$, there is a
  one-to-one correspondence between the occurrences of $Y$ in
  $\T^{\infty}$ starting in $[1 \dd n]$, and positions $s \in \S$
  satisfying (a) $\T^{\infty}[s \dd s \,{+}\, |X'|) = X'$, and (b)
  $\T^{\infty}[s \,{-}\, |X| \dd s) = X$. Let us interpret the
  process of identifying the subsequence of $(s_i')_{i \in [1 \dd
  |\S|]}$ containing all such $s$ as a two-step search.

  First, we note that $s'_i \in \S$ satisfies condition (a) if and
  only if $i \in [f \dd f']$. We refer to the process of identifying
  the range $[f \dd f']$ as the \emph{forward search}. Then, to
  additionally satisfy (b), we select a subsequence of $\Wpref[f \dd
  f']$ containing only strings ending with $X$ (\emph{backward
  search}). By definition of $[b \dd b']$, such subsequence is given
  by $I_{\revstr{X}}[b \dd b']$, and moreover, $B_{\revstr{X}}[b \dd
  b']$ contains symbols preceding suffix $X$ in all $\Wpref[f \dd
  f']$ having $X$ as a suffix. This yields the claim.

  2. Let $\widetilde{Y}$ be any substring of $\T^{\infty}$ such that
  $Y$ is a prefix of $\widetilde{Y}$ and $|\widetilde{Y}| = |X| +
  7\tau$.  Let $[\widetilde{y} \dd \widetilde{y}']$, $[\widetilde{f}
  \dd \widetilde{f}']$, and $[\widetilde{b} \dd \widetilde{b}']$ be
  the ranges (as in the lemma statement) for $\widetilde{Y}$.  Since
  $Y$ is a prefix of $\widetilde{Y}$ and $D_{\widetilde{Y}} = D_{Y}$,
  we obtain $[\widetilde{y} \dd \widetilde{y}'] \subseteq [y \dd y']$,
  $[\widetilde{f} \dd \widetilde{f}'] \subseteq [f \dd f']$, and
  $[\widetilde{b} \dd \widetilde{b}'] \subseteq [b \dd b']$.
  Moreover, by definition of $I_{\revstr{X}}$, the range $[b \dd b']$
  is a disjoint union of ranges $[\widetilde{b} \dd \widetilde{b}']$
  corresponding to all choices of $\widetilde{Y}$.

  Since $|\widetilde{Y}|-|X| = 7\tau$ implies $|{\rm
  RL}(\Wsuf[\widetilde{f} \dd \widetilde{f}'])| = 1$, the symbols in
  $B_{\revstr{X}}[\widetilde{b} \dd \widetilde{b}']$ appear in the
  nondecreasing order. Consequently, $B_{\revstr{X}}[b \dd b']$ can be
  obtained by partitioning $\BWT[y \dd y']$ into blocks corresponding
  to all $\widetilde{Y}$, and sorting the symbols in each block. If
  $\BWT[y \dd y']$ initially contains $k$ runs, this adds at most
  $2(k-1)$ new runs.
\end{proof}\vspace{1ex}

Let $\BWT_{\ell}[y \dd y'] = c^{\delta} \in \mathbb{N}^{+}$ be a run
in $\BWT_{\ell}$. Since $\T^{\infty}[c \dd c + \ell)$ is left-maximal,
we have $c + \ell \leq n$. Let $Y = \T[c \dd c + \ell)$. By $3\tau
\leq \ell$ and $\R = \emptyset$, we obtain $[c \dd c + \tau) \cap \S
\neq \emptyset$. Moreover, if $Y = XX'$ is such that $|X| = \Delta$,
where
\[
  c + \Delta = \min(\S \cap [c \dd c + \tau)),
\]
then $D'_Y = X$. Since we also have $|X'| \leq 2\ell \leq 7\tau$ (due
to $q \geq 4$), \cref{lm:main} holds for $Y$. Let $[b \dd b']$ be the
range of $B_{\revstr{X}}$ corresponding to $Y$ through
\cref{lm:main}. Then, $|{\rm RL}(B_{\revstr{X}})| > 1$. Moreover:
\begin{itemize}
\item For every $j \in [0 \dd \delta)$ such that $\BWT_{2\ell}[y + j]
  \in \Sigma$, it holds: $\BWT_{2\ell}[y + j] =
  B_{\revstr{X}}[b + j]$.  To see this, apply \cref{lm:main} to all
  strings $\widetilde{Y} \in \mathcal{Y} := \{\T^{\infty}[\SA[j] \dd
  {\SA[j] + 2\ell}) : j \in [y \dd y']\}$ ordered
  lexicographically. Since $\smash{D_{\widetilde{Y}} = D_Y}$, the
  corresponding ranges $\smash{[\widetilde{b} \dd \widetilde{b}']}$
  form a left-to-right partition of $[b \dd b']$.  Thus, if
  $\widetilde{Y}$ is not left-maximal, its $\BWT$ block is
  $\BWT[\widetilde{y} \dd \widetilde{y}']=B_{\revstr{X}}[\widetilde{b}
  \dd \widetilde{b}']$.

\item On the other hand, if $c' = \BWT_{2\ell}[y + j] \in \mathbb{N}$
  holds for some $j \in [0 \dd \delta)$, then the string
  $B_{\revstr{X}}[\widetilde{b} \dd \widetilde{b}']$ corresponding
  (through \cref{lm:main}) to $\widetilde{Y} = \T^{\infty}[{c' \dd c'
  + 2\ell}) \in \mathcal{Y}$ is not unary, i.e., there exists an index
  $\hat{b} \in [\widetilde{b} \dd \widetilde{b}')$ satisfying
  $B_{\revstr{X}}[\hat{b}] \neq B_{\revstr{X}}[\hat{b} + 1]$, and
  $\lcp(\Wsuf[I_{\revstr{X}}[\hat{b}]],\allowbreak
  \Wsuf[I_{\revstr{X}}[{\hat{b} + 1}]]) \geq 2\ell - |X|$. The
  converse is also true: if $\hat{b} \in [b \dd b')$ satisfies the two
  conditions, then $\BWT_{2\ell}[y + (\hat{b} - b)] \in
  \mathbb{N}$. Consequently, the set of left-maximal strings in
  $\mathcal{Y}$ is
  \begin{align*}
    &\{X \cdot \Wsuf[I_{\revstr{X}}[\hat{b}]][1\dd 2\ell{-}|X|]
    :\\ &\hspace{1.2cm}\hat{b} \in [b \dd b'),\;
    B_{\revstr{X}}[\hat{b}] \neq B_{\revstr{X}}[\hat{b}+1],\text{
    and}\\ &\hspace{1.2cm}\lcp(\Wsuf[I_{\revstr{X}}[\hat{b}]],
    \Wsuf[I_{\revstr{X}}[\hat{b}+1]]) \geq 2\ell-|X|\}.
  \end{align*}
  Moreover, letting $\widetilde{Y} = X \cdot
  \Wsuf[I_{\revstr{X}}[\hat{b}]][1\dd 2\ell{-}|X|]$ for any $\hat{b}$
  satisfying the above conditions, the range $[\widetilde{y} \dd
  \widetilde{y}'] = \{j \in [1 \dd n]: \T^{\infty}[\SA[j] \dd \SA[j]
  + 2\ell) = \widetilde{Y}\}$ satisfies $[\widetilde{y} \dd
  \widetilde{y}'] = y - b + [\widetilde{b} \dd \widetilde{b}']$,
  where $\smash{B_{\revstr{X}}[\widetilde{b} \dd \widetilde{b}']}$
  corresponds to $\smash{\widetilde{Y}}$ through \cref{lm:main}.
\end{itemize}

The algorithm processing a run $\BWT_{\ell}[y \dd y'] = c^{\delta} \in
\mathbb{N}^+$ is thus as follows. Letting $Y = \T[c \dd c +
\ell)$ and $X=D'_Y$, we first compute $|X|$ and the pointer to
$v_{\revstr{X}}$. We then perform a single forward and backward search
to find the ranges $[f \dd f']$ and $[b \dd b']$ for $Y$. Given these,
the computation of $\BWT_{2\ell}[y \dd y']$ is achieved by a series of
forward and backward searches (at most one per run of
$B_{\revstr{X}}[b \dd b']$).

The length $|X|$ is computed using ${\rm comp}_7(\S)$. The pointers to
$v_{\revstr{X}}$ are precomputed since $|{\rm RL}(B_{\revstr{X}})| >
1$ implies $v_{\revstr{X}} \in V(\mathcal{T}_c)$.  To implement a
forward search, we use LCE queries on $\T$. A backward search is
implemented using primary index queries on the wavelet tree of $W$.
We thus obtain:

\begin{proposition}\label{pr:bwt-nonperiodic}
  Let $\T$ be a string of length $n$, and let $\ell = 2^q$ be such
  that $q \in [4 \dd \lceil{\log n}\rceil)$.  If $\R = \emptyset$,
  then, given ${\rm RL}(\BWT_{\ell})$ and the LZ77 parsing of $\T$,
  the sequence ${\rm RL}(\BWT_{2\ell})$ can be constructed in
  $\bigO((r + z) \log^5 n)$ time.
\end{proposition}
\begin{proof}
  Let $\tau = \lfloor \tfrac{\ell}{3} \rfloor$. We start by
  constructing the compressed representation ${\rm comp}_7(\S)$ of a
  $\tau$-synchronizing set $\S$ of $\T$ satisfying $|{\rm comp}_7(\S)|
  = \bigO(z)$. By \cref{th:csss-construction}, this takes $\bigO(z
  \log^5 n)$ time. Next we construct the compact representations of
  $\Wboth$, $\Wpref$, and $\Wsuf$. By
  \cref{lm:w-size,lm:w-construction}, they need $\bigO(z)$ space and
  can be built in $\bigO(z \log^4 n)$ time. Lastly, we construct the
  compressed wavelet tree $\mathcal{T}_c$ for $\Wpref$, augmented with
  the data structure supporting primary index queries. By
  \cref{th:pi-construction-rev}, this takes $\bigO(z \log^2 n)$ time.

  Let $\BWT_{\ell}[y \dd y'] = c^{\delta} \in \mathbb{N}^+$ be one of
  the runs in $\BWT_{\ell}$. As noted earlier, this implies $c + \ell
  \leq n$.  By $3\tau \leq \ell$ and $\R = \emptyset$, we then obtain
  $\S \cap [c \dd c + \tau) \neq \emptyset$. Denote $Y = \T[c \dd c +
  \ell) = XX'$, with $X = D'_Y$. Then:
  \begin{enumerate}
  \item By $c = \lpos(Y)$, we have $\S \cap [c \dd c + \tau)
    \subseteq {\rm comp}_7(\S)$ and get $|X| = \min(\S \cap
    [c \dd c {+} \tau)) {-} c$ in $\bigO(\log n)$ time.
  \item To compute the pointer to $v_{\revstr{X}}$, we store a hash
    table than maps $(|Z|, \rpos(Z))$ to $v_Z$ for every $v_Z \in
    V(\mathcal{T}_c)$.  Such table is easily initialized in $\bigO(z
    \log^5 n)$ time during the construction of $\mathcal{T}_c$. Thus,
    using \cref{th:index-rightmost-occ}, we obtain a pointer to
    $v_{\revstr{X}}$ in $\bigO(\log^4 n)$ time.
  \item To find the range $[f \dd f']$ (i.e., the forward search), we
    use the fact that $\Wsuf$ is lexicographically sorted. Thus, by
    \cref{th:index-lce}, we only need $\bigO(\log^2 n)$ time.
  \item Lastly, to find the range $[b \dd b']$ (i.e., the backward
    search), we use primary index queries on $\mathcal{T}_c$. Using
    binary search and the data structure from
    \cref{th:pi-construction-rev}, computing $[b \dd b']$ takes
    $\bigO(\log^5 n)$ time.
  \end{enumerate}
  We then proceed to the construction of ${\rm RL}(\BWT_{2\ell}[y \dd
  y'])$.  We begin by initializing the output to ${\rm
  RL}(B_{\revstr{X}}[b \dd b'])$. By the above discussion, it
  remains to identify all left-maximal $\widetilde{Y} \in \mathcal{Y}$
  and replace (in the output) the range $[\widetilde{y} \dd
  \widetilde{y}'] \subseteq [y \dd y']$ corresponding to each such
  $\widetilde{Y}$ with $(\lpos(\widetilde{Y}))^{\widetilde{y}' -
  \widetilde{y}+1}$.

  Using ${\rm RL}(B_{\revstr{X}}[b \dd b'])$, and LCE and primary
  index queries we find all $\hat{b} \in [b \dd b')$ satisfying
  $B_{\revstr{X}}[\hat{b}] \neq B_{\revstr{X}}[\hat{b} + 1]$ and
  $\lcp(\Wsuf[I_{\revstr{X}}[\hat{b}]], \Wsuf[I_{\revstr{X}}[\hat{b} +
  1]]) \geq 2\ell {-} |X|$. For every such $\hat{b}$, we find the
  ranges $[\widetilde{f} \dd \widetilde{f}']$ and $[\widetilde{b}
  \dd \widetilde{b}']$, corresponding to the left-maximal
  $\widetilde{Y} = X \cdot \Wsuf[I_{\revstr{X}}[\hat{b}]][1 \dd 2\ell
  {-} |X|]$ as above (using LCE and primary index queries). To
  compute $\lpos(\widetilde{Y})$ we use \cref{th:index-leftmost-occ},
  noting that $\widetilde{Y}$ occurs in
  $\Wboth[I_{\revstr{X}}[\hat{b}]]$.

  Initializing the auxiliary indexes
  (\cref{th:index-lce,th:index-leftmost-occ,th:index-rightmost-occ})
  takes $\bigO(z \log^ 4 n)$ time. By \cref{lm:main} we spend at most
  $\bigO(\log^5 n)$ time per run of $\BWT_{2\ell}$.  Thus, by
  \cref{lm:bwtm-size}, we spend $\bigO(r \log^5 n)$ time overall.
\end{proof}\vspace{1ex}

\subsubsection{The General Case}\label{sec:bwt-general}

In this section, we show how to extend the algorithm from the previous
section to allow $\R \neq \emptyset$.

We start by observing, that the only place where we used the
assumption $\R = \emptyset$ in Proposition~\ref{pr:bwt-nonperiodic} is
to deduce that for any run $\BWT_{\ell}[y \dd y'] = c^{\delta}$ of
$\BWT_{\ell}$ that satisfies $c^{\delta} \in \mathbb{N}^+$, it holds
$\S \cap [c \dd c + \tau) \neq \emptyset$. All lemmas in the previous
section hold, however, even when $\R \neq \emptyset$. The general
algorithm is therefore an extension of the procedure in
\cref{pr:bwt-nonperiodic}, i.e., all runs $\BWT_{\ell}[y \dd y'] =
c^{\delta} \in \mathbb{N}^+$ satisfying $\S \cap [c \dd c + \tau) \neq
\emptyset$ are processed as in
\cref{pr:bwt-nonperiodic}, except rather than computing $\S$
satisfying $|{\rm comp}_7(\S)| = \bigO(z)$, we instead compute
$\S$ satisfying $|{\rm comp}_{9}(\S)| = \bigO(z)$. Any
such $\S$ clearly also satisfies $|{\rm comp}_7(\S)| = \bigO(z)$.

The remaining runs are handled as follows.
Let $\R' := \{ j \in \R : j - 1 \not \in \R \}$ be a subset of
$\R$. For $k \in [1 \dd n]$, let $\mathcal{F}_k = \{\T^{\infty}[i \dd
i + k) : i \in \R\}$ and $\mathcal{F}_{k}' := \left \{ \T^{\infty}[i
\dd i + k) : i \in \R' \right \} \subseteq \mathcal{F}_{k}$.

\begin{lemma}\label{lm:periodic}
  It holds:
  \begin{enumerate}
  \item $|\mathcal{F}_{2\ell}'| \leq |{\rm comp}_7(\S)|$.
  \item If $\T^{\infty}[i \dd i + 2\ell) = \widetilde{Y} \in
    \mathcal{F}_{2\ell} \setminus \mathcal{F}_{2\ell}'$ then
    $\T^{\infty}[i - 1] = \widetilde{Y}[\per(\widetilde{Y}[1 \dd 3\tau
        - 1])]$.
  \end{enumerate}
\end{lemma}
\begin{proof}
  1. Let $\widetilde{Y} \in \mathcal{F}_{2\ell}'$. Then,
  there exists $j \in \R'$ such that $\T^{\infty}[j \dd j + 2\ell) =
  \widetilde{Y}$. Let $j' = \lpos(\T^{\infty}[j - 1 \dd j + 2\ell))$.
  Since $j - 1 \in \S$, by the consistency condition, we have $j' \in \S$.
  Moreover, since for $\ell \geq 16$ we have $2\ell+1 \leq 7\tau$
  (recall, that $\tau = \lfloor \frac{\ell}{3} \rfloor$), it also holds
  $j' \in {\rm comp}_7(\S)$. This proves the claim as this mapping from
  $\mathcal{F}_{2\ell}'$ to ${\rm comp}_7(\S)$ is injective.

  2. By definition of $\mathcal{F}_{2\ell}'$, $\widetilde{Y} \in
  \mathcal{F}_{2\ell} \setminus \mathcal{F}_{2\ell}'$ implies $i>1$
  and $i - 1 \in \R$. We also have $i \in \R$. This implies that $p :=
  \per(\T^{\infty}[i {-} 1 \dd i {+} 3\tau {-} 2))$ and $p' :=
  \per(\T^{\infty}[i \dd i {+} 3\tau {-} 1)) = \per(\widetilde{Y}[1
  \dd 3\tau-1])$ satisfy $p, p' \leq \frac{1}{3}\tau$. By
  periodicity lemma, $p = p'$ and hence $\T^{\infty}[i - 1] =
  \T^{\infty}[i - 1 + p'] = \widetilde{Y}[p']$.
\end{proof}\vspace{1ex}

Let $\BWT_{\ell}[y \dd y'] = c^{\delta} \in \mathbb{N}^{+}$ be a run
in ${\rm RL}(\BWT_{\ell})$. Let $Y = \T[c \dd c + \ell)$ and assume
$\S \cap [c \dd c + \tau) = \emptyset$.  Consider $\widetilde{Y} \in
\mathcal{Y} := \{\T^{\infty}[\SA[i] \dd \SA[i] + 2\ell) : i \in [y \dd
y']\}$.  Let moreover $[\widetilde{y} \dd \widetilde{y}']
\subseteq [y \dd y']$ be the range of all indices $i$ such that
$\T^{\infty}[\SA[i] \dd ]$ starts with $\widetilde{Y}$ for $i \in
[\widetilde{y} \dd \widetilde{y}']$.
Observe, that $\widetilde{Y} \in \mathcal{F}_{2\ell}$ and thus, by the
above lemma, whenever $\widetilde{Y} \in \mathcal{F}_{2\ell} \setminus
\mathcal{F}_{2\ell}'$, we have $\BWT_{2\ell}[\widetilde{y} \dd
\widetilde{y}'] = \widetilde{Y}[p]^{\widetilde{y}' - \widetilde{y} +
1} = Y[p]^{\widetilde{y}' - \widetilde{y} + 1}$, where
$p = \per(\widetilde{Y}[1 \dd 3\tau - 1]) = \per(Y[1 \dd
3\tau - 1])$, i.e., $\BWT_{2\ell}[\widetilde{y} \dd \widetilde{y}']$
is a unary string consisting of the symbol that depends only on $Y$.

With this in mind, whenever during the
processing of ${\rm RL}(\BWT_{\ell})$ in the algorithm of
\cref{sec:bwt-nonperiodic} we encounter a run $\BWT_{\ell}[y \dd y'] =
c^{\delta} \in \mathbb{N}^{+}$ such that $\S \cap [c \dd c + \tau) =
\emptyset$ (to recognize such runs, we can simply check if $\S \cap
[c \dd c + \tau) = \emptyset$, since by definition of ${\rm
BWT}_{\ell}$, the occurrence $\T[c \dd c + \ell)$ is the leftmost
in $\T$), we set (letting $Y = \T[c \dd c + \ell)$)
\[
  \BWT_{2\ell}[y \dd y'] = Y[\per(Y[1 \dd 3\tau - 1])]^{y'-y+1}.
\]

The resulting $\BWT_{2\ell}$ is correct, except for $\SA$ ranges
corresponding to $\widetilde{Y} \in \mathcal{F}_{2\ell}'$.  Since
there are at most $|{\rm comp}_7(\S)| = \bigO(z)$ such
$\widetilde{Y}$, we can individually find each such range and correct
the corresponding symbols of $\BWT_{2\ell}$.

In the rest of the section we focus on implementing the correction
algorithm. Observe first that to compute $\mathcal{F}_{2\ell}'$, it
suffices to iterate through positions in ${\rm comp}_9(\S)$. By the
density condition, whenever for adjacent elements $s_i < s_{i+1}$ we
have $s_{i+1}-s_i>\tau$, then $s_{i}+1 \in \R'$. For each such $s_i$,
we collect a symbol-string pair $(\T[s_i]$, $\T^{\infty}[s_i + 1 \dd
s_i + 1 + 2\ell))$. The resulting collection contains $\bigO(z)$
pairs.  If we now sort it according to the strings, we can easily tell
whether the range $\BWT[\widetilde{y} \dd \widetilde{y}']$
corresponding to each $\widetilde{Y} \in \mathcal{F}_{2\ell}'$ is
uniform or not. If so, we know the preceding symbol and the frequency
of $\widetilde{Y}$ is obtained using \cref{th:index-count}. Otherwise
we only need the frequency. The main challenge in the correction
phrase is to compute the starting index $\widetilde{y}$ of the range
in $\SA$ corresponding to each $\widetilde{Y} \in
\mathcal{F}_{2\ell}'$. More precisely, for each $X \in
\mathcal{F}_{2\ell}'$ it suffices to compute a \emph{local rank}
$r_{X} = |{\rm pos}(\widetilde{Y})|$, where\footnote{To ease the
notation, in the rest of this section we use $X$ to denote a generic
string from $\mathcal{F}_{2\ell}'$, but remark that such $X$
corresponds to $\widetilde{Y}$ in Section~\ref{sec:bwt-nonperiodic}.}
\[
  {\rm pos}(X) := {\Big \{} j \in \R : \T^{\infty}[j \dd j + 2\ell)
  \prec X \text{ and } \T^{\infty}[j \dd j + \ell) = X[1 \dd \ell]
  {\Big \}}.
\]

The outline of the algorithm for computing all $r_{X}$ is similar
to~\cite[Section 6.1.2]{sss}. However, nearly every step needs to be
redesigned to use only $\bigO(z \polylog n)$ space.  Motivated by
\cref{lm:periodic}, we focus on the properties of runs of consecutive
positions in $\R$. We first partition such runs into classes where the
computation can be done independently.

\paragraph{Equivalent Runs}
We say that two positions $i, j \in \R$ are \emph{$\R$-equivalent} if
there exists a primitive string $H$ of length $|H| \leq
\frac{1}{3}\tau$ such that both $\T[i \dd i + \tau)$ and $\T[j \dd j +
\tau)$ are substrings of $H^{\infty}$. It is easy to see that if $i$
and $j$ are $\R$-equivalent, then there is exactly $\per(\T[i \dd i +
\tau)) = \per(\T[j \dd j + \tau))$ choices for $H$.  For each
equivalence class of this relation, we designate a unique $H \in
\Sigma^+$ and call it the \emph{$\R$-root} of $j \in [i]$, denoting
$\R\text{-root}(j) = H$. Then, $i$ is $\R$-equivalent to $j$ if and
only if $\R\text{-root}(i) = \R\text{-root}(j)$.

Let us now explain how we choose $\R$-roots. Let $j \in \R$ and let $p
= \per(\T[j \dd j + \tau))$. We define $\R\text{-root}(j) := \T[j' \dd
j' + p)$, where $j' = \min \{ j'' \in \R : j'' \text{ is }
\R\text{-equivalent to } j \}$.  Crucially, in the above definition it
always holds $j' \in \R'$.  This allows us to efficiently compute the
$\R\text{-root}$ for any $j \in \R$ as follows. Let $(F_1, \ldots,
F_k)$ be the sequence such that $\{F_1, \ldots, F_k\} =
\mathcal{F}_{2\tau}'$ and $\lpos(F_1) < \ldots < \lpos(F_k)$. Consider
the string $\T_{\rm root}$ defined as follows:
\[
  \T_{\rm root} := \T \cdot \left( \bigodot_{i = 0}^{k - 1} F_{k - i}
  \cdot \#_{i} \right),
\]
where $\#_i$ are distinct sentinel symbols not occurring in $\T$.

Let $j \in \R$ and assume we computed $p = \per(\T[j \dd j + \tau))$
using \cref{thm:period} (note that here we are guaranteed that $p \leq
\frac{\tau}{3}$ thus computing the exact period is equivalent to a
2-period query).  To compute $\R\text{-root}(j)$, we first find the
rightmost occurrence $\T_{\rm root}[j' \dd j' + \tau)$ of $\T[j \dd j
+ \tau)$ in $\T_{\rm root}$. Clearly, we must have $j' > |T|$. We then
binary search $b$ and $e$ such that $\T[j' \dd j' + \tau)$ occurs
inside $\T_{\rm root}[b \dd e) = F_t$ for some $t \in [1 \dd k]$ and
obtain $\R\text{-root}(j) = \T_{\rm root}[b \dd b + p)$.  Moreover,
$\Delta := (p + b - i_{\rm left}) \bmod p < p$ satisfies
$\R\text{-root}(j) = \T[j + \Delta \dd j + \Delta + p)$.

\begin{lemma}\label{lm:root}
  Given the LZ77 parsing and a string synchronizing set $\S$ of text
  $\T$ satisfying $|{\rm comp}_{9}(\S)| = \bigO(z)$, we can in
  $\bigO(z \log^4 n)$ time construct a data structure, than given $j
  \in \R$, returns in $\bigO(\log^4 n)$ time a value $\Delta < p =
  \per(\T[j \dd j + \tau)) \leq \frac13 \tau$ satisfying
  $\R\text{-root}(j) = \T[j + \Delta \dd j + \Delta + p)$.
\end{lemma}
\begin{proof}
  The sequence $(F_1, \ldots, F_k)$ is easily obtained using ${\rm
  comp}_{9}(\S)$ and \cref{th:index-lce,th:index-leftmost-occ}. We
  then construct the LZ77 parsing of $\T_{\rm root}$ by taking the
  parsing of $\T$ and appending $2k$ phrases. We then feed the
  resulting parsing to \cref{th:index-rightmost-occ}.
\end{proof}

The usefulness of $\R$-roots is due to the following two
properties. First, if $j \in \R \setminus \R'$ then $\R\text{-root}(j)
= \R\text{-root}(j - 1)$. Second, if $X = \T^{\infty}[j \dd j +
2\ell)$ where $j \in \R'$, then every $j' \in {\rm pos}(X)$ satisfies
$\R\text{-root}(j') = \R\text{-root}(j)$. Thus, we can first partition
the set of runs in $\R$ according to $\R$-roots and then, to compute
${\rm pos}(X)$, it suffices to search through $j'$ with the same
$\R\text{-root}$ as $j$.

\paragraph{Equivalent Positions}
For $j \in \R$, let $\rend{j} = \min \{ j' \geq j : j' \not \in \R \}
+ 3\tau - 2$.  Recall~\cite[Section 6.1.2]{sss} that runs of
consecutive positions in $\R$ (and hence in particular positions in
$\R'$) are easily identified as follows. Denote $\S = \{s_1, \ldots,
s_{n'} \}$, where $s_i < s_j$ if $i < j$, and let $s_0 = 0$, $s_{n' +
1} = n - 2\tau + 2$. Then, by density condition, it holds:
\[
  \R' = \{s_{i} + 1 : i \in [0 \dd n'] \text{ and } s_{i+1} - s_i
  > \tau \}.
\]
Furthermore, whenever $j - 1 = s_i$ for $j \in \R$, then $\rend{j} =
s_{i+1} + 2\tau - 1$. Thus, by~\cite[Fact 3.2]{sss}, for any $j \in
\R$, $\T[j \dd \rend{j})$ is the longest prefix of $\T[j \dd n]$
having period $p = \per(\T[j \dd j + 3\tau - 1))$.

For $j \in \R$ we define ${\rm type}(j) = +1$ if $\T[\rend{j}] \succ
\T[\rend{j} - p]$ and ${\rm type}(j) = -1$ otherwise, where $p =
\per(\T[j \dd \rend{j}))$.  Similarly as for $\R\text{-root}$, if $j
\in \R \setminus \R'$, then ${\rm type}(j) = {\rm type}(j - 1)$.
Moreover, if $X = \T^{\infty}[j \dd j + 2\ell) \subseteq
\mathcal{F}_{2\ell}'$ (where $j \in \R'$) is such that $\rend{j} - j <
2\ell$ and ${\rm type}(j) = -1$ then ${\rm type}(j') = -1$ holds for
all $j' \in {\rm pos}(X)$.  With the above observations in mind, let
$\R^{-} = \{j \in \R : {\rm type}(j) = -1\}$, $\R^{+} = \R \setminus
\R^{-}$, $\R'^{-} = \R' \cap \R^{-}$ and $\R'^{+} = \R' \cap
\R^{+}$. Moreover, let $\mathcal{F}'^{-}_{k}:= \{ \T^{\infty}[j \dd j
+ k) : j \in \R'^{-} \text{ and } \rend{j} - j < k \}$. The set
$\mathcal{F}'^{+}_k$ is defined analogously.  The above observation
can then be stated as follows: if $X \in \mathcal{F}'^{-}_{2\ell}$
then any $j \in {\rm pos}(X)$ satisfies $j \in \R^{-}$. Note also that
any such $j$ satisfies $\rend{j} - j < 2\ell$. We can therefore focus
on computing $r_{X}$ only for $X \in \mathcal{F}'^{-}_{2\ell}$. The
set $\mathcal{F}'^{+}_{2\ell}$ is processed symmetrically, and strings
$X$ corresponding to $j \in \R'$ satisfying $\rend{j} - j \geq 2\ell$
are processed separately.

To efficiently use $\R$-roots as a mean of comparing equivalent runs
in $\R$, we classify individual positions within a run.  Let $H =
\R\text{-root}(j)$ for $j \in \R$. Observe that the following
\emph{$\R$-decomposition} $\T[j \dd \rend{j}) = H'H^k H''$ (where $k
\geq 1$, $H'$ is a proper prefix of $H$, and $H''$ is a proper suffix
of $H$) is unique. We call the triple $\R\text{-sig}(j) := (|H'|, k,
|H''|)$ the \emph{$\R$-signature} of $j \in \R$ and the value
$\R\text{-exp}(j) = k$ its \emph{$\R$-exponent}. Note that $|H'|$ in
the $\R$-signature is the value $\Delta$ computed in
\cref{lm:root}. It is hence easy to use \cref{lm:root} and
\cref{th:index-lce} to compute $\rend{j}$ and the $\R$-signature for
each $j \in \R$ in $\bigO(\log^4 n)$ time. Moreover, if $X =
\T^{\infty}[j \dd j + 2\ell)$ for $j \in \R'^{-}$, then it holds
$\R\text{-exp}(j') \leq \R\text{-exp}(j)$ for all $j' \in {\rm
pos}(X)$. Thus, letting
\begin{align*}
  r_{X}^{=} & := | \{ j' \in {\rm pos}(X) :
    \R\text{-exp}(j') =
    \R\text{-exp}(j) \} |, \\
  r_{X}^{<} & := | \{ j' \in {\rm pos}(X) :
    \R\text{-exp}(j') <
    \R\text{-exp}(j) \} |,
\end{align*}
we have $r_{X} = r_{X}^{=} + r_{X}^{<}$. We will compute the terms
separately.

\paragraph{Computing $r_{X}^{=}$}

Denote by $(z_i)_{i \in [1 \dd |\R'^{-}|]}$ a sequence containing all $j
\in \R'^{-}$ sorted first according to their $\R$-root, second (in
case of ties) according to $|H''|$ in their $\R$-signature, and third
(in case there are still ties) according to suffix $\T[\rend{j} \dd
n]$. We can use this sequence to compute $r_{X}^{=}$, where $X =
\T^{\infty}[j \dd j + 2\ell) \in \mathcal{F}_{2\ell}'^{-}$ for $j \in \R'^{-}$,
as follows.
  
Let $X = X'X''$ where $|X'| = \rend{j} - j < 2\ell$ (by the
consistency of $\S$, the decomposition does not depend on the choice
of $j$). Then, let $i' \in [1 \dd |\R'^{-}|]$ be the smallest index
such that $j' = z_{i'}$ satisfies $\R\text{-root}(j') =
\R\text{-root}(j)$ and $X''$ is a prefix of $\T[\rend{j'} \dd n]$ (at
least one such $i'$ exists since $j$ occurs in $(z_i)_{i \in [1 \dd
|\R'^{-}|]}$). By definition of the sequence $(z_i)_{i \in [1 \dd
|\R'^{-}|]}$, $i'$ can be computed using binary search and LCE
queries (\cref{th:index-lce}).  The value $r_{X}^{=}$ is then equal to
the number of indices $i'' \in [1 \dd i')$ such that $j' = z_{i''}$
satisfies $\R\text{-root}(j') = \R\text{-root}(j)$ and the factor
$H'H^k$ in the $\R$-decomposition of $j'$ is at least as long as for
$j$. To compute the number of such $i'$ we store a collection of 2D
points containing a point for each $j \in \R'^{-}$ with its position in
$(z_i)_{i \in [1 \dd |\R'^{-}|]}$ as an $x$-coordinate and $|H'H^k|$
in its $\R$-decomposition as its $y$-coordinate. Computing $r_{X}^{=}$
then (as explained above) corresponds to a 3-sided orthogonal range
counting query~\cite{chazelle}.

The above algorithm requires storing and searching sets of size
$|\R'^{-}|$ which can be as large as $\Omega(\frac{n}{\tau})$.  We
observe, however, that the used sequences can be compressed, and
furthermore, a compressed representation of this (or more precisely,
sufficiently similar) sequence can be computed quickly, given ${\rm
comp}_{9}(\S)$ and the LZ77 parsing of $\T$.  The key observation
is that the relevant for the algorithm information about each $j \in
\R'^{-}$ is the length-$\Theta(\ell)$ substring of $\T$ located around
position $\rend{j}$. More precisely, we first note that since $|X''|
\leq \ell + 3 \leq 2\ell$, to compute $i'$ we only access substrings
$\T^{\infty}[\rend{i} \dd \rend{i} + 2\ell)$ where $i \in \R'^{-}$. On
the other hand, since $|X'| < 2\ell$, to determine the $\R$-root and
select all $j' \in \R'^{-}$ for which the factor $H'H^k$ in the
$\R$-decomposition is at least as long as for $j$, it suffices to know
$\T[\rend{i} - 2\ell \dd \rend{i})$, $i \in \R'^{-}$. Thus, rather
than $(z_i)_{i \in [1 \dd |\R'^{-}|]}$, the above algorithm can be
executed on the sequence of strings $Z[1 \dd |\R'^{-}|]$ defined as
follows. For $i \in [1 \dd |\R'^{-}|]$,
\[
  Z[i] := \T^{\infty}[\rend{z_i} - 2\ell \dd \rend{z_i} + 2\ell).
\]

Using $Z$ directly is problematic, since ${\rm RL}(Z)$ may have
multiple runs of the same string and hence we may have $|{\rm RL}(Z)|
= \omega(z\,{\rm polylog}\,n)$. We note, however, that $|\mathcal{Z}|
= \bigO(z)$ holds, where $\mathcal{Z} := \{ Z[i] : i \in [1 \dd
|\R'^{-}|]\}$.  To see this, note that if $\T^{\infty}[i \dd i +
|Z|) = Z$ for some $Z \in \mathcal{Z}$, then by consistency
condition $i + 2\ell = \rend{i'}$ for some $i' \in \R$.  Similarly, if
$j \in \R$, then $\rend{j} - 2\tau + 1 \in \S$.  Thus, $i + 2\ell -
2\tau + 1 \in \S$. In particular, $j' = \lpos(Z)$ satisfies $j' +
2\ell - 2\tau + 1 \in \S$. This mapping is injective, and moreover,
since $|Z| = 4\ell$ and $2\ell \leq 7\tau$, we have
$j' + 2\ell - 2\tau + 1 \in {\rm comp}_{9}(\S)$.

Combined with the observation that two equal elements of $Z$ are
always either both included or both excluded, when computing
$r_{X}^{=}$, we thus instead consider the string sequence $Z'$ defined
by $(z_i')_{i \in [1 \dd |\R'^{-}|]}$ containing all elements $j \in
\R'^{-}$ sorted first according to their $\R$-root, second according
to $|H''|$, third according to the substring $\T^{\infty}[\rend{j} \dd
\rend{j} + 2\ell)$, and finally according to the reversed substring
$\T^{\infty}[\rend{j} - 2\ell \dd \rend{j})$. Then, $|{\rm RL}(Z')| =
\bigO(z)$ holds, and ${\rm RL}(Z')$ can be computed using LZ77 of
$\T$, ${\rm comp}_{9}(\S)$, and
\cref{th:index-leftmost-occ,th:index-lce}.  The query algorithm
remains unchanged, except each 2D point now represents a run of
strings, and hence is augmented with the weight corresponding to its
frequency.

\begin{lemma}\label{lm:local-rank-1}
  Let $\T$ be a length-$n$ text. Given the LZ77 parsing and a string
  synchronizing set $\S$ of $\T$ satisfying $|{\rm comp}_{9}(\S)| =
  \bigO(z)$, the set $\{r_X^{=}\}_{X \in \mathcal{F}'^{-}_{2\ell}}$
  can be computed in $\bigO(z \log^4 n)$ time.
\end{lemma}

\paragraph{Computing $r_{X}^{<}$}

Let us start with the following inefficient algorithm. Consider
sorting all $j \in \R'^{-}$ first by $\R\text{-root}(j)$ and
then by $\R\text{-exp}(j)$. Let
\[
  \mathcal{H} := \{ \R\text{-root}(j) : j \in \R \}.
\]
For $H \in \mathcal{H}$, let $\R_{H}^{-} := \{ j \in \R^{-}
: \R\text{-root}(j) = H \}$ and $\R'^{-}_{H} := \R'
\cap \R_H^{-}$.

Let us fix some $H \in \mathcal{H}$. The algorithm processes all
elements of $\R'^{-}_{H}$ in rounds. In round $i$, we consider all
positions in the set $Q_i := \{ j \in \R'^{-}_H : \R\text{-exp}(j) = i
\}$. We execute rounds in increasing order, skipping $i$ for which
$Q_i = \emptyset$.  The algorithm maintains an array $C[0 \dd |H|)$
that satisfies the following invariant: at the end of round $i$,
$C[t] = | \{ j \in \R_{H,t}^{-} : \R\text{-exp}(j) \leq i \} |$,
where $ \R_{H,t}^{-} := \{ j \in \R_H^{-} : \R\text{-sig}(j) \in \{
t \} \times \mathbb{N}^2 \}$. Round $i$ proceeds as follows.
\begin{enumerate}
\item First, we look at $|Q_{i-1}|$ and if $Q_{i-1}=\emptyset$, we
  increment all $C[0 \dd |H|)$ by $(i - i_{\rm prev} - 1) \cdot
  \sum_{i' \geq i}|Q_{i'}|$, where $i_{\rm prev} = \max \{ i' < i :
  Q_{i'} \neq \emptyset \}$.  This accounts for skipped
  exponents. After this, it holds $C[t] = | \{ j \in \R_{H,t}^{-} :
  \R\text{-exp}(j) \leq i - 1 \} |$ for all $t \in [0 \dd |H|)$.
\item Now, iterate through $X = \T^{\infty}[j \dd j + 2\ell) \in
  \mathcal{F}'^{-}_{2\ell}$ satisfying $\R\text{-root}(j) = H$ and
  $\R\text{-exp}(j) = i$. For each such $j$, we answer $r_{X}^{<} =
  C[t] + \ldots + C[|H| - 1]$, where $\R\text{-sig}(j) \in \{ t \}
  \times \mathbb{N}^2$ (note that $\R\text{-sig}(j)$ does not depend
  on the exact choice of $j$). Since all considered $j$ satisfy $j \in
  Q_i$, there is nothing to do in this step when $Q_i =
  \emptyset$. Hence, skipping such $Q_i$ is correct.
\item Then, for $j \in Q_i$ add one to all counters in the range to
  $C[0 \dd t]$, where $\R\text{-sig}(j) \in \{ t \} \times
  \mathbb{N}^2$. This accounts for $j \in \R'^{-}_{H}$ having
  $\R\text{-exp}(j) = i$.
\item Finally, increment $C[0 \dd |H|)$ by $\sum_{i' > i}|Q_{i'}|$,
  accounting $j \in \R_{H}^{-} \setminus \R'^{-}_{H}$ having
  $\R\text{-exp}(j) = i$.
\end{enumerate}

If we represent $C$ using a balanced BST, each update in the algorithm
takes $\bigO(\log n)$ time.

The above algorithm processes each element $j \in \R'^{-}_{H}$
separately. To turn it into an efficient algorithm for compressible
strings (obtaining a runtime of the form $\bigO(z \, {\rm polylog} \,
n)$), we first note that the only used information for $j \in
\R'^{-}_H$ is $\R\text{-sig}(j)$.  Moreover, two elements with equal
$\R$-signatures are processed in the same way. We can thus process
them together, i.e., in step 3, rather than by one, we increment $C[0
\dd t]$ by the frequency of a given signature. Finally, we note that
all $X = \T^{\infty}[j \dd j + 2\ell)$ considered in step 2 satisfy
$\rend{j} - j < 2\ell$, thus we only need to execute rounds up to
$i_{\max} := \left \lceil 2\ell/|H| \right \rceil$. Consequently, the
algorithm remains correct if instead of using $\R\text{-sig}(j)$ for
all $j \in \R'^{-}$, we use \emph{truncated $\R$-signatures}
defined as:
\[
  \R\text{-sig}'(j) :=
  \begin{dcases*}
    \left ( |H'|, k, |H''| \right)
      & if $k < \left\lceil\tfrac{2\ell}{|H|}\right\rceil$,\\
    \left ( 0, \left \lceil \tfrac{2\ell}{|H|} \right \rceil,
        |H''| \right)
      & otherwise,
  \end{dcases*}
\]
where $\R\text{-sig}(j) = (|H'|, k, |H''|)$. The following lemma
ensures that this (assuming truncating signatures is combined with group
processing of equal signatures) significantly improves the runtime.
Note that we need to count separately for each $H \in \mathcal{H}$, as
(truncated) $\R$-signatures for different $\R$-roots can be
equal.

\begin{lemma}\label{lm:sig}
  \[
    \sum_{H \in \mathcal{H}} \left | \left \{ \R\text{-}{\rm
      sig}'(j) : j \in \R_{H}'^{-} \right \} \right | =
    \bigO(z).
  \]
\end{lemma}
\begin{proof}
  Let us fix $H \in \mathcal{H}$.

  Consider first $j \in \R_{H}'^{-}$ such that $\R\text{-sig}(j) \in
  \mathbb{N} \times \{ k \} \times \mathbb{N}$, where $k < \left\lceil
  2\ell/|H| \right\rceil$. We injectively assign $\R\text{-sig}'(j) =
  \R\text{-sig}(j)$ to an element of ${\rm comp}_8(\S)$. Let $i_{\rm
  left} = \lpos(\T[j - 1 \dd \rend{j} + 1))$ (taking one symbol past
  the periodic region ensures that the mapping is injective).  By the
  consistency of $\S$, it holds $i_{\rm left} \in \S$. Furthermore, by
  $k < \left\lceil 2\ell/|H| \right\rceil$, it holds $\rend{j} - j + 2
  \leq (k + 2)|H| \leq 2\ell + 2|H| \leq 8\tau$ (since $2\ell \leq
  7\tau$ and $|H| \leq \tfrac13 \tau$).  Thus, for some $t \in [1 \dd
  z]$, $i_{\rm left} \in \S \cap (e_t - 8\tau \dd e_t] \subseteq
  {\rm comp}_8(\S)$.  This map remains injective even considering all
  $H \in \mathcal{H}$.

  Consider now $j \in \R_{H}'^{-}$ such that $\R\text{-sig}(j) \in
  \mathbb{N} \times \{ k \} \times \{ |H''| \}$, where $k \geq
  \left\lceil 2\ell/|H| \right\rceil$.  We again construct an
  injective map of $\R\text{-sig}'(j)$ to ${\rm comp}_8(\S)$.  Denote
  $\ell_{\rm tail} = |H| \cdot \left \lceil 2\ell/|H| \right \rceil +
  |H''|$. Let $i_{\rm left} = \lpos(\T[\rend{j} - \ell_{\rm tail} \dd
  \rend{j} + 1))$. By the consistency condition of $\S$, for every $j
  \in \R$ it holds $\rend{j} - 2\tau + 1 \in \S$. Thus, since $i_{\rm
  left} \in \R$ and $\rend{i_{\rm left}} = i_{\rm left} + \ell_{\rm
  tail}$, it holds $i_{\rm left} + \ell_{\rm tail} - 2\tau + 1 \in
  \S$. We have $6\tau \leq 2\ell \leq \ell_{\rm tail} + 1 \leq 2\ell +
  2|H| \leq 8\tau$.  Thus, for some $t \in [1 \dd z]$, we have $e_t -
  8\tau < i_{\rm left} \leq e_t$ and so $i_{\rm left} + \ell_{\rm
  tail} - 2\tau + 1 \in \S \cap (e_t - 4\tau \dd e_t + 6\tau]
  \subseteq {\rm comp}_8(\S)$. The map remains injective considering
  all $H \in \mathcal{H}$.
\end{proof}\vspace{1ex}

Using the above lemma, the set of distinct truncated $\R$-signatures,
along with their frequencies, can be constructed from the LZ77 parsing
and ${\rm comp}_{9}(\S)$ using \cref{th:index-lce,th:index-count}.

\begin{lemma}\label{lm:local-rank-2}
  Let $\T$ be a length-$n$ text. Given the LZ77 parsing and a string
  synchronizing set $\S$ of $\T$ satisfying $|{\rm comp}_{9}(\S)| =
  \bigO(z)$, the set $\{r_X^{<}\}_{X \in \mathcal{F}'^{-}_{2\ell}}$
  can be computed in $\bigO(z \log^4 n)$ time.
\end{lemma}

\paragraph{Putting Things Together}

It remains to explain how to compute $r_X$ for $X = \T^{\infty}[j \dd j +
2\ell) \in \mathcal{F}_{2\ell}'$ where $j \in \R'$ and
$\rend{j} - j \geq 2\ell$.  We observe that for such $X$, $r_X^{=}$
and $r_X^{<}$ can be obtained during the computation of $\{r_X^{=} : X
\in \mathcal{F}'^{-}_{2\ell}\}$ and $\{r_X^{<} : X \in
\mathcal{F}'^{-}_{2\ell}\}$.

More precisely, after completing the last round during the processing
of some contiguous subsequence of all $j' \in \R'^{-}_H$, we can
set, for every $X = \T^{\infty}[j \dd j + 2\ell) \in \mathcal{F}_{2\ell}'$,  
such that $\R\text{-root}(j) = H$ and $\rend{j}-j \geq 2\ell$,
$r_X^{<} := C[t] + \ldots + C[|H| - 1]$, where $\R\text{-sig}(j)
\in \{ t \} \times \mathbb{N}^2$.  Similarly, computing $r_X^{=}$ for
each $X = \T^{\infty}[j \dd j + 2\ell) \in \mathcal{F}_{2\ell}'$ with $e_j -
j \geq 2\ell$ only requires one query on the 2D orthogonal range 
counting data structure.  Thus, handling all $X \in
\mathcal{F}_{2\ell}' \setminus \left ( \mathcal{F}'^{-}_{2\ell} \cup
\mathcal{F}_{2\ell}'^{+} \right)$ is not more expensive than handling
$\mathcal{F}_{2\ell}'^{-}$ and $\mathcal{F}'^{+}_{2\ell}$.

By combining \cref{pr:bwt-nonperiodic} with
\cref{lm:local-rank-1,lm:local-rank-2}, we therefore
obtain the following result and consequently the main result of this
section.

\begin{proposition}\label{pr:bwt}
  Let $\T$ be a string of length $n$, and let $\ell = 2^q$ be such
  that $q \in [4 \dd \lceil{\log n}\rceil)$. Then, given ${\rm
  RL}(\BWT_{\ell})$ and the LZ77 parsing of $\T$, the sequence
  ${\rm RL}(\BWT_{2\ell})$ can be constructed in $\bigO((r + z)
  \log^5 n)$ time.
\end{proposition}

By \cref{pr:bwt-base} we can compute $\BWT_{\ell}$ for $\ell=2^q$ and
$q<4$ in $\bigO(z \log^4 n)$ time. For $q \geq 4$, we use \cref{pr:bwt}.
Thus, by the upper bound $r = \bigO(z \log^2 n)$ from
\cref{th:basic-upper-bound}, we obtain the main result of this
section.

\begin{theorem}\label{th:bwt}
  There exists a Las-Vegas randomized algorithm that, given the LZ77
  parsing of a text $\T$ of length~$n$, computes its run-length
  compressed Burrows--Wheeler transform in $\bigO((r + z) \log^6 n) =
  \bigO(z \log^8 n)$ time.
\end{theorem}

\section{Auxiliary Recompression-Based Data Structures}\label{sec:aux}

\newcommand{\N}{\mathcal{N}}
\renewcommand{\S}{\mathcal{S}}
\newcommand{\G}{\mathcal{G}}
\newcommand{\Tr}{\mathcal{T}}
\newcommand{\rhs}{\mathsf{rhs}}
\newcommand{\height}{\mathsf{height}}

For a context-free grammar $\G$, we denote by $\Sigma_\G$ and $\N_\G$,
the set of non-terminals and the set of terminals, respectively. The
set of \emph{symbols} is $\S_\G:=\Sigma_\G\cup \N_\G$.  If the grammar
$\G$ is clear from context, the subscripts $\G$ might be omitted.

A \emph{straight-line grammar} (SLG) is a context-free grammar $\G$
such that:
\begin{itemize}
\item each non-terminal $A\in \N$ has a unique production $A\to
  \rhs(A)$, where $\rhs(A)\in \S^*$,
\item the set of symbols $\S$ admits a partial order $\prec$ such that
  $B \prec A$ if $B$ appears in $\rhs(A)$.
\end{itemize}
A simple inductive argument shows that, for each symbol $A\in \S$, the
language $L(A)$ consists of a unique string, which we call the
\emph{expansion} of $A$ and denote $\exp(A)$.  In particular, the
expansion $\exp(S)$ of a starting symbol $S\in \S$ is the unique
string \emph{represented} by $\G$.

The \emph{parse tree} $\Tr(A)$ of a symbol $A\in \S$ is a rooted
ordered tree with each node $\nu$ associated to a symbol $s(\nu)\in
\S$.  The root of $\Tr(A)$ is a node $\rho$ with $s(\rho)=A$.  If $A
\in \Sigma$, then $\rho$ has no children.  If $A\in \N$ and
$\rhs(A)=B_1\cdots B_k$, then $\rho$ has $k$ children, and the subtree
rooted at the $i$th child is (a copy of) $\Tr(B_i)$.  The
\emph{height} $\height(A)$ of a symbol $A\in \S$ is defined as the
height of its parse tree $\Tr(A)$.  In other words, $\height(A)=0$ if
$A\in \Sigma$ and $\height(A)=1+\max_{i=1}^k \height(B_i)$ if
$\rhs(A)=B_1\cdots B_k$. The parse tree $\Tr_\G$ of an SLG $\G$ is
defined as the parse tree $\Tr(S)$ of the starting symbol $S$, and the
height of $\G$ is defined as the height of $S$.

Each node $\nu$ of $\Tr(A)$ is associated with a fragment $\exp(\nu)$
of $\exp(A)$ matching $\exp(s(\nu))$.  For the root $\rho$, we define
$\exp(\rho)=\exp(A)[1\dd |\exp(A)|]$ to be the whole $\exp(A)$.
Moreover, if $\exp(\nu)=\exp(A)[\ell\dd r)$, $\rhs(s(\nu))=B_1\cdots
B_k$, and $\nu_1,\ldots,\nu_k$ are the children of $\nu$, then
$\exp(\nu_i)=\exp(A)[r_{i-1}\dd r_{i})$, where $r_i = \sum_{j=1}^{i}
|\exp(B_j)|$ for $0\le i \le k$.  This way, the fragments
$\exp(\nu_i)$ form a partition of $\exp(\nu)$, and $\exp(\nu_i)$
matches $\exp(s(\nu_i))$ (as claimed).

Without loss of generality, we assume that each symbol $A\in \S$
appears as $s(\nu)$ for a node $\nu$ of $\Tr_\G$; the remaining
symbols can be removed from $\G$ without affecting the string
generated by~$\G$.

\paragraph*{Straight-Line Programs}
We say that a straight-line grammar $\G$ is in \emph{Chomsky normal
form} (CNF) if $|\rhs(A)|=2$ holds for each $A\in \N$. A
straight-line grammar in CNF is called a \emph{straight-line program}
(SLP).  An SLP $\G$ of size $g$ (with $g$ symbols) representing a text
$T$ of length $n$ can be stored $\Oh(g)$ space ($\Oh(g \log n)$ bits)
with each non-terminal $A\in \N$ storing $\rhs(A)$ and $|\exp(A)|$.
This representation allows for efficiently traversing the parse
tree~$\Tr_\G$: given a node $\nu$ represented as a pair
$(s(\nu),\exp(\nu))$, it is possible to retrieve in constant time an
analogous representation of a child $\nu_i$ of $\nu$ given its index
$i\in\{1,2\}$ (among the children of $\nu$) or an arbitrary position
$T[j]$ contained in $\exp(\nu_i)$.

Rytter~\cite{Rytter03} provided an efficient algorithm that converts
the LZ77 representation of a string into an SLP generating it.
Unfortunately, he assumed a weaker \emph{non-self-referential} variant
of LZ77, where $T[i\dd i+\ell)$ is a previous factor only if there
exists $i'\in [1\dd i-\ell]$ with $\LCE(i,i')\ge \ell$.  In the
following theorem, we adapt his methods to the self-referential
variant allowing $i'\in [1\dd i)$.
\begin{theorem}\label{thm:rytter}
  Given an LZ77-like parsing of a string $T[1\dd n]$ into $f$ phrases,
  an SLP $\G$ of size $\Oh(f \log n)$ and height $\Oh(\log n)$
  generating $T$ can be constructed in $\Oh(f\log n)$ time.
\end{theorem}
\begin{proof}
  Rytter~\cite[Section 3]{Rytter03} defined \emph{AVL grammars} as
  SLPs satisfying the following extra condition: if $\rhs(A)=BC$ for
  $A\in \N$, then $|\height(B)-\height(C)|\le 1$.  This
  guarantees~\cite[Lemma 1]{Rytter03} that $\height(A) = \Oh(\log
  |\exp(A)|)$ holds for every $A\in \S$.

  The algorithm of~\cite{Rytter03} builds $\G$ incrementally: each
  step involves adding a symbol $A$ with a desired expansion
  $\exp(A)$, as well as a bounded number of auxiliary symbols. In the
  last step, the starting symbol $S$ with $\exp(S)=T$ is added.  Final
  post-processing includes pruning the grammar to remove symbols not
  occurring in $\Tr_{\G}$.  Each step is of one of three kinds:
  \begin{enumerate}[label={\rm(\alph*)}]
  \item\label{it:new} A new terminal symbol can be added to $\G$ in
    $\Oh(1)$ time (along with no auxiliary symbols).
  \item\label{it:concat} Given two symbols $B,C\in \S$, a new symbol
    $A$ with $\exp(A)=\exp(B)\exp(C)$ can be added to $\G$ in
    $\Oh(1+|\height(B)-\height(C)|)$ time along with
    $\Oh(|\height(B)-\height(C)|)$ auxiliary symbols~\cite[Lemma
    2]{Rytter03}.
  \item\label{it:extract} Given a symbol $A\in \S$ and two positions
    $1\le i \le j \le |\exp(A)|$, a new symbol $B$ with
    $\exp(B)=\exp(A)[i\dd j]$ can be added to $\G$ in $\Oh(1+\log
    |\exp(A)|)$ time along with $\Oh(\log |\exp(A)|)$ auxiliary
    symbols~\cite[Lemma 3 and Theorem 2]{Rytter03}.
  \end{enumerate}

  Now, a non-self-referential LZ77-like parsing $T=F_1\cdots F_f$ can
  be processed by iteratively constructing symbols $A_j$ with
  $\exp(A_j)=F_1\cdots F_j$ for $j\in [1\dd f]$.  At each iteration
  $j$, a symbol $B_j$ with $\exp(B_j)=F_j$ is first constructed.  If
  $F_j=T[i\dd i+\ell)$ is a previous fragment represented by
  $(i',\ell)$ such that $i'\in [1\dd i-\ell]$ and $\LCE(i,i')\ge
  \ell$, then $F_j = (F_1\cdots F_{j-1})[i'\dd
  i'+\ell)=\exp(A_{j-1})[i'\dd i'+\ell-1]$, so $B_j$ can be
  constructed using operation~\ref{it:extract}.  Otherwise $F_j=T[i]$
  is a single character (not occurring in $F_1\cdots F_{j-1}$), so
  $B_j$ can be constructed using operation~\ref{it:new}.  Finally,
  $A_{j}$ is obtained based on $A_{j-1}$ and $B_j$ using
  operation~\ref{it:concat}.  Consequently iteration $j$ involves
  $\Oh(\log |F_1\cdots F_j|)=\Oh(\log n)$ new symbols and can be
  implemented in $\Oh(\log n)$ time, for a total of $\Oh(f \log n)$
  symbols and $\Oh(f\log n)$ time across all iterations.

  In order to apply the same scheme for a (potentially)
  self-referential LZ77-like parsing, we need to specify the
  construction of $B_j$ for a previous fragment $F_j=T[i\dd i+\ell)$
  represented with $i'\in (i-\ell\dd i)$ such that $\LCE(i,i')\ge
  \ell$.  Note that $P:= T[i'\dd i)$ is then a string period of $F_j$
  and, consequently, $F_j = P^{\infty}[1\dd \ell]$.  Thus, we first
  use operation~\ref{it:extract} to construct a symbol $P_{j,0}$
  representing $P$; this costs $\Oh(1+\log |F_1\cdots
  F_{j-1}|)=\Oh(\log n)$ time and auxiliary symbols.  Next, for $k\in
  [1\dd \lceil \log \frac{\ell}{|P|}\rceil]$, we use
  operation~\ref{it:concat} to construct symbols $P_{j,k}$
  representing $P^{2^k}$.  Note that each application of
  operation~\ref{it:concat} involves $\Oh(1)$ time and no auxiliary
  symbols, for a total of $\Oh(\log \frac{\ell}{|P|})=\Oh(\log
  \ell)=\Oh(\log n)$ time and symbols.  Finally, we derive $B_j$ using
  operation~\ref{it:extract} based on $F_j = P^\infty[1\dd
  \ell]=\exp(P_{j,k'})[1\dd \ell]$, where $k' = \lceil \log
  \frac{\ell}{|P|} \rceil$; this costs $\Oh(\log |P^{2^{k'}}|)=\Oh(k'+
  \log |P|) =\Oh(\log \ell)=\Oh(\log n)$ time and symbols.  Overall,
  iteration $j$ still costs $\Oh(\log n)$ time and symbols.
\end{proof}

\paragraph*{Run-Length Straight-Line Programs}

A \emph{run-length straight-line program} (RLSLP) is a straight-line
grammar $\G$ whose non-terminals can be classified into \emph{pairs}
$A \to BC$, where $B,C\in \S$ and $B \ne C$, and \emph{powers} $A\to
B^k$, where $B\in \S$ and $k \ge 2$ is an integer.  Analogously to an
SLP, an RLSLP of size $g$ (with $g$ symbols) representing a text $T$
of length $n$ can be stored in $\Oh(g)$ space ($\Oh(g \log n)$ bits)
allowing efficient traversal of the parse tree $\Tr_\G$.

\paragraph*{Recompression}

\emph{Recompression}~\cite{Jez2016} is a technique of computing a
small and \emph{locally consistent} RLSLP $\G$ representing a given
text $T$. Recompression is based on a sequence of strings
$T_1,\ldots,T_r\in \S^*$ such that $T_r=S$, where $S$ is the starting
symbol, the string $T_{j-1}$ for $2 \le j \le r$ can be obtained from
$T_j$ by replacing some non-terminals $A$ in $T_j$ with $\rhs(A)$, and
$T_1 = T$.

Recompression proceeds iteratively starting from $T_1=T$.  As long as
$|T_j|>1$, the algorithm partitions $T_j$ into blocks using one of the
following two schemes:
\begin{itemize}
  \item \textbf{Run-length encoding:} If $j$ is odd, then $T_j$ is
    partitioned into maximal blocks consisting of equal symbols
    (runs).
  \item \textbf{Alphabet partitioning:} If $j$ is even, then the set
    of symbols occurring in $T_j$ is decomposed into \emph{left}
    symbols and \emph{right} symbols, respectively (implementation of
    this process differs between variants of recompression).  Whenever
    $T_j[i]$ is a left symbol and $T_j[i+1]$ is a right symbol,
    $T_j[i\dd i+1]$ forms a length-2 block. The remaining positions
    form length-1 blocks.
\end{itemize}
Then, for blocks $T_j[\ell \dd r]$ of length at least 2, new
non-terminals $A$ with $\rhs(A)=T_j[\ell \dd r]$ are created. This
process is implemented so that matching blocks get the same
non-terminal.  Finally, $T_{j+1}$ is obtained from $T_j$ by
\emph{collapsing} each block $T_j[\ell \dd r]$ of length at least 2 to
the corresponding non-terminal.  When $|T_j|=1$, the procedure
terminates (resulting in $r:=j$), and the only symbol of $T_j$ is
declared to be the starting symbol of $\G$.  We then say that $\G$ is
an \emph{$r$-round recompression RLSLP}, noting that the height of
$\G$ does not exceed $r-1$.

A recompression RLSLP can be constructed efficiently based on any
LZ77-like representation.

\begin{theorem}\label{thm:recompression}
  There exists an algorithm that, given an LZ77-like parsing of a
  string $T[1\dd n]$ into $f$ phrases, constructs in $\Oh(f\log^2 n)$
  time an $\Oh(\log n)$-round recompression RLSLP $\G$ of size
  $\Oh(f\log^2 n)$ generating $T$.
\end{theorem}
\begin{proof}
  Given the LZ77-like parsing of $\Tr$, the algorithm first uses
  \cref{thm:rytter} to represent $T$ as an SLP $\hat{\G}$ of size
  $|\S_{\hat{\G}}|=\Oh(f \log n)$.  Then, using the algorithm of
  Jeż~\cite{Jez2015,Jez2016a} (see also~\cite[Section
  4]{tomohiro-lce}), $\hat{\G}$ is converted into an $\Oh(\log
  n)$-round recompression RLSLP $\G$ of size $\Oh(|\S_{\hat{\G}}| \log
  n)=\Oh(f\log^2 n)$.  The running time of the two steps is $\Oh(f
  \log n)$ and $\Oh(f\log^2 n)$, respectively.
\end{proof}

\paragraph*{LCE Queries}
I~\cite{tomohiro-lce} showed that $\LCE$ queries can be answered in
$\Oh(r)$ time in a string $T$ represented with an $r$-round
recompression RLSLP.  Since a recompression RLSLP generating the
reversed string $\bar{T}$ can be obtained by reversing $\rhs(A)$ for
each non-terminal $A$, the same holds for $\LCE$ queries in $\bar{T}$.
Due to \cref{thm:recompression}, this yields the following result.
\begin{theorem}\label{th:index-lce}
  Given an LZ77-like parsing of a string $T[1 \dd n]$ into $f$
  phrases, we can in $\bigO(f \log^2 n)$ time construct a data
  structure that supports $\LCE$ queries in $T$ and in its reverse
  $\bar{T}$ in $\bigO(\log n)$ time.
\end{theorem}

\subsection{Pattern Matching in RLSLPs}

Given a pattern $P$ and a text $T$, a fragment $T[j\dd j+|P|)$
matching $P$ is called an \emph{occurrence} of $P$ in $T$.  We denote
the set of starting positions of occurrences of $P$ in $T$ by
$\Occ(P,T)=\{j \in [1\dd |T|-|P|+1] : P=T[j\dd j+|P|)\}$.  In this
section, we characterize the occurrences of $P$ in $T$ based on the
parse tree $\Tr$ of an RLSLP $\G$ representing $T$.  The underlying
concepts originate from multiple works on pattern matching in
compressed and dynamic
strings~\cite{Alstrup2000,Gawrychowski2015,NishimotoDAM,Christiansen2019}.

Let $x=T[\ell \dd r)$ be a non-empty fragment of $T$.  The \emph{hook}
of $x$, denoted $\hook(x)$, is the lowest node $\nu$ in $\Tr$ such
that $\exp(\nu)$ contains $x$.  Equivalently, $\hook(x)$ is the lowest
common ancestor of the leaf representing $T[\ell]$ and the leaf
representing $T[r-1]$ in $\Tr$.  For $|x|>1$, the \emph{anchor} of
$x$, denoted $\anch(x)$, is the length $a$ of the longest prefix $y$
of $x$ such that $\hook(y)\ne \hook(x)$.  For $|x|=1$, we define
$\anch(x)=0$.  For a node $\nu$ of $\Tr$ and an integer $a\in [0\dd
|P|)$, let $\Occ(P,\nu,a) = \{j \in \Occ(P,T) : \hook(T[j\dd
j+|P|))=\nu \text{ and } \anch(T[j\dd j+|P|)) = a\}$. The following
observation characterizes this set.
\begin{observation}\label{obs:occ}
  Let $\nu$ be node of $\Tr$ with $s(\nu)=A$ and $\exp(\nu)=T[\ell\dd
  r)$, and let $P=P_L\cdot P_R$ be a non-empty pattern. Then
  $\Occ(P,\nu,|P_L|)\ne \emptyset$ if and only if one of the
  following holds:
  \begin{enumerate}
  \item $A\in \Sigma$, $P_L=\eps$, and $P_R=A$. Then
    $\Occ(P,\nu,|P_L|)=\{\ell\}$.
  \item $A\to BC$, $P_L \ne \eps$ is a suffix of $\exp(B)$, and
    $P_R\ne \eps $ is a prefix of $\exp(C)$.  Then
    $\Occ(P,\nu,|P_L|)=\{\ell+|\exp(B)|-|P_L|\}$.
  \item $A\to B^k$, $P_L \ne \eps$ is a suffix of $\exp(B)$, and $P_R
    \ne \eps$ is a prefix of $\exp(B)^{k-1}$.  Then $\Occ(P,\nu,
    |P_L|)=\big\{\ell+i|\exp(B)|-|P_L| : i \in \big[1\dd
      k-\big\lceil{\frac{|P_R|}{|\exp(B)|}\big\rceil}\big]\big\}$.
  \end{enumerate}
\end{observation}

\newcommand{\A}{\mathsf{Anch}}

Local consistency of recompression RLSLPs implies that the occurrences
of any pattern have just a few different anchors and that a set of few
\emph{potential anchors} can be computed efficiently.
\begin{lemma}[{\cite[Corollary 10.4]{Gawrychowski2015},~\cite[Lemma 10]{tomohiro-lce}}]\label{lem:anch}
  For a pattern $P$ and a text $T$ represented by an RLSLP $\G$,
  define $\anch(P,T) = \{\anch(T[j\dd j+|P|)) : j \in \Occ(P,T)\}$. If
  $\G$ is an $r$-round recompression RLSLP, then $|\anch(P,T)| =
  \Oh(r)$. Moreover, given any $j\in \Occ(P,T)$, a superset
  $\A(P)\supseteq \anch(P,T)$ of size $|\A(P)|= \Oh(r)$ can be
  computed in $\Oh(r)$ time.
\end{lemma}

\paragraph*{Internal Pattern Matching Queries in LZ77-Compressed Strings}
For a (static) text $T$, internal pattern matching queries (IPM
queries)~\cite{Kociumaka2015} given two fragments $x,y$ of $T$ with
$|y|<2|x|$, ask for the occurrences of $x$ contained within $y$.  Due
to the assumption $|y|<2|x|$, the output $\Occ(x,y)$ can be
represented using a single arithmetic progression~\cite[Lemma
2.1]{Kociumaka2015}; see also~\cite{Breslauer1995,Plandowski1998}.
If $T$ is uncompressed, then IPM queries can be answered in $\Oh(1)$
time after $\Oh(n)$-time preprocessing~\cite{Kociumaka2015,phdtomek}.
Combining \cref{thm:recompression,th:index-lce,obs:occ,lem:anch}, we
can answer IPM queries in LZ77-compressed strings in $\Oh(\log^3 n)$
time.

\begin{theorem}\label{thm:ipm}
  Given an LZ77-like parsing of a string $T[1\dd n]$ into $f$ phrases,
  we can in $\Oh(f\log^2 n)$ time construct a data structure that
  supports IPM queries in $T$ in $\Oh(\log^3 n)$ time.
\end{theorem}
\begin{proof}
  Our data structure consists of a recompression RLSLP $\G$ generating
  $T$ (constructed using \cref{thm:recompression}), and a component
  for $\LCE$ queries in $T$ and in its reverse $\bar{T}$ (built using
  \cref{th:index-lce}).  Thus, the construction time is $\Oh(f \log^2
  n)$.

  Given $x=T[\ell_x\dd r_x)$ and $y=T[\ell_y \dd r_y)$, the query
  algorithm works as follows.  First, the algorithm applies
  \cref{lem:anch} to obtain a set $\A(x)$ of potential anchors.  Since
  $\G$ is an $\Oh(\log n)$-round recompression RLSLP, we have
  $|\A(x)|=\Oh(\log n)$, and this step takes $\Oh(\log n)$ time.
  Next, the algorithm identifies all nodes $\nu$ in the parse tree
  $\Tr$ for which $\exp(\nu)$ intersects $y=T[\ell_y\dd r_y)$ on at
  least $|x|$ positions.  Due to $|y|<2|x|$, for each of these nodes
  $\nu$, the fragment $\exp(\nu)$ contains position $T[\ell_y+|x|-1]$.
  Therefore, it suffices to traverse the path from the root of $\T$
  towards the leaf representing $T[\ell_y+|x|-1]$ as long as the
  intersection is at least $|x|$ positions.  Consequently, this step
  takes $\Oh(\log n)$ time, and it results in $\Oh(\log n)$ nodes
  $\nu$.  We call them \emph{potential hooks} because the hook any
  occurrence of $x$ contained in $y$ must be one of these nodes.

  For each potential anchor $a\in \A(x)$ and for each potential hook
  $\nu$, the algorithm constructs $\Occ(x,\nu,a)$ using
  \cref{obs:occ}.  Each of the three cases depending on $A:= s(\nu)$
  is easy to implement since $\exp(\nu)= T[\ell\dd r)$, $P_L =
  T[\ell_x \dd \ell_x+a)$, and $P_R = T[\ell_x +a\dd r_x)$ are all
  given as fragments of $T$.  The case of $A\in \Sigma$ is
  particularly simple and costs $\Oh(1)$ time.  The case of $A \to BC$
  requires computing the longest common suffix of $P_L$ and
  $\exp(B)=T[\ell\dd \ell+|\exp(B)|)$, and the longest common prefix
  of $P_R$ and $\exp(C)=T[\ell+|\exp(B)|\dd r)$, which reduces to an
  $\LCE$ query in $\bar{T}$ and $T$, respectively.  Similarly, the
  case of $A \to B^k$ requires computing the longest common suffix of
  $P_L$ and $\exp(B)=T[\ell\dd \ell+|\exp(B)|)$, and the longest
  common prefix of $P_R$ and $\exp(B)^{k-1}=T[\ell+|\exp(B)|\dd r)$.
  Summing up, retrieving $\Occ(x,\nu,a)$ for a potential anchor
  $a\in \A$ and potential hook $\nu$ costs $\Oh(\log n)$ time.

  Finally, the algorithm filters occurrences contained in $y$, that
  is, starting in $[\ell_y\dd r_y-|x|]$, shifts the starting positions
  by $\ell_y-1$ (so that they become relative to $y$), and takes the
  union across potential anchors $a\in \A(x)$ and potential hooks
  $\nu$.  This post-processing takes $\Oh(\log^2 n)$ time; the
  resulting set $\Occ(x,y)$ must form an arithmetic
  progression~\cite[Lemma 2.1]{Kociumaka2015}.
\end{proof}

Next, we show that 2-period queries~\cite{Kociumaka2012} can be
answered efficiently in grammar-compressed strings.  A 2-period query,
given a fragment $x$ of $T$, asks to return $\per(x)$ or to report
that $x$ is not periodic, that is $\per(x)>\frac12|x|$.  If $T$ is
uncompressed, then 2-period queries can be answered in $\Oh(1)$ time
after $\Oh(n)$-time
preprocessing~\cite{Kociumaka2015,phdtomek,Bannai2017} using
relatively simple tools.  Nevertheless, in grammar-compressed strings,
it is convenient to derive 2-period queries from IPM queries.

\begin{theorem}\label{thm:period}
  Given an LZ77-like parsing of a string $T[1\dd n]$ into $f$ phrases,
  we can in $\Oh(f \log^2 n)$ time construct a data structure that
  supports 2-period queries in $T$ in $\Oh(\log^3 n)$ time.
\end{theorem}
\begin{proof}
  The data structure consists of components for $\LCE$ queries
  (\cref{th:index-lce}) and IPM queries (\cref{thm:ipm}), both of
  which take $\Oh(f \log^2 n)$ to construct.  The query algorithm,
  given $x=T[\ell \dd r)$, makes an IPM query asking for the
  occurrences of $x[1\dd \lceil{\frac12|x|\rceil}]=T[\ell \dd
  \lceil{\frac{\ell+r}{2}\rceil})$ in $x[2\dd |x|]=T[\ell+1\dd r)$.
  If there are no occurrences, then the algorithm reports that $x$ is
  not periodic.  Otherwise, it retrieves the starting position $p$ of
  the leftmost occurrence and checks if $p$ is a period of $x$, that
  is whether $x[1+p\dd |x|]=x[1\dd |x|-p]$, using an
  $\LCE(\ell,\ell+p)$ query in $T$.  If this test succeeds, then the
  algorithm returns $p$; if the test fails, the algorithm reports that
  $x$ is not periodic. This procedure costs $\Oh(\log^3 n)$ time,
  dominated by the IPM query.

  The correctness follows from the fact that if $\per(x)\le
  \frac12|x|$, then the first two occurrences of $x[1\dd
  \lceil{\frac12|x|\rceil}]$ in $x$ start at positions $1$ and
  $\per(x)+1$.  The lack of the intermediate occurrences is due to the
  primitivity of $x[1\dd \per(x)]$, which occurs in $x[1\dd 2\per(x)]$
  just twice, as a prefix and as a suffix, by the synchronization
  property of primitive strings.
\end{proof}

\subsection{Indexing LZ77-Compressed Texts}

In this section, we describe an efficient index that, after
preprocessing an LZ77-compressed text $T$, given a pattern $P$,
represented by its arbitrary occurrence in $T$, efficiently reports
all the occurrences of $P$ in $T$ (returns $\Occ(P,T)$), finds the
leftmost occurrence (returns $\min \Occ(P,T)$), or counts the
occurrences (returns $|\Occ(P,T)|$).  While numerous indexes for
dictionary-compressed strings have been developed
(see~\cite{Gagie2020} for a recent overview), ours satisfies two rare
properties: it can be constructed efficiently without decompressing
$T$, and it lets the pattern $P$ to be given by its occurrence in $T$
(rather than in the plain representation). In particular, we are not
aware of any index answering counting queries and satisfying either
property. On the other hand, in this version of the manuscript we do
not optimize the $\polylog n$ factors in the running times.

\newcommand{\first}{\mathsf{first}}
\newcommand{\cnt}{\mathsf{count}}
\newcommand{\depth}{\mathsf{depth}}
\newcommand{\nodes}{\mathsf{nodes}}

We start with two auxiliary results, and then we proceed to describing
our index.  For an RLSLP $\G$ with parse tree $\Tr$ and a symbol $A\in
\S$, let $\nodes(A)$ denote the set of nodes $\nu$ in $T$ with
$s(\nu)=A$.  Moreover, $\depth(\nu)$ denotes the number of nodes on
the path from $\nu$ to the root.

\begin{lemma}\label{lem:rlslp}
  Given an RLSLP $\G$, we can in $\Oh(|\S|)$ time compute, for any
  symbol $A\in \S$, the leftmost node $\first(A)\in \nodes(A)$ and the
  number of nodes $\cnt(A)=|\nodes(A)|$.  Moreover, we can in
  $\Oh(|\S|)$ time construct a data structure that, for any symbol
  $A\in \S$, enumerates $\nodes(A)$ in $\Oh(\depth(\nu))$ time per
  each $\nu \in \nodes(A)$.
\end{lemma}
\begin{proof}
  To compute $\first(A)$ for each $A\in \S$, the algorithm performs a
  pre-order traversal of the parse tree $\Tr$.  Upon entering a node
  $\nu$, the algorithm retrieves $A=s(\nu)$ and checks if $\first(A)$
  has already been computed.  If so, then the algorithm ignores the
  subtree of $\nu$.  Otherwise, $\first(A)$ is set to $\nu$, and the
  algorithm continues traversing the subtree of $\nu$.  If $A\to B^k$,
  the algorithm visit only the first child of $\nu$, so that the total
  running time is $\Oh(|\S|)$.

  To compute $\cnt(A)$ for each $A\in \S$, the algorithm processes
  symbols in the $\succ$ order (if $B$ appears in $\rhs(A)$, then $B$
  is processed after $A$).  At the beginning, $\cnt(A)$ is initialized
  to 1 if $A$ is the starting symbol and to 0 otherwise.  Processing a
  non-terminal $A\in \Sigma$ is void.  Processing a terminal $A\to
  BC$, the algorithm adds $\cnt(A)$ to both $\cnt(B)$ and $\cnt(C)$.
  Processing a terminal $A\to B^k$, the algorithm adds $k\cdot
  \cnt(A)$ to $\cnt(B)$.  This procedure guarantees that the value
  $\cnt(A)$ becomes correct before $A$ is processed.  The running time
  is $\Oh(|\S|)$.

  Preprocessing for enumerating $\nodes(A)$ consists in listing, for
  each $A\in \S$, all the symbols $B$ for which $A$ appears in
  $\rhs(B)$.  The query algorithm, given $A\in \S$, recursively
  enumerates $\nodes(B)$ for each such symbol.  Next, for each node
  $\nu\in \nodes(B)$ and for each position $i$ where $A$ occurs in
  $\rhs(B)$, the algorithm adds the $i$th child $\nu_i$ of $\nu$ to
  $\nodes(A)$.  Since $\depth(\nu_i)=\depth(\nu)+1$, the amortized
  cost of retrieving $\nu_i\in \nodes(A)$ is $\Oh(\depth(\nu_i))$.
\end{proof}

\newcommand{\Pts}{\mathcal{P}}
\newcommand{\Xs}{\mathcal{X}}
\newcommand{\Ys}{\mathcal{Y}}
\newcommand{\Ws}{\mathcal{W}}
\newcommand{\rank}{\mathsf{rank}}
\begin{lemma}\label{lem:2drange}
  Suppose that $\LCE$ queries in a text $T[1\dd n]$ and in its reverse
  $\bar{T}$ can be answered in $\Oh(\log n)$ time.  Given a multiset
  $\Pts$ of $n^{\Oh(1)}$ triples $(x,y,w)$, where $x$ and $y$ are
  fragments of $T$ and $w$ is an integer, we can in $\Oh(|\Pts| \log^2
  n)$ time construct a data structure that, given a pair $(x',y')$ of
  fragments of $T$ answers the following queries regarding the
  multiset $\Ws(x',y'):=\{w : (x,y,w)\in \Pts\text{, }x'\text{ is a
  suffix of }x\text{, and }y'\text{ is a prefix of }y\}$: enumerate
  $\Ws(x',y')$ in $\Oh(\log^2 n + |\Ws(x',y')|)$ time, compute
  $\min\Ws(x',y')$ in $\Oh(\log^2 n)$ time, and compute $\sum
  \Ws(x',y')$ in $\Oh(\log^2 n)$ time.
\end{lemma}
\begin{proof}
  Let $\Xs = \{x : (x,y,w)\in \Pts\}$ and $\Ys = \{y : (x,y,w)\in
  \Pts\}$.  At construction time, the algorithm orders $\Ys$
  lexicographically and $\Xs$ according to the lexicographic order of
  the reversed strings $\bar{x}$ for $x\in \Xs$.  As a result, each
  $y\in \Ys$ and each $x\in \Xs$ is associated to its rank
  $\rank_{\Ys}(y)$ and $\rank_{\Xs}(x)$, respectively.  Using $\LCE$
  queries in $T$ and in $\bar{T}$, this step can be implemented in
  $\Oh(|\Pts| \log^2 n)$ time.  Next, a range searching data structure
  of Chazelle~\cite{chazelle} is constructed with a point
  $(\rank_{\Xs}(x),\rank_{\Ys}(y))$ of weight $w$ associated to each
  $(x,y,w)\in \Pts$.  This data structure costs $\Oh(|\Pts| \log
  |\Pts| )$ time to build, and it answers range reporting queries in
  $\Oh(\log |\Pts| +k)$ time (where $k$ is the number of reported
  points) and semigroup range searching queries in $\Oh(\log^2 |\Pts|
  )$ time. In particular, if the semigroup operation is $+$ or $\min$,
  we obtain (weighted) range counting and range minimum queries.

  Given a query $(x',y')$, the algorithm first finds all $x\in \Xs$
  with suffix $x'$ and all $y\in \Ys$ with prefix $y'$.  Such
  fragments appear consecutively in $\Xs$ and $\Ys$, respectively (due
  to the way these multisets are ordered), and the underlying ranges
  $\Xs[\ell_\Xs\dd r_\Xs]$ and $\Ys[\ell_\Ys\dd r_\Ys]$ can be found
  in $\Oh(\log^2 n)$ using binary search (with $\LCE$ queries in
  $\bar{T}$ and $T$ applied in the comparisons).  The multiset
  $\Ws(x',y')$ is precisely the multiset of weights of points in the
  range $[\ell_\Xs\dd r_\Xs]\times[\ell_\Ys\dd r_\Ys]$ in the data
  structure of~\cite{chazelle}.  Hence, the algorithm makes a range
  query to this data structure with this range.  Depending on whether
  the task is to enumerate $\Ws(x',y')$, compute $\min\Ws(x',y')$, or
  compute $\sum \Ws(x',y')$, this is a range reporting query, a range
  minimum query, or a range counting query, respectively.  In the
  latter two cases, the answer is forwarded to the output, while in
  the first case, the algorithm transforms the multiset of (weighted)
  points to a multiset of weights.  The overall query time is
  $\Oh(\log ^2 n)$ plus $\Oh(|\Ws(x',y')|)$ in case of enumeration
  queries.
\end{proof}

Next, we describe indexes for reporting all the occurrences and
finding the leftmost occurrence.

\begin{theorem}\label{thm:index-reporting}
  Given an LZ77-like parsing of a string $\T[1 \dd n]$ into $f$
  phrases, we can in $\Oh(f \log^4 n)$ time construct a data structure
  that, for any pattern $P$ represented by its arbitrary occurrence in
  $T$, reports all the occurrences of $P$ in $T$ in $\Oh(\log^3 n +
  |\Occ(P,T)|\log n)$ time.
\end{theorem}
\begin{proof}
  Our data structure consists of a recompression RLSLP $\G$ generating
  $T$ (constructed using \cref{thm:recompression}), and a component
  for $\LCE$ queries in $T$ and in its reverse $\bar{T}$ (built using
  \cref{th:index-lce}).  Moreover, the data structure of
  \cref{lem:rlslp} is constructed on top of $\G$, and the data
  structure of \cref{lem:2drange} is built with $\Pts =
  \{(\exp(B),\exp(C),A) : A\in \N\text{ and }{A\to
    BC}\}\cup\{(\exp(B),\exp(B)^{k-1},A):A\in \N\text{ and }A\to
  B^k\}$.  Here, $A$ on the third coordinate is technically
  represented by an integer identifier of $A$.  Moreover, $\exp(B)$ on
  the first coordinate and $\exp(C)$ or $\exp(B)^{k-1}$ on the second
  coordinate are represented as fragments of $T$.  Such fragments can
  be retrieved based on $\nu=\first(A)$: if $\exp(\nu)=T[\ell \dd r)$,
  then $T[\ell\dd \ell+|\exp(B)|)$ matches $\exp(B)$, and
  $T[\ell+|\exp(B)|\dd r)$ matches $\exp(C)$ or $\exp(B)^{k-1}$.
  Since $|\N|=\Oh(f\log^2 n)$, the overall construction time is
  $\Oh(f\log^4 n)$.

  If the pattern $P$ consists of a single letter $A\in \Sigma$, then
  the query algorithm simply retrieves $\nodes(A)$ and reports
  $\exp(\nu)$ as an occurrence of $\nu$ for each $\nu\in \nodes(A)$.
  Since the parse tree $\Tr$ is of height $\Oh(\log n)$, this takes
  $\Oh(|\Occ(P,T)|\log n)$ time.

  In the following, we assume that $|P|\ge 2$.  The algorithm first
  uses \cref{lem:anch} to construct a set $\anch(P)$ of $\Oh(\log n)$
  potential anchors.  For each potential anchor $a\in \anch(P)$, the
  occurrences with anchor $a$ are reported independently. Due to
  $|P|\ge 2$, we may assume that $\anch(P)\sub [1\dd |P|)$, so
  $P_L=P[1\dd a]$ and $P_R=P[a+1\dd |P|]$ are both non-empty.
  Moreover, both $P_L$ and $P_R$ are represented as fragments of $T$.
  The algorithm makes a reporting query with $(x',y')=(P_L,P_R)$ to
  the data structure of \cref{lem:2drange}.  Due to the
  characterization of \cref{obs:occ}, this results in a set of
  non-terminals $\mathcal{A}\sub \N$ such that $\Occ(P,a,\nu)\ne
  \emptyset$ if and only if $s(\nu)\in \mathcal{A}$.  Next, for every
  $A\in \mathcal{A}$, the algorithm uses \cref{lem:rlslp} to enumerate
  $\nodes(A)$, and for each $\nu \in \nodes(A)$, it appends
  $\Occ(P,a,\nu)$ to the constructed set $\Occ(P,T)$.  Since
  $\Occ(P,a,\nu)\ne \emptyset$ and since no occurrence is reported
  multiple times, the overall running time is $\Oh(\log^3 n)$
  (dominated by $\Oh(\log n)$ queries to the data structure of
  \cref{lem:2drange}) plus $\Oh(|\Occ(P,T)|\log n)$ (dominated by
  enumerating $\nodes(A)$ for each $A\in \mathcal{A}$).
\end{proof}

\begin{theorem}\label{th:index-leftmost-occ}
  Given an LZ77-like parsing of a string $\T[1 \dd n]$ into $f$
  phrases, we can in $\bigO(f \log^4 n)$ time construct a data
  structure that, for any pattern $P$ represented by its arbitrary
  occurrence in $T$, returns the leftmost occurrence of $P$ in $T$ in
  $\Oh(\log^3 n)$ time.
\end{theorem}
\begin{proof}
  The index is the same as in the proof of \cref{thm:index-reporting},
  except for the weights $w$ in triples $(x,y,w)\in\Pts$: for each
  $A\in \N$ with $A\to BC$ or $A\to B^k$, the corresponding weight is
  set to $\ell+|\exp(B)|$, where $T[\ell\dd r)=\exp(\first(A))$.
  Given that \cref{lem:rlslp} provides $\first(A)$ for each $A\in \S$,
  the set $\Pts$ can still be constructed in $\Oh(f\log^2 n)$ time,
  and the whole index can be constructed in $\Oh(f\log^4 n)$ time.

  If the pattern $P$ consists of a single letter $A\in \Sigma$, then
  the query algorithm simply returns $\exp(\first(A))$; this costs
  $\Oh(1)$ time.  In the following, we assume that $|P|\ge 2$.  The
  algorithm first uses \cref{lem:anch} to construct a set $\anch(P)$
  of $\Oh(\log n)$ potential anchors.  For each potential anchor $a\in
  \anch(P)$, the algorithm makes a minimum query with $(x',y')=(P[1\dd
  a],P[a+1\dd |P|))$ to the data structure of \cref{lem:2drange}.
  By \cref{obs:occ}, for every $A\in \N$ and $\nu\in\nodes(A)$, we
  have that $\Occ(P,a,\nu)\ne \emptyset$ if and only if $x'$ is a
  suffix of $x$ and $y'$ is a prefix of $y$ for the triple $(x,y,w)\in
  \Pts$ corresponding to $A$.  Moreover, $w-a$ is then the starting
  position of the leftmost occurrence of $P$ with anchor $a$ and hook
  $\nu\in\nodes(A)$.  Consequently, if the query to the data structure
  of \cref{lem:2drange} yields a value $p$, then the leftmost
  occurrence of $P$ with anchor $a$ starts at position $p-a$.  (If
  $p=\infty$, then there are no occurrences with anchor $a$.)
  Therefore, the algorithm reports the minimum among the values $p-a$
  obtained for various potential anchors $a\in \anch(P)$ as the
  starting position of the leftmost occurrence of $P$ in $T$.  The
  total query time is $\Oh(\log^3 n)$, dominated by
  $|\anch(P)|=\Oh(\log n)$ queries to the data structure of
  \cref{lem:2drange}.
\end{proof}

Similarly, as in \cref{th:index-lce}, since a recompression RLSLP
generating the reversed string $\bar{T}$ can be obtained by reversing
$\rhs(A)$ for each non-terminal $A$, the above construction yields
also an index for finding the rightmost occurrences.
\begin{theorem}\label{th:index-rightmost-occ}
  Given an LZ77-like parsing of a string $\T[1 \dd n]$ into $f$
  phrases, we can in $\bigO(f \log^4 n)$ time construct a data
  structure that, for any pattern $P$ represented by its arbitrary
  occurrence in $T$, returns the rightmost occurrence of $P$ in $T$ in
  $\Oh(\log^3 n)$ time.
\end{theorem}

Counting the occurrences is more challenging, because the size
$|\Occ(P,a,\nu)|$ may vary depending on $|P|$: if $s(\nu)=A$, $A\to
B^k$, $P[1\dd a]$ is a suffix of $\exp(B)$, and $P[a+1\dd |P|)$ is a
prefix of $\exp(B)^{k-1}$, then
$|\Occ(P,a,\nu)|=k-\big\lceil{\frac{|P|-a}{|\exp(B)|}}\big\rceil$
according to \cref{obs:occ}.  In other words, we have one occurrence
for each exponent $k'\in [1\dd k)$ such that $P[a+1\dd |P|)$ is a
prefix of $\exp(B)^{k'}$.  A simple solution would be to add to $\Pts$
a triple $(\exp(B),\exp(B)^{k'},\cnt(A))$ for each $A\to B^k$ and each
$k'\in [1\dd k)$. However, this may result in $|\Pts|=\Theta(n)$ even
if $f=\Oh(1)$.  Hence, we apply a more sophisticated solution based
on~\cite[Section 7]{Christiansen2019}.  The main idea is to classify
the fragments of $T$ (and thus also the occurrences of $P$ in $T$)
into \emph{regular} and \emph{special}.

\begin{definition}
  Consider an RLSLP $\G$ representing $T$.  We call a fragment $x$ of
  $T$ \emph{regular} if $|x|=1$ or $x$ overlaps $\exp(\nu')$ for at
  most three children $\nu'$ of $\hook(x)$, and \emph{special}
  otherwise.
\end{definition}

We design separate data structures for counting the regular and
special occurrences of $P$ in $T$.  For this, we partition $\Occ(P,T)$
into $\Occ^r(P,T)$ and $\Occ^s(P,T)$, which contain the starting
positions of regular and special occurrences of $P$ in $T$,
respectively.  Similarly, we partition $\Occ(P,a,\nu)$ for any
potential anchor $a$ and any node $\nu$ in $\Tr$ into
$\Occ^r(P,a,\nu)$ and $\Occ^s(P,a,\nu)$.

\paragraph*{Counting Regular Occurrences.}
Regular occurrences are counted using an index similar to that in
\cref{thm:index-reporting,th:index-leftmost-occ}, based on the
following variant of \cref{obs:occ}.

\begin{observation}\label{obs:rocc}
  Let $\nu$ be node of $\Tr$ with $s(\nu)=A$, and let $P=P_L\cdot P_R$
  be a non-empty pattern.  Then $\Occ^r(P,\nu,|P_L|)\ne \emptyset$
  only if one of the following holds:
  \begin{enumerate}
  \item $A\in \Sigma$, $P_L=\eps$, and $P_R=A$. Then
    $|\Occ^r(P,\nu,|P_L|)|=1$.
  \item $A\to BC$, $P_L \ne \eps$ is a suffix of $\exp(B)$, and
    $P_R\ne \eps $ is a prefix of $\exp(C)$.  Then
    $|\Occ^r(P,\nu,|P_L|)|=1$.
  \item $A\to B^k$, $P_L \ne \eps$ is a suffix of $\exp(B)$, and $P_R
    \ne \eps$ is a prefix of $\exp(B)^2$.  Then
    $|\Occ^r(P,\nu,|P_L|)|=k-1$ if $P_R$ is a prefix of $\exp(B)$ and
    $|\Occ^r(P,\nu,|P_L|)|=k-2$ otherwise.
  \end{enumerate}
\end{observation}

\begin{proposition}\label{prp:regular}
  Given an $\Oh(\log n)$-round recompression RLSLP $\G$ generating a
  string $T[1\dd n]$, we can in $\Oh(|\S|\log^2 n)$ time construct a
  data structure that, for any pattern $P$ represented by its
  arbitrary occurrence in $T$, returns the number $|\Occ^r(P,T)|$ of
  regular occurrences of $P$ in $T$ in $\Oh(\log^3 n)$ time.
\end{proposition}
\begin{proof}
  Recall that $\G$ allows answering $\LCE$ queries in $T$ and its
  reverse $\bar{T}$ in $\Oh(\log n)$ time~\cite{tomohiro-lce}.  Our
  index applies \cref{lem:rlslp} on top of $\G$, and builds the data
  structure of \cref{lem:2drange} with
  \begin{align*}
    \Pts =&\; \{(\exp(B),\exp(C),\cnt(A)) : A\in \N\text{ and }{A\to
      BC}\}\\ \cup&\;\{(\exp(B),\exp(B),\cnt(A)):A\in \N\text{ and
    }A\to B^k\} \\ \cup&\;\{(\exp(B),\exp(B)^{2},(k-2)\cnt(A)):A\in
    \N\text{ and }A\to B^k\}.
  \end{align*}
  As in the proof of \cref{thm:recompression}, the fragments on the
  first two coordinates can be retrieved based on $\exp(\first(A))$
  for each $A\in \N$.  The overall construction time is
  $\Oh(|\S|+|\N|\log^2 n)=\Oh(|\S|\log^2 n)$.
  
  If the pattern $P$ consists of a single letter $A\in \Sigma$, then
  the query algorithm simply returns $\cnt(A)$, which takes $\Oh(1)$
  time.  In the following, we assume that $|P|\ge 2$.  The algorithm
  first uses \cref{lem:anch} to construct a set $\anch(P)$ of
  $\Oh(\log n)$ potential anchors.  For each potential anchor $a\in
  \anch(P)$, the algorithm makes a counting query with
  $(x',y')=(P[1\dd a],P[a+1\dd |P|))$ to the data structure of
  \cref{lem:2drange}.  Due to the characterization of \cref{obs:rocc},
  this results in the total number of regular occurrences of $P$ with
  anchor $a$.  Finally, the algorithm sums up the results obtained for
  all potential anchors $a\in \anch(P)$.  The overall running time is
  $\Oh(\log^3 n)$ (dominated by $|\anch(P)|=\Oh(\log n)$ queries to
  the data structure of \cref{lem:2drange}).
\end{proof}

\paragraph*{Counting Special Occurrences.}

First, we need to take a detour and prove two properties of
recompression RLSLPs.  For this, let us extend the expansion function
$\exp$ so that $\exp(X)=\exp(X[0])\cdots \exp(X[|X|-1])$ for $X\in
\S^*$.

\begin{lemma}\label{lem:uniq}
  If $\exp(A)=\exp(A')$ for $A,A\in \S$ in a recompression RLSLP $\G$,
  then $A=A'$.
\end{lemma}
\begin{proof}
  Recall the strings $T_1,\ldots,T_h$ in the definition of $\G$.  We
  shall inductively prove the following claim: if $\exp(x)=\exp(x')$
  for two fragments of $T_j$, then $x=x'$.  The claim holds trivially
  for $j=1$ due to $x=\exp(x)=\exp(x')=x'$.  Thus, consider $j>1$ and
  suppose that the claim holds for $j-1$.  Let $\bar{x}$ and
  $\bar{x}'$ be the fragments of $T_{j-1}$ corresponding to $x$ and
  $x'$ in $T_j$, respectively.  Due to
  $\exp(\bar{x})=\exp(x)=\exp(x')=\exp(\bar{x}')$, the inductive
  assumption yields $\bar{x}=\bar{x}'$.  Notice that block boundaries
  between two symbols of $T_{j-1}$ are placed solely based on the
  values of these symbols (a boundary is \emph{not} placed between
  matching symbols in run-length encoding and between a left symbol
  and a right symbol in alphabet partitioning).  Hence, block
  boundaries within $\bar{x}$ and $\bar{x}'$ are placed in the same
  way. Moreover, block boundaries are placed at the endpoints of
  $\bar{x}$ and $\bar{x}'$ (since they are collapsed to $x$ and $x'$,
  respectively).  Thus, the partition of $T_{j-1}$ into blocks
  partitions $\bar{x}$ and $\bar{x}'$ into full blocks in the same
  way.  As matching blocks are replaced by the same non-terminal, we
  conclude that $x=x'$.  This completes the inductive proof.

  Without loss of generality suppose that the recompression algorithm
  creates $A$ before $A'$, and suppose that $A$ appears for the first
  time in $T_j$ (so that neither $A$ nor $A'$ appears in $T_{j'}$ for
  $j'<j$).  Observe that, for each node $\nu \in \nodes(A)\cup
  \nodes(A')$, the fragment $\exp(\nu)$ is collapsed to a fragment of
  $T_j$ (which gets collapsed to $A$ or $A'$ at some iteration $j'\ge
  j$).  As $A$ appears in $T_j$, one of these fragments of $T_j$
  consists of a single symbol $A$, and by the claim above, this means
  that all the fragments consist of a single symbol $A$.  We conclude
  that $A=A'$ because blocks of size $1$ are never collapsed to new
  non-terminals
\end{proof}

\begin{lemma}\label{lem:primitive}
  If $A\to B^k$ in a recompression RLSLP $\G$, then
  $\per(\exp(B)^2)=|\exp(B)|$.
\end{lemma}
\begin{proof}
  Recall the strings $T_1,\ldots,T_h$ in the definition of $\G$.  We
  shall inductively prove the following claim: if $\per(x)=\frac12|x|$
  for a fragment $x$ of $T_j$, then $\per(\exp(x))=\frac12|\exp(x)|$.
  The claim holds trivially for $j=1$ due to $x=\exp(x)$.  Thus,
  consider $j>1$ and suppose that the claim holds for $j-1$.  Let
  $\bar{x}$ be the fragment of $T_{j-1}$ corresponding to $x$ in
  $T_j$.  Note that $\frac12|\bar{x}|$ is a period of $\bar{x}$.  By
  periodicity lemma, we therefore have
  $\per(\bar{x})=\frac{1}{2k}|\bar{x}|$ for some integer $k\ge
  1$. Consequently, $\bar{x}$ can be decomposed into $2k$ matching
  fragments $\bar{y}_1,\ldots,\bar{y}_{2k}$.  Now, consider the block
  boundaries that the recompression algorithm places within $T_{j-1}$.
  By definition of $\bar{x}$, there are block boundaries at the
  endpoints of $\bar{x}$, that is, before $\bar{y}_1$ and after
  $\bar{y}_{2k}$.  Moreover, since $\frac12|x|$ is a period of $x$,
  there is a block boundary in the midpoint of $\bar{x}$, that is,
  between $\bar{y}_k$ and $\bar{y}_{k+1}$.  However, since the block
  boundaries between two symbols of $T_{j-1}$ are placed solely based
  on the values of these two symbols, we conclude that block
  boundaries are placed between $\bar{y}_i$ and $\bar{y}_{i+1}$ for
  each $i\in [1\dd 2k)$.  Thus, the partition of $T_{j-1}$ into blocks
  partitions each fragment $\bar{y}_i$ (for $i\in [1\dd 2k]$) into
  full blocks the same way.  As matching blocks are replaced by the
  same non-terminal, we conclude that $x$ can be partitioned into $2k$
  matching fragments $y_1,\ldots,y_{2k}$, that is, $\frac{1}{2k}|x|$
  is a period of $x$.  Hence, $k=1$, and thus
  $\per(\bar{x})=\frac12|\bar{x}|$.  The inductive assumption yields
  $\per(\exp(x))=\per(\exp(\bar{x}))=\frac12|\exp(\bar{x})|=\frac12|\exp(x)|$,
  which completes the inductive proof.

  Suppose that $A$ appears for the first time in $T_j$ (so that it
  does not appear in $T_{j'}$ for $j'<j$).  Let us fix any occurrence
  of $A$ in $T_j$. The corresponding fragment of $T_{j-1}$ matches
  $B^k$.  Applying the claim to its prefix matching $B^2$, we conclude
  that $\per(\exp(B)^2)=|\exp(B)|$.
\end{proof}

We are now ready to characterize special occurrences just like regular
occurrences were characterized in \cref{obs:rocc}.

\begin{lemma}\label{lem:socc}
  Let $P$ be a pattern, $\nu$ be a node of $\Tr$, and $a$ be an
  integer.  Then $\Occ^s(P,\nu,a)\ne \emptyset$ if and only if
  $A=s(\nu)$ is of the form $A\to B^k$, $a\in [1\dd |P|)$,
  $\max\big(a,\frac{|P|-a}{k-1}\big)\le \per(P)< \frac{|P|-a}{2}$,
  and $\exp(B)=P[a+1\dd a+\per(P)]$. In this case
  $|\Occ^s(P,a,\nu)|=k-\big\lceil{\frac{|P|-a}{\per(P)}}\big\rceil$.
\end{lemma}
\begin{proof}
  First, consider a special occurrence $x$ of $P$ with anchor $a$ and
  hook $\nu$.  Since $x$ overlaps $\exp(\nu')$ for at least 4 children
  $\nu'$ of $\nu$, the symbol $A=s(\nu)$ must be of the form $A\to
  B^k$ for $k\ge 4$.  \cref{obs:occ} implies that $P_L:=P[1\dd a]$ is
  a non-empty suffix of $\exp(B)$ and $P_R:=P[a+1\dd |P|)$ is a
  non-empty prefix of $\exp(B)^{k-1}$.  In particular, $a\in [1\dd
  |P|)$ and $|\exp(B)|$ is a period of $P$.  Moreover, since $x$
  overlaps $\exp(\nu')$ for at least 4 children $\nu'$ of $\nu$, the
  string $\exp(B)^2$ must be a proper prefix of $P_R$, and hence a
  substring of $P$.  Due to \cref{lem:primitive}, we have
  $\per(\exp(B)^2)=|\exp(B)|$, and thus
  $\per(P)=|\exp(B)|$. Consequently, $\exp(B)=P_R[1\dd
  \per(P)]=P[a+1\dd a+\per(P)]$.  Finally, $\per(P)\ge a$ since
  $P_L$ is a suffix of $\exp(B)$, $\per(P)\ge \frac{|P|-a}{k-1}$ since
  $P_R$ is a prefix of $\exp(B)^{k-1}$, and $\per(P)< \frac{|P|-a}{2}$
  since $\exp(B)^2$ is a proper prefix of $P_R$.  Therefore, all the
  conditions in the lemma statement are satisfied.

  It remains to prove the converse implication: if $P$, $\nu$, and $a$
  satisfy the conditions in the lemma statement, then
  $|\Occ^s(P,\nu,a)|=k-\big\lceil{\frac{|P|-a}{\per(P)}}\big\rceil>0$.
  Again, let $P_L := P[1\dd a]$ and $P_R:=P[a+1\dd |P|)$; these
  fragments are non-empty due to $a\in[1\dd |P|)$.  As
  $\exp(B)=P_R[1\dd \per(P)]$, we conclude that $P_L$ is a suffix and
  $P_R$ is a prefix of a sufficiently large power $\exp(B)^{k'}$.  Due
  to $|\exp(B)|\ge \max(|P_L|,\frac{|P_R|}{k-1})$, in fact, $P_L$ is a
  suffix of $\exp(B)$ and $P_R$ is a prefix of $\exp(B)^{k-1}$.
  Consequently, \cref{obs:occ} yields $\Occ(P,\nu,
  |P_L|)=\big\{\ell+i|\exp(B)|-|P_L| : i \in \big[1\dd
  k-\big\lceil{\frac{|P_R|}{|\exp(B)|}\big\rceil}\big]\big\}$, where
  $\exp(\nu)=T[\ell \dd r)$.  Since $|\exp(B)|<\frac12|P_R|$, the
  occurrence of $P$ at position $\ell+i|\exp(B)|-|P_L|$ is special as
  it overlaps the fragments $\exp(\nu')$ for at least four children of
  $\nu$: these are $T[\ell+j|\exp(B)|\dd \ell+(j+1)|\exp(B))$ for
  $j\in \{i-1,i,i+1,i+2\}$.  Thus,
  $|\Occ^s(P,\nu,a)|=|\Occ(P,\nu,a)|=k-\big\lceil{\frac{|P_R|}{|\exp(B)|}\big\rceil}=
  k-\big\lceil{\frac{|P|-a}{\per(P)}\big\rceil}>0$.
\end{proof}

Next, we extend \cref{lem:rlslp} so that generalized $\cnt(B)$ queries
can be answered for $B\in \S$.

\begin{lemma}\label{lem:rlslps}
  Given an RLSLP $\G$, we can in $\Oh(|\S|\log |\S|)$ time construct a
  data structure that, given a symbol $B$ and a positive integer $m$,
  computes the following value in $\Oh(\log |\S|)$ time:
  \[
    \cnt(B,m) := \sum_{A\in \S\::\:\rhs(A)=B^{k}\text{ \rm for }k> m}
    (k-m)\cdot \cnt(A).
  \]
\end{lemma}
\begin{proof}
  For each symbol $B\in \S$, let us define a set $K(B)$ containing $1$
  and all the exponents $k$ such that $\rhs(A)=B^k$ for some $A\in
  \S$.  The construction algorithm builds sets $K(B)$ represented as
  sorted arrays.  Each element $m\in K(B)$ is augmented with two
  values
  \[
    \sum_{A\in \S\;:\;\rhs(A)=B^{k}\text{ \rm for }k> m} k\cdot
    \cnt(A)\quad\text{and}\quad \sum_{A\in \S\;:\;\rhs(A)=B^{k}\text{
        \rm for }k> m} \cnt(A),
  \]
  computed using a right-to-left scan of $K(B)$, after the counts
  $\cnt(A)$ are obtained using \cref{lem:rlslp}.  The overall
  construction time is $\Oh(|\S|\log |\S|)$, dominated by sorting
  $K(B)$.

  The query algorithm, given a symbol $B$ and an integer $m\ge 1$,
  finds the predecessor $m'\in K(B)$ of $m$ (using binary search) and
  computes $\cnt(B,m)$ based on the values stored with $m'$:
  \[
    \cnt(B,m) = \sum_{A\in \S\;:\;\rhs(A)=B^{k}\text{ \rm for }k> m'}
    k\cdot \cnt(A) \quad - \quad m\cdot \sum_{A\in
      \S\;:\;\rhs(A)=B^{k}\text{ \rm for }k> m'} \cnt(A).
  \]
  This equality holds because, by definition of $m'$ as the
  predecessor of $m$ in $K(B)$, if $\rhs(A)=B^k$, then $k>m$ is
  equivalent to $k>m'$.  The query time is $\Oh(\log |\S|)$ dominated
  by finding the predecessor.
\end{proof}

Finally, we describe our index for counting special occurrences.

\begin{proposition}\label{prp:special}
  Given an $\Oh(\log n)$-round recompression RLSLP $\G$ generating a
  string $T[1\dd n]$, we can in $\Oh(|\S|\log^2 n)$ time construct a
  data structure that, for any pattern $P$ represented by its
  arbitrary occurrence in $T$, returns the number $|\Occ^s(P,T)|$ of
  special occurrences of $P$ in $T$ in $\Oh(\log^3 n)$ time.
\end{proposition}
\begin{proof}
  Recall that $\G$ allows answering $\LCE$ queries in $T$ in $\Oh(\log
  n)$ time~\cite{tomohiro-lce}.  The construction algorithm uses these
  queries, combined with \cref{lem:rlslp}, which yields $\first(A)$
  for each $A\in \S$, to sort the symbols $A\in \S$ according to the
  lexicographic order of their expansions $\exp(A)$.  Additionally,
  the index contains a component for 2-Period queries
  (\cref{thm:period}) and a component for computing the values
  $\cnt(B,m)$ (\cref{lem:rlslps}).  The overall construction time is
  $\Oh(|\S|\log^2 n)$, dominated by sorting $\S$, which involves
  $\Oh(|\S|\log |\S|)$ comparisons in $\Oh(\log n)$ time each.

  The query algorithm first makes a 2-period query on the pattern $P$.
  If $\per(P)>\frac12|P|$, then the algorithm returns $0$, which is
  correct due to the condition $\per(P)<\frac{|P|-a}{2}<\frac12|P|$ in
  \cref{lem:socc}.  Otherwise, the algorithm uses \cref{lem:anch} to
  construct a set $\anch(P)$ of $\Oh(\log n)$ potential anchors.  For
  each potential anchor $a\in \anch(P)\cap [1\dd |P|)$, the algorithm
  checks whether $a \le \per(P)<\frac{|P|-a}{2}$.  If this test fails,
  then $P$ has no special occurrences according to \cref{lem:socc}.
  If the test succeeds, then the algorithm finds a symbol $B\in \S$
  such that $\exp(B)=P[a+1\dd a+\per(P)]$.  For this, the algorithm
  performs a binary search on the list of symbols (sorted by
  expansions) with an $\LCE$ query applied to implement each
  comparison.  If there is no such symbol $B$, then $P$ has no special
  occurrences with anchor $a$ according to \cref{lem:socc}.  If a
  symbol $B$ exists, it is unique by \cref{lem:uniq}.  In this case,
  the algorithm adds
  $\cnt\big(B,\big\lceil{\frac{|P|-a}{\per(P)}\big\rceil}\big)$
  (obtained using \cref{lem:rlslps}) to the count of special
  occurrences of $P$.  This term represents the number of special
  occurrences with anchor $a$, because each node $\nu$ with
  $\rhs(s(\nu))=B^k$ has a non-empty set $\Occ^s(P,\nu,a)$ only if
  $\frac{|P|-a}{k-1}\le \per(P)$, which is equivalent to $k-1 \ge
  \frac{|P|-a}{\per(P)}$ and to $k >
  \big\lceil{\frac{|P|-a}{\per(P)}\big\rceil}$.  Moreover, the size of
  $\Occ^s(P,\nu,a)$ is then equal to
  $k-\big\lceil{\frac{|P|-a}{\per(P)}\big\rceil}$ by \cref{lem:socc}.
  Consequently, the algorithm correctly counts special occurrences of
  $P$ in $T$.  The query time is $\Oh(\log^3 n)$ dominated by
  $\Oh(\log n)$ $\LCE$ queries made for each potential anchor.
\end{proof}

Combining \cref{thm:recompression} (for transforming an LZ77-like
parsing into a recompression RLSLP) with
\cref{prp:regular,prp:special} (for counting regular and special
occurrences, respectively), we obtain our counting index.

\begin{theorem}\label{th:index-count}
  Given an LZ77-like parsing of a string $\T[1 \dd n]$ into $f$
  phrases, we can in $\bigO(f \log^4 n)$ time construct a data
  structure that, for any pattern $P$ represented by its arbitrary
  occurrence in $T$, returns the number of occurrences of $P$ in $T$
  in $\Oh(\log^3 n)$ time.
\end{theorem}

\paragraph{Acknowledgment.}
The authors would like to thank Barna Saha for helpful discussions.

\bibliographystyle{plainurl}
\bibliography{paper}

\end{document}